\documentclass[paper=a4]{article}

\usepackage{courier}
\usepackage{ucs}
\usepackage[utf8x]{inputenc}
\usepackage[T1]{fontenc}
\usepackage[greek,english,ngerman]{babel}
\usepackage{bbold}  
\usepackage{url}
\usepackage{graphicx}
\usepackage{amssymb}
\usepackage{amsmath}
\usepackage{amsthm}
\usepackage{eurosym}
\usepackage{color}

\begin{document}

\selectlanguage{\english}
\newcounter{definition}
\newcommand{\definition}{\refstepcounter{definition} {\noindent \bf Definition \thedefinition: }}
\renewcommand{\thedefinition}{\arabic{definition}}

\newcounter{axiom}
\newcommand{\axiom}{\refstepcounter{axiom} {\vspace{0.5cm} \noindent \bf Axiom \theaxiom: }}
\renewcommand{\theaxiom}{\arabic{axiom}}

\newcounter{proposition}
\newcommand{\proposition}{\refstepcounter{proposition} {\vspace{0.5cm} \noindent \bf Proposition \theproposition: }}
\renewcommand{\theproposition}{\arabic{proposition}}

\newcounter{theorem}
\newcommand{\theorem}{\refstepcounter{theorem} {\vspace{0.5cm} \noindent \bf Theorem \thetheorem: }}
\renewcommand{\thetheorem}{\arabic{theorem}}

\newcounter{lemma}
\newcommand{\lemma}{\refstepcounter{proposition} {\vspace{0.5cm} \noindent \bf Lemma \thelemma: }}
\renewcommand{\thelemma}{\arabic{lemma}}


\title{Composition, Cooperation, and Coordination of Computational Systems}
\author{Johannes Reich, johannes.reich@sophoscape.de}

\maketitle

\begin{abstract}
In this paper I elaborated on the idea of David Harel and Amir Pnueli to think systems and their interaction from the point of view of their compositional behaviour. The basic idea is to assume a functional relation between input, internal and output state functions as the system-constituting property, allowing to partition the world into a system and a rest. As versatile as the system concept is, I introduce several, quite distinct system models.

The obvious idea to base the composition of systems on the concept of computable functions and their compositional behaviour leads to supersystem formation by composing simple and recursive systems. 

But this approach does not allow to account adequately for systems that interact with many other systems in a stateful and nondeterministic way, which is why I introduce the concept of interactivity and cooperation. 
In order to describe interactive systems satisfactorily, a balance is needed between the representation of their relationship to all the other systems and what happens within the systems. I thus introduce the complementary descriptions of external interactions and internal coordination, both based on a role concept in the sense of a projection of a system onto its interactions. It actually mirrors the internal vs. external distinction initially introduced by the system model and reflects the problem how systems are supposed to cooperate without melting into a common supersystem. 

Beside the interesting distinction between composition, cooperation and coordination that the presented approach allows, it also fits well with other rather well known concepts. First the concept of components and interfaces. Components become systems, intended for (a given) composition and with the notion of an interface we subsume all relevant information necessary for this composition.  

Another tightly related concept is that of decisions and games. I introduce the concept of decisions as an additional internal input alphabet in order to determine nondeterministic interactions and thus fictitiously assume a system function where we actually do not have the knowledge to do so. Thus, the close relationship between protocols and games becomes obvious.

As the core of the architecture concept of an IT system is its structure in terms of the compositional relationships of its interacting components, I finally transfer the gained insights to the field of IT system architecture and introduce the concept of the ''interaction oriented architecture (IOA)'' for interactive systems with its three elements of roles, coordination rules, and decisions. 
\end{abstract}

\tableofcontents


%
\section{Introduction}
%

Civil, mechanical, electrical, or software engineering - they all deal with the design of systems. The term ''system engineering'' was coined in the Bell Laboratories in the 1940s (e.g. \cite{Buede2011_Engineering}). According to Kenneth J. Schlager \cite{Schlager1956}, a main driver of the field was the complexity of the composition behaviour of components. Often composing seemingly satisfactorily working components did not result in satisfactorily working composed systems. 

The focus on systems was transferred to other disciplines like biology (e.g. \cite{Bertalanffy1968_GST}) or sociology (e.g. \cite{Parson1951_SocialSystems}). Especially with the advent of cyber physical systems (e.g. \cite{Lee2008_CPS,Alur2015_Principles}), also the software engineering discipline articulates the importance of a unifying view (e.g. \cite{Sifakis2011_Vision}).

To substantiate our talk about a ''unifying view'' we must provide a system notion where all disciplines agree upon. I propose to ground it on the notion of a system's functionality, a notion that is tightly related to the concepts of computability and determinism of computer science. A system in the sense of this article is determined by its states and the functional relation between its state values in time. Thereby the system function becomes system-constitutive and separates the inner parts from the outside, the rest of the world. Depending on our model of states, time and function, we get different system models. In this article, I will focus mainly on discrete systems with a computable system function. 

This isolation of the inner parts from the outside also implies system interactions, where systems share their externally accessible states such that the output state of one system is the input state of another system. In this article, I approach the description of these interactions with the composition concept. Due to their interactions, systems compose to larger structures. We can thereby classify system interactions by the different classes of composition that occur. I will investigate intensively two different composition classes: hierarchical composition of entire systems to supersystems and the composition of a special class of system parts, which I will name ''roles'', to protocols, creating two very different system relations: hierarchical composition (see Fig. \ref{fig_system_composition_general}) and cooperation (see Fig. \ref{fig_network_of_relations}). 

\begin{figure}[ht]
\begin{center}
\includegraphics[width=7cm]{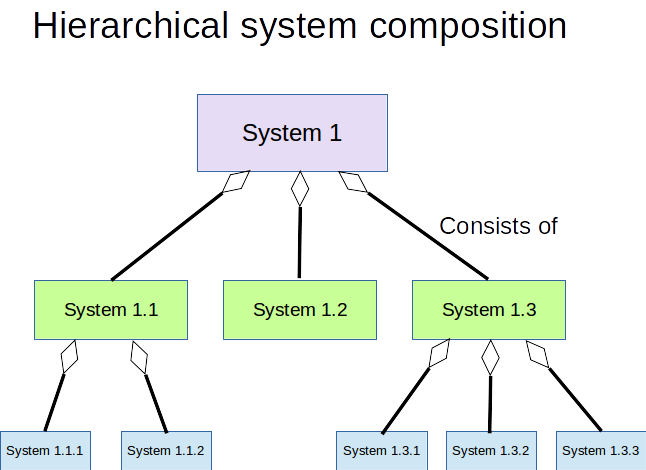}
\end{center}
\caption[]
{\label{fig_system_composition_general}An example of a hierarchically composed system, namely system 1. It consists of 3 subsystems which by themselves consist of a couple of further sub-subsystems.}
\end{figure}

Hierarchical composition means that, by the interaction of systems, super systems are created and the interacting systems become subsystems. As a trivial consequence, there is no interaction between a supersystem and its subsystems.
As the system notion is based on the function notion, the concept of system composition can be traced back to the concept of composition of functions which is in our considered discrete case at the heart of computability. Especially for software components, the borrowed distinction between non-recursive and recursive system composition becomes important. As is illustrated in Fig. \ref{fig_system_composition_general}, system composition leads to hierarchies of systems which are determined by the 'consists of'-relation between supersystems and their subsystems.

\begin{figure}[ht]
\begin{center}
\includegraphics[width=7cm]{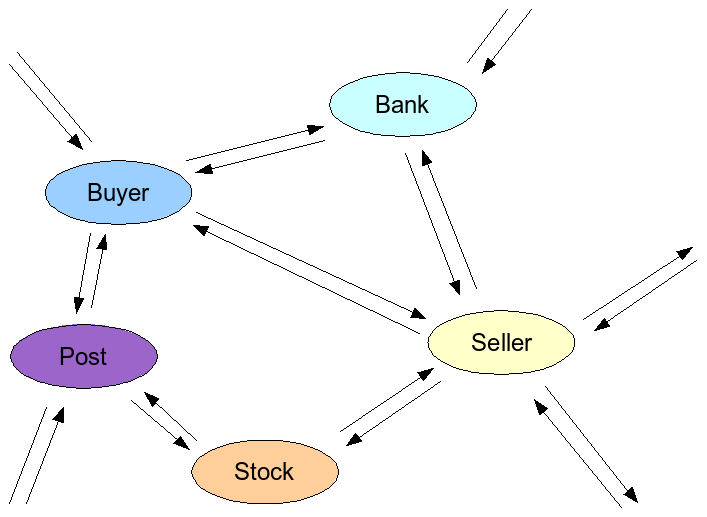}
\end{center}
\caption[]{\label{fig_network_of_relations}A cutout of an open business network.}
\end{figure}

Quite contrasting, system cooperation means that systems interact sensibly somehow ''loosely'' without such supersystem creation. The basic reason to evade a construction of supersystems lies in a nondeterministic interaction. This is typically the case in interaction networks where all the actors interact in the same manner. As an example, Fig. \ref{fig_network_of_relations} sketches a cutout of a network of business relations between a buyer, a seller, its stock, a post and a bank. In these networks none of the participants is in total control of all the other participants: the participants' interactions don't, in general, determine the participants' actions. These networks are usually open in the sense that we cannot describe them completely. 

\begin{figure}[ht]
\begin{center}
\includegraphics[width=5cm]{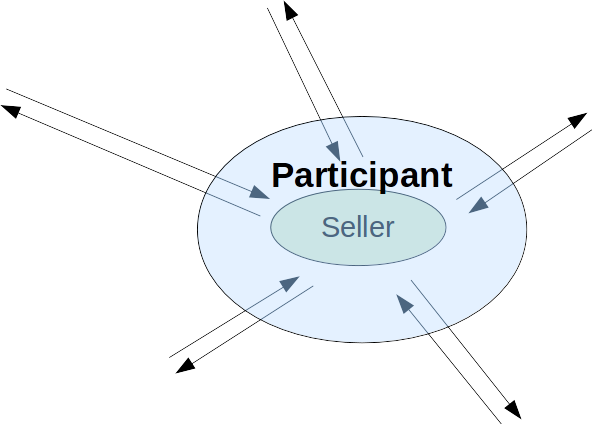}\hspace{2cm}
\includegraphics[width=5cm]{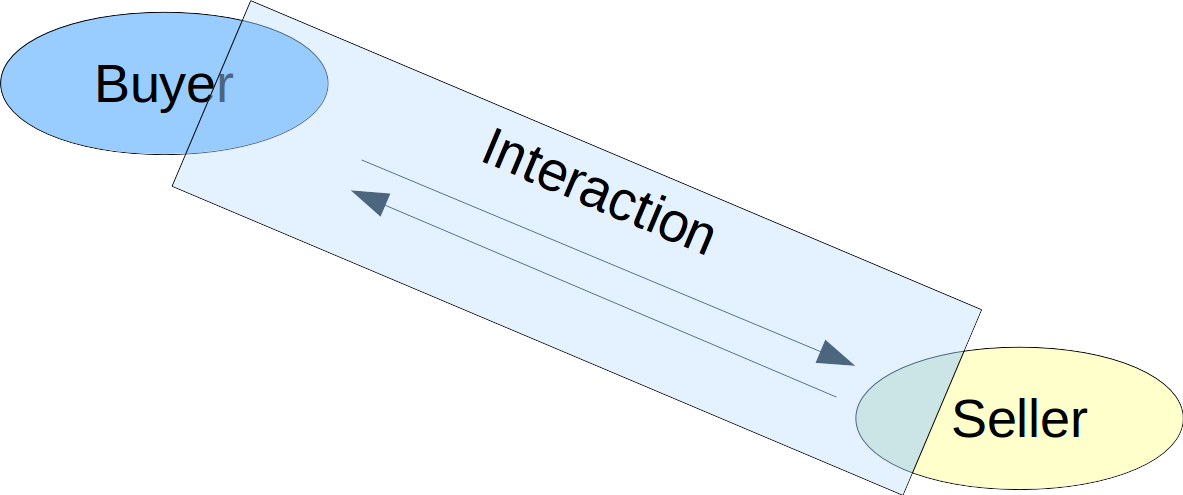}
\end{center}
\caption[]
{\label{fig_two_perspectives}The two perspectives describing a network of interactions: one can focus on the participant or the interaction. }
\end{figure}

I think, to describe these interaction networks satisfactorily, the concepts of functionality and  computability are not sufficient. Instead they must be complemented by new tools and a new language. In fact, the search for an appropriate approach is in full swing, as can be seen by the variety of different attempts in the literature, some of which I discuss in section \ref{s_related_work}. 
I base my approach on two perspectives which I delineate in Fig. \ref{fig_two_perspectives}. Focusing too much on the nodes implies putting the system function to the fore and neglecting the interaction, the edges of the network. The central idea is that an adequate model must describe both, the nodes, i.e. the participants, and the edges, i.e. the interactions of such an interaction network in a mutually compatible way. This leads to the model of interactive systems interacting through protocols. An interactive system can be thought of as being a system, coordinating its different roles it plays in its different protocol interactions. It seems to me that coordination could become an equally important area of informatics as computability.


The structure of the article is as follows. I first introduce the basic concepts of compositionality and computability. 
In section \ref{s_simple_systems} I investigate the composition of what I call ''simple'' systems. They are simple in the sense that their system function always operate on their total input to create the complete output. 
In section \ref{s_recursive_systems} I extend the system model to operate only on parts of its input to realise feedback-based recursive computations. Both sections together thereby describe the compositionality of systems that realise computational functionality. 

In section \ref{s_interactive_systems} I introduce the notion of interactivity to describe nondeterministically cooperating systems together with their interactions in a mutually compatible way. 

In section \ref{s_partitioning} I present two different equivalence class constructions. One for the deterministic and one for the nondeterministic case, which help to tackle the state explosion problem of more complex problems. 

I present related system models of other authors in section \ref{s_related_work}. Here, I dwell on the reactive systems of David Harel and Amir Pnueli, the system models of Manfred Broy and Rajeev Alur as well as the BIP component framework of Joseph Sifakis et al. 

In the final section \ref{s_it_system_architecture} I transfer the gained insights to the field of IT system architecture and introduce the concept of the ''interaction oriented architecture (IOA)'' for interactive systems with its three elements of roles, coordination rules, and decisions. 

\subsection{Preliminaries}
Throughout this article, elements and functions are denoted by small letters, sets and relations by large letters and mathematical structures by large calligraphic letters. 
The components of a structure may be denoted by the structure's symbol or, in case of enumerated structures, index as subscript. The subscript is dropped if it is clear to which structure a component belongs.

Character sets and sets of state values are assumed to be enumerable if not stated otherwise. For any character set or alphabet $A$, $A^\epsilon := A\cup\{\epsilon\}$ where $\epsilon$ is the empty character. For state value sets $Q$, $Q^\epsilon := Q\cup\{\epsilon\}$ where $\epsilon$ is the undefined value.   If either a character or state value set $A = A_1 \times \dots \times A_n$ is a Cartesian product then $A^\epsilon = A_1^\epsilon \times \dots \times A_n^\epsilon$. 

Elements of character sets or state value sets can be vectors. There will be no notational distinction between single elements and vectors. The change of a state vector $p=(p_1,  \dots, p_n)$ in a position $k$ from $a$ to $b$ is written as $p \left[\frac{b}{a}, k\right]$. 
An n-dimensional vector of characters $a\in C = C_1\times\dots\times C_n$ with the $k$-th component $v$ and the rest $\epsilon$ is written as $\epsilon_C[v,k]=(\epsilon_1, \dots, \epsilon_{k-1}, v, \epsilon_{k+1}, \dots, \epsilon_n)$. The subscript is dropped if it is clear from the context.

The power set of a set $A$ is written as $\wp(A)$.

%
\section{Compositionality}
%
Mathematically, composition\footnote{It was mainly Arend Rensink in his talk ''Compositionality huh?'' at the Dagstuhl Workshop ``Divide and Conquer: the Quest for Compositional Design and Analysis'' in December 2012, who inspired me to these thoughts and to distinguish between composition of systems and the property of being compositional for the properties of the systems.} means making one out of two or more mathematical objects with the help of a mathematical mapping. For example, we can take two functions $f, g: \mathbb{N} \rightarrow \mathbb{N}$, which map the natural numbers onto themselves, and, with the help of the concatenation operator $\circ$, we can define a function $h = f \circ g$ by $h(n) = f(g(n))$.  

If we apply this notion to interacting systems, which we denote by ${\mathcal S}_1, \dots, {\mathcal S}_n$, then, regardless of the concrete representation of these systems, we can define their composition into a supersystem by means of a corresponding composition operator $comp_{\mathcal S}$ as a partial function\footnote{''Partial'' means that this function is not defined for all possible systems, i.e. not every system is suitable for every composition} for systems:

\begin{equation} \label{eq_system_composition}
{\mathcal S}_{tot} = comp_{\mathcal S}({\mathcal S}_1, \dots, {\mathcal S}_n).
\end{equation}

The first benefit we can derive from this definition is that we can now classify the properties of the supersystem into those that arise comparatively simply from the same properties of the subsystems and those that arise from other properties of the subsystems: 

\begin{definition} \label{def_compositionality} 
A property $\alpha:S \rightarrow A$ of a system ${\mathcal S}\in S$ is a partial function which attributes values of some attribute set $A$ to a system ${\mathcal S}\in S$. I call a property $\alpha$ of a composed system ${\mathcal S}_{tot}$ ''{\it (homogeneous) compositional}'' with respect to the composition $comp_{\mathcal S}$, if there exists an operator $comp_\alpha$ such that $\alpha(S_{tot})$ results as $comp_\alpha(\alpha({\mathcal S}_1), \dots, \alpha({\mathcal S}_n))$, thus, it holds:
\begin{equation}
  \alpha(comp_{\mathcal S}({\mathcal S}_1, \dots, {\mathcal S}_n)) = comp_\alpha(\alpha({\mathcal S}_1), \dots, \alpha({\mathcal S}_n))
\end{equation}
Otherwise I call this property ''{\it emergent}''.
\end{definition}

In other words, compositional properties of the composed entity result exclusively from the respective properties of the parts. If the entities are mathematical structures, then $\alpha$ is a homomorphism. Emergent properties may result also from other properties $\alpha_i$ of the parts \cite{Reich2001} if $\alpha\left(comp_{\mathcal S}({\cal S}_1, \dots, {\cal S}_n)\right) = comp_\alpha\left(\alpha_1({\cal S}_1), \dots, \alpha_n({\cal S}_n)\right))$.

A simple example of a homogeneous compositional property of physical systems is their mass: The mass of a total system is equal to the sum of the masses of the individual systems. A simple example of an emergent property of a physical system is the resonance frequency of a resonant circuit consisting of a coil and a capacitor. The resonance frequency of the oscillating circuit does not result from the resonance frequencies of the coil and capacitor, because these do not have such a resonance frequency, but from their inductance and capacity. 

\subsection{Computable functionality \label{ss_computable_functionality}}
One of the most important properties of systems in computer science is the computability of their system function. Based on considerations of Kurt Gödel, Stephen Kleene was able to show in his groundbreaking work \cite{Kleene1936} that this property is indeed compositional by construction. Starting from given elementary operations (successor, constant and identity), all further computable operations on natural numbers can be constructed by the following 3 rules (let $F_n$ be the set of all functions on the natural numbers with arity $n$):

\begin{enumerate}  
\item {\bf\it Comp:} Be $g_1, \dots, g_n \in F_m$ computable and $h\in F_n$ computable, then $f = h(g_1, \dots, g_n)$ is computable.\label{computation_1st_rule}
\item {\bf\it PrimRec:} Are $g\in F_n$ and $h\in F_{n+2}$ both computable and $a\in \mathbb{N}^n$, $b\in \mathbb{N}$ then also the function $f\in F_{n+1}$ given by $f(a, 0) = g(a)$ and $f(a, b+1) = h(a, b, f(a,b))$ is computable  \label{computation_2nd_rule}.
\item {\bf\it $\mu$-Rec:} Be $g\in F_{n+1}$ computable and $\forall a\exists b$ such that $g(a,b) = 0$ and the $\mu$-Operation $\mu_b[g(a,b) = 0]$ is defined as the smallest $b$ with $g(a,b) = 0$. Then $f(a) = \mu_b[g(a,b) = 0]$ is computable.\label{computation_3rd_rule}
\end{enumerate}

{\it Comp} states that given computable functions, their successive as well as their parallel application is in turn a computable function. {\it PrimRec} defines simple recursion, where the function value is computed by applying a given computable function successively a predefined number of times. In imperative programming languages it can be found as FOR loop construct as well as operation construct with one step. The third rule {\it $\mu$-Rec} states that a recursive calculation can also exist as an iterative solution of a computable problem, in the case of the natural numbers as a determination of roots, for which one cannot say in advance how many steps one will need to reach the first solution. Indeed, one may not even know whether such a solution exists at all. In imperative programming languages, this corresponds to the WHILE loop construct.

\subsection{The notion of interface and component}
The second benefit we can derive from our definition of composition is a clear definition of the notion of interface and component. 
We are now in the need for a term that sums up everything a composition operator needs to know about a system. \cite{Tripakis2016_Compositionality,AlfaroHenzinger2001_InterfaceTheories} use the term {\it interface} for this, a suggestion I am happy to endorse. This makes the question of what an interface actually is decidable. 

The recipe is that the one who claims that a mathematical object is an interface must first provide a system model, secondly show what its composition operator looks like, so that thirdly it becomes comprehensible which parts of the system model belong to the interface. 

According to Gerard J. Holzmann \cite{Holzmann2018_Software_Components}, the term software component was coined by Doug McIlroy \cite{McIlroy1968_Software_Components} at the 1968 NATO Conference on Software Engineering in Garmisch, Germany. I propose understand a component as a system that is intended for a particular composition and therefore has correspondingly well-defined interfaces that, by definition, express the intended composition.

\subsection{Substitutability and compatibility} \label{ss_substitutability}
The third benefit of our composition concept comes from the possibility of defining substitutability and to distinguish downward from upward compatibility. 

For this purpose we first consider the two systems ${\mathcal A}$ and ${\mathcal B}$ which can be composed by a composition operator $C$ into the supersystem ${\mathcal S} = C({\mathcal A}, {\mathcal B})$. The system ${\mathcal A}$ can certainly be substituted by another system ${\mathcal A}'$ in this composition if ${\mathcal S} = C({\mathcal A}', {\mathcal B})$ also holds. 
Please note that substitutability does not imply any further relation between the involved systems and can therefore relate seemingly unrelated things, such as raspberry juice and cod liver oil, both of which are suitable fuel for a Fliwatüt \cite{Lornsen1967}.

Now, let us assume that ${\mathcal A}'$ additionally extends ${\mathcal A}$ in a sense to be concretised, notated as ${\mathcal A} \sqsubseteq {\mathcal A}'$ and another operator $C'$ composes ${\mathcal A}'$ and again ${\mathcal B}$ into the supersystem ${\mathcal S}' = C'({\mathcal A}', {\mathcal B})$, so that the property of extension is preserved, thus also ${\mathcal S} \sqsubseteq {\mathcal S}'$ holds. Then we speak of ''{\it compatibility}'' under the following conditions.

If $C'$ applied to the old system ${\mathcal A}$ still produces the old system ${\mathcal S} = C'({\mathcal A}, {\mathcal B})$, then ${\mathcal A}$ behaves ''{\it upward compatible}'' in the context of the composition $C'$. And if our original substitution relation  ${\mathcal S} = C({\mathcal A}', {\mathcal B})$ now holds, that is, ${\mathcal A}$ can be replaced by ${\mathcal A}'$ in context $C$, we say that ${\mathcal A}'$ behaves ''{\it downward compatible}'' in context $C$. 

A simple example is a red light-emitting diode (LED) ${\mathcal A}$ and a socket ${\mathcal B}$ composing to a red lamp ${\mathcal S}=C({\mathcal A}, {\mathcal B})$. A new LED $A'$ that adds to the old one the ability to glow green when the current is reversed is called downward compatible with ${\mathcal A}$ if it fits into the same socket and makes the lamp behave as with the original LED.

Actually, of course, the bi-colour LED is intended for a new lamp circuit that supports this bi-colour. Since the socket remained unchanged, inserting the old LED into the new lamp results in the old, single-colour red lamp despite the new circuit. The old LED behaves upwardly compatible in this context.  

%
\section{Simple systems} \label{s_simple_systems}
%
\subsection{System definition \label{s_system_definition}}
There seems to be a consensus (e.g. \cite{HarelPnueli1985_Reactive,Broy2010_Logical,Sifakis2011_Vision,ISO_15288-2014,IEC_60050}) that a system separates an inside from the rest of the world, the environment. 
So, to gain a well defined system model we have to answer the two questions: What gets separated? And: what separates?

What gets separated? It's time dependent properties each taking a single out of a set of possible values at a given time. These time dependent functions are commonly called ''state variables'' or ''signal'' and the values are called ''states'' \cite{IEC_60050}. However, as often the term ''state'' is also used to denote state functions, I prefer to speak of ''state function'' and ''state value''. 

A state function $s$ in this sense is a function from the time domain $T$ to some set of state values $A$. Some of the state functions may not be directly accessible from the outside, these are the system's inner state functions. 

What separates? The key idea is that a system comes into existence because of a unique relation between the state values at given time values: the system function. It separates the state functions of a system from the state functions of the rest of the world. It also gives the system's state functions their input-, output-, or inner character. Such a relation logically implies causality and a time scale.  

Depending on the class of state function, system function or time, we can identify different classes of systems. Important classes of functions are computable functions, finite functions and analytic functions. Important classes of times are discrete and continuous times\footnote{Referring to quantum physics we could also distinguish between classical and quantum states.}. 

For the description of physical systems, we use the known physical quantities like current, voltage, pressure, mass, etc. as value range for the state functions. However, the idea of informatics is to only focus on the distinguishable \cite{Shannon1948, Shannon1949_Noise}, setting aside the original physical character of the state functions. Consequently, we have to introduce an additional character set or ''alphabet'' to our engineering language to name these distinguishable state values. Thereby, we create the concepts of communication and information as information being the something that becomes transported by communication: what was distinguishable in one system and becomes transported is now distinguishable in another system. Information becomes measurable by the unit which could be only just securely distinguished: the bit. The different names of the characters of such an alphabet carry no semantics beside that they denote the same distinguishable (physical) state value in any physical system of our choice in our engineering language. 

This independence of the concept of information from its concrete physical realisation gives purely informational systems a virtual character in the sense that they could be realised by infinitely many physical systems. With ''purely informational systems'' I mean systems that interact only through communication with the rest of the world, that is by sending and receiving information. In contrast so called ''cyber physical systems'' \cite{Alur2015_Principles} also affect the rest of the world through other actors, possibly irreversibly and thereby loose this virtual character. In short, we can describe the essence of us talking to someone in the framework of information transport (and processing) but not the essence of hitting someone's nose with our fist. 

In this article, I will focus on three different discrete informational systems with increasingly complex I/O-behaviour: simple systems, multi-input-systems with recursion and interactive systems. 

\begin{definition} \label{def_system_einfach}
A simple (discrete) system ${\mathcal{S}}$ is defined by ${\mathcal{S}} = (T, succ, Q, I, O, q,$ $ in, out, f)$. 

The {\it system time} is the structure ${\cal T} = (T, succ)$ with $T$ the enumerable set of all time values and $succ:T\rightarrow T$ with $t' = succ(t)$ is the successor function. For an infinite number of time values, $T$ can be identified with the set of natural numbers $\mathbb{N}$. For a finite number of time values, $succ(t_{max})$ is undefined. $t'$ is also called the successor time of $t$ with respect to $f$. I also write $t+n = succ^n(t)$.

$Q$, $I$, and $O$ are alphabets and at least $Q$ is non-empty. The state functions $(q, in, out):T\rightarrow Q \times I \times O$ form a simple system for time step $(t, t'=succ(t))$ if they are mapped by the partial function $f:Q \times I \rightarrow Q \times O$ with $f =(f^{int}, f^{ext})$ such that\footnote{Please note that input and output state functions relate to consecutive points in time. In so called synchronous reactive systems \cite{BenvenisteCaspiEdwardsHalbwachsLeGuernicDeSimone2003,Lunze2006_Ereignisdiskrete}, the input- and output state functions relate to the same points in time. See also section \ref{s_synchronous_reactive_systems}} 
\[{q(t') \choose out(t')} = {f^{int}(q(t), i(t)) \choose f^{ext}(q(t), i(t))}\,.\]

If $Q$, $I$ and $O$ are finite then the system is said to be finite. If $|Q| == 1$ I call the system ''stateless''.
\end{definition}

As the elements of the alphabets can be chosen arbitrarily, as long as they are distinguishable, it is important to realise that we actually describe systems only up to an isomorphism. 
 
\begin{definition}
Be ${\mathcal S}_1$ and ${\mathcal S}_2$ two simple systems. They are isomorph to each other, ${\mathcal S}_1 \cong {\mathcal S}_2$, if there exists a bijective function $\varphi$ with $\varphi:Q_1 \cup I_1 \cup O_1 \cup T_1 \rightarrow Q_2 \cup I_2 \cup O_2 \cup T_2$ such that for all $(p, i)\in Q_1\times I_1$ where $f_1(p, i)$ is defined, it holds $\varphi(f_1(p, i))=f_2(\varphi(p),\varphi(i))$ and for all $t\in T_1$ with defined $succ_1(t)$, it holds $\varphi(succ_1(t)) = succ_2(\varphi(t))$.
\end{definition}

We can use the simple stepwise behaviour to eliminate the explicit usage of the time parameter from the behaviour description of the system with an I/O-transition system (e.g. \cite{Reich2010,Alur2015_Principles}). 

\begin{definition} \label{def_IO-transition_system}
The I/O-transition system (I/O-TS) ${\mathcal A} = (I, O, Q, (q_0, o_0), \Delta)_{\mathcal A}$ with $\Delta_{\mathcal A} \subseteq I \times O \times Q \times Q$ describes the behaviour of the simple system ${\mathcal S}$ from time $t=t_0$, if $Q_{\mathcal S} \subseteq Q_{\mathcal A}$, $I_{\mathcal S} \subseteq I_{\mathcal A}$, $O_{\mathcal S} \subseteq O_{\mathcal A}$, $q_{\mathcal S}(t_0) = q_0$ and $out_{\mathcal S}(t_0) = o_0$, as well as $\Delta_{\mathcal A}$ is the smallest possible set, such that for all times $t \geq t_0$ with $t'=succ(t)$ for all possible input state functions $in_{\mathcal S}$ holds: $(in_{\mathcal S}(t), out_{\mathcal S}(t'), q_{\mathcal S}(t), q_{\mathcal S}(t')) \in \Delta_{\mathcal A}$. 

Instead of $(i,o,p,q) \in \Delta$ I also write $p \stackrel{i/o}{\rightarrow}_{\Delta} q$, where I drop the $\Delta$-subscript, if it is clear from its context.
\end{definition}

The transition relation $\Delta_{\mathcal A}$ is the graph of the system operation $f_{\mathcal S}$. Thus, I could also provide this function $f_{\mathcal S}: Q\times I \rightarrow Q \times O$ with $(q, o) = f_{\mathcal S}(p,i)$ for any $p \stackrel{i/o}{\rightarrow}_{\Delta} q$ instead of the transition relation. 

If $Q_{\mathcal A}$ is as small as possible and $|Q_{\mathcal A}| = 1$, then I call the I/O-TS as well as its behaviour stateless, otherwise I call it stateful.

The properties of a I/O-TS may depend on its execution model. I define the following execution model:
 
\begin{definition} \label{def_ausführungsvorschrift}
Be ${\cal A}$ an I/O-TS. A run as a sequence $(q_0,o_0), (q_1, o_1), \dots$ is calculated for an input sequence $i_0, i_1, \dots$ according to the following rules:

\begin{enumerate} 
\item {\bf Initialisation (time $j=0$):} $(q_0, o_0) = (q_0, o_0)_{\mathcal A}$.
\item {\bf Check:}\label{step_2} If $i_j$ is undefined at time $j$ (the finite case) then stop the calculation. 
\item {\bf Transition:} Be $i_j$ the current input sign at time $j$ and $q_j$ the current state value, then chose one of the possible transitions with $(i_j, o_{j+1}, q_j, q_{j+1})\in \Delta_{\cal A}$ and determine the next state value $q_{j+1}$ and output character $o_{j+1}$. If no such transition exist, stop the calculation with an exception.  
\item {\bf Repeat:} transit to next point in time and proceed with next Check (\ref{step_2}.)
\end{enumerate}
\end{definition}

\subsection{System interaction \label{ss_system_interaction}}

Interaction between two systems in our sense means that they share a common state function, a so called ''Shannon state function'' or idealised Shannon channel \cite{BangemannDiedrichReich2016}, which is an output state function for the ''sender'' system and an input state function of the ''receiver system. In the context of I/O-TSs, this means that the state values of an output component of a transition of a ''sender'' I/O-TS are reproduced in the input component of the ''receiver'' system and serve there as input of a further transition (see Fig. \ref{fig_coupling_by_interaction}). 
In other words, in our model, interaction simply means that information is transmitted\footnote{This is not as self-evident as it perhaps seems. There are many process calculi whose interaction model is based on the synchronisation of identical or complementary named transitions, e.g. \cite{Milner1980,Hoare1985,Milner1992}.} and beyond that the structure of the interacting systems remains invariant. Accordingly, the coupling mechanism of Shannon state functions result in a coupling mechanism in the context of I/O-TSs being based on the use of equal characters in the sending and receiving I/O-TS. 

\begin{figure}[htbp]
  \begin{center}
    \includegraphics[width=6cm]{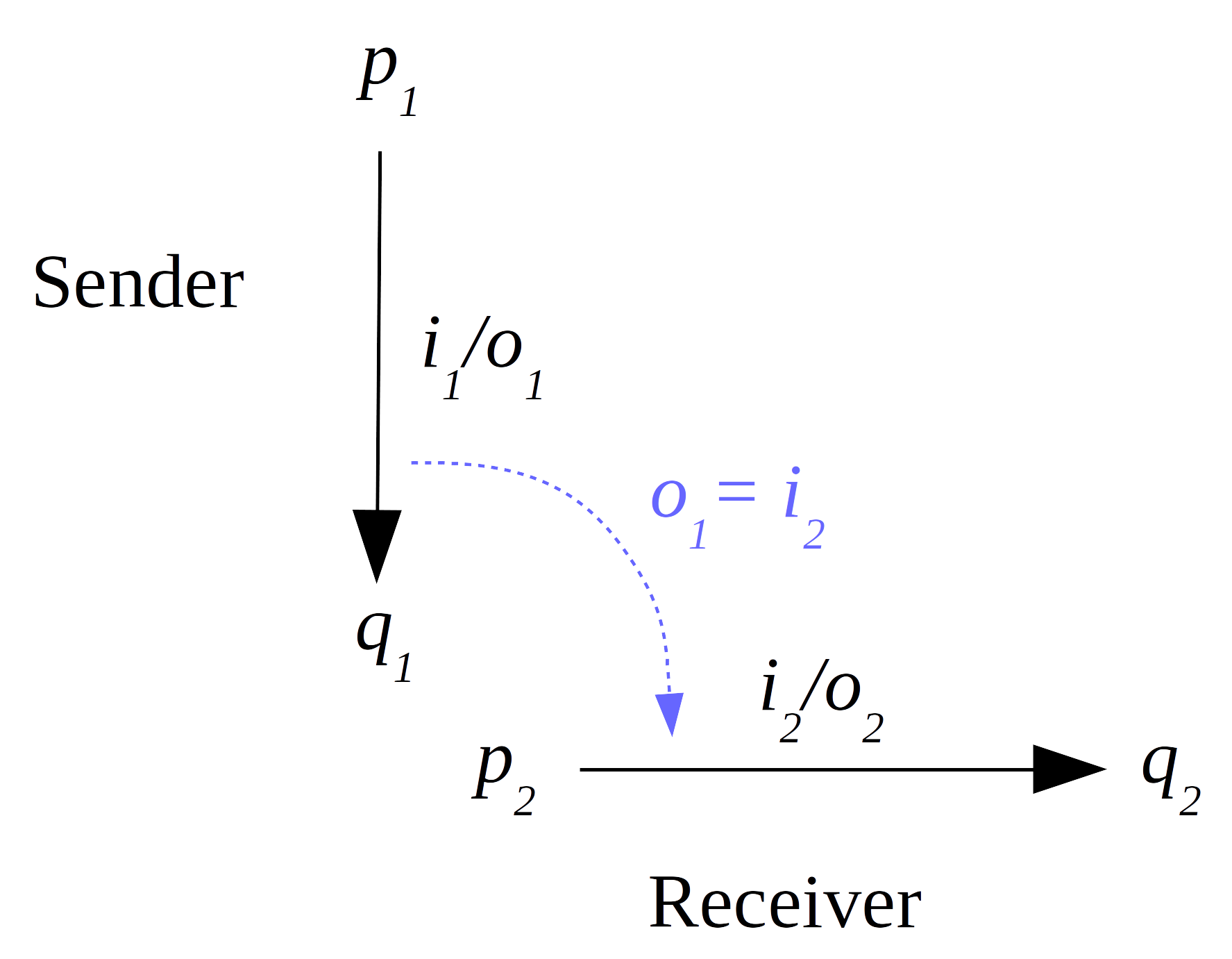}
 \end{center}
\caption[]{Interaction between two I/O-TSs in which the output character of a ''sender'' system is used as the input character of a ''receiver'' system. Interaction therefore means the coupling of the two I/O-TSs of sender and receiver based on the ''exchanged'' character.} \label{fig_coupling_by_interaction}
\end{figure}

Later we will have systems with multi-component input or output. There, such an idealised Shannon channel between a sending system ${\mathcal U}$ and a receiving system ${\mathcal V}$ is characterised by a pair of two indices $(k_{\mathcal U},l_{\mathcal V})$.

What is the effect of a Shannon channel on the execution rule? The execution rule for the transmission of a character by an idealised Shannon channel is, that after sending a character through a Shannon channel, in the next step a transition of the receiver system must be selected which processes this character.

\subsection{Sequential system composition \label{ss_sequential_system_composition}}
First, I define what I mean by saying that two simple systems work sequentially and then I show that under these circumstances, a super system can always be identified. Sequential system composition means that one system's output is completely fed into another system's input, see Fig. \ref{fig_system_composition_sequential}. 

\begin{figure}[htbp]
  \begin{center}
    \includegraphics[width=12cm]{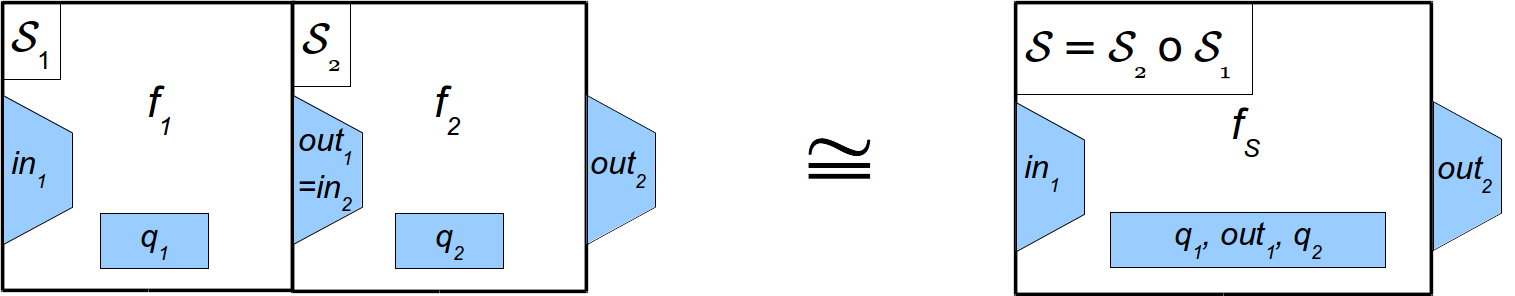}\\
 \end{center}
\caption[]{Diagram of a sequential system composition as defined in Def. \ref{def_system_composition_sequential}. The output state function of the first system $out_1$ is identical to the input state function $in_2$ of the second system.} \label{fig_system_composition_sequential}
\end{figure}

\begin{definition} \label{def_system_concatenation}
Two simple systems ${\cal S}_1$ and ${\cal S}_2$ with $O_1 \subseteq I_2$ and $\tau_i, \tau_i' \in T_i$ are said to ''{\it work sequentially}'' or to be ''{\it concatenated}'' at $(\tau_1, \tau_2)$ if $out_1(\tau_1') = in_2(\tau_2)$ is an idealised Shannon channel.
\end{definition}

\begin{proposition} \label{def_system_composition_sequential}
Let ${\cal S}_1$ and ${\cal S}_2$ be two simple systems that work sequentially at time $(\tau_1, \tau_2)$. Then the states  $(q, in, out)$ with $q = (q_1, q_2)$, $in = in_1$, and $out=out_2$ form a simple system ${\mathcal S}$ for time step $((\tau_1, \tau_2), (\tau_1', \tau_2'))$. I also write ${\cal S} = {\cal S}_2 \circ_{(\tau_1, \tau_2, out_1, in_2)} {\cal S}_1$ with the composition operator for sequential composition (or concatenation) $\circ_{(\tau_1, \tau_2, out_1, in_2)}$. I omit the time step and the concatenation states if they are clear from the context. \end{proposition}

\begin{proof}
To prove the proposition we have to identify a system function and a time function for the composed system. Based on the concatenation condition $out_1(\tau_1') = in_2(\tau_2)$ we can replace $in_2(\tau_2)$ by $out_1(\tau_1') = f^{ext}_1(q_1(t_1), i_1(t_1))$. We get 
\begin{multline} \label{eq_sequential_composition}
\left( \begin{array}{c}
q_1(\tau_1'), q_2(\tau_2') \\
out_2(\tau_2') 
\end{array}\right) = \\
\left(\begin{array}{c}
f^{int}_1 \left(q_1(\tau_1), i_1(\tau_1)\right), f^{int}_2 \left(q_2(\tau_2), f^{ext}_1(q_1(\tau_1), i_1(\tau_1))\right)\\ 
f^{ext}_2 \left(q_2(\tau_2), f^{ext}_1(q_1(\tau_1), i_1(\tau_1))\right) \\ 
\end{array}\right)
\end{multline}
where the right hand side depends exclusively on $\tau_1$ and $\tau_2$ and the left hand side, the result, exclusively on $\tau_1'$ and $\tau_2'$.

The composed system ${\mathcal S}$ exists with the time values $T_{\mathcal S} = \{0,1\}$. With the functions $map_i: T_{\mathcal S} \rightarrow T_i$, $(i=1,2)$ with $map_i(0) = \tau_i$ and $map_i(1)=\tau_i'$, the states of ${\mathcal S}$ are: 
\begin{eqnarray*}
  out_{\mathcal S}(\tau_{\mathcal S}) & = & out_2(map_2(\tau_{\mathcal S})), \\
  in_{\mathcal S}(\tau_{\mathcal S})  & = & in_1(map_1(\tau_{\mathcal S})) \text{, and} \\
  q_{\mathcal S}(\tau_{\mathcal S})   & = & (q_1(map_1(\tau_{\mathcal S})), out_1(map_1(\tau_{\mathcal S})), q_2(map_2(\tau_{\mathcal S}))).
\end{eqnarray*}
\end{proof}

The extension to more then two concatenated systems or time steps is straight forward.
The concatenation condition  $out_1(\tau_1') = in_2(\tau_2)$ implies that there is an important difference between the inner and outer time structure of sequentially composed systems. To compose sequentially, the second system can take its step only after the first one has completed its own. 

In the special case, where $f^{ext}_i: I_i \rightarrow O_i$ is stateless in both systems, the concatenation of these stateless systems results in the simple concatenation of the two system functions: $out_2(\tau_2') = (f^{ext}_2 \circ f^{ext}_1)(in_1(\tau_1))$ 

\subsection{Parallel system composition}
Parallel processing simple systems can also be viewed as one system if there is a common input to them and the parallel processing operations are independent and therefore well defined  as is illustrated in Fig. \ref{fig_system_composition_parallel}.

\begin{figure}[htbp]
  \begin{center}
    \includegraphics[width=8cm]{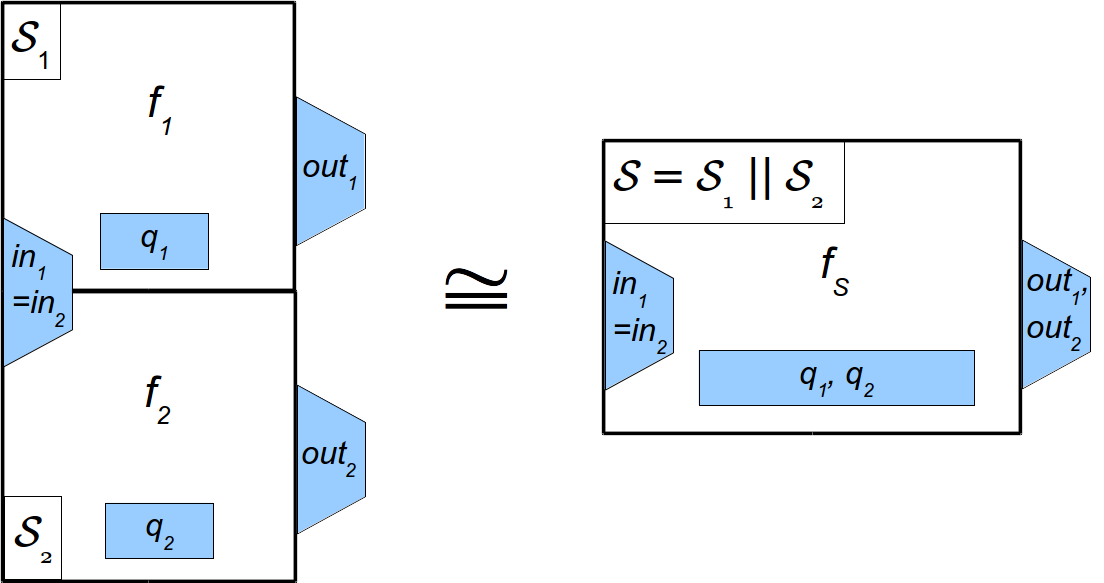}\\
 \end{center}
\caption[]{Diagram of a parallel system composition as defined in Def. \ref{def_system_composition_parallel}. The input of both systems is identical.} \label{fig_system_composition_parallel}
\end{figure}

\begin{definition} \label{def_system_composition_parallel}
Two discrete systems ${\cal S}_1$ and ${\cal S}_2$ with $I_1=I_2$ are said to ''{\it work in parallel}'' at $(\tau_1, \tau_2)$ if $q_1 \neq q_2$, $out_1 \neq out_2$ and $\tau_1\in T_1$, $\tau_2\in T_2$ exist such that $in_1(\tau_1) = in_2(\tau_2)$.
\end{definition}

\begin{proposition} 
Let  ${\cal S}_1$ and ${\cal S}_2$ be two systems that work in parallel at $(\tau_1, \tau_2)$. Then the state $(q, in, out)$, with $q = (q_1, q_2)$ $in=in_1=in_2$, $out = (out_1, out_2)$ form a system ${\mathcal S}$ for time step $((\tau_1, \tau_2), (\tau_1', \tau_2'))$. I also write  ${\cal S} = {\cal S}_2 ||_{(\tau_1, \tau_2, in_1, in_2)} {\cal S}_1$ with the composition operator for parallel composition $||_{(\tau_1, \tau_2, in_1, in_2)}$. Again, I omit the time step and the states if they are clear from the context.
\end{proposition}

\begin{proof} 
To prove, we have to provide a system function for the composed system $C$ at times $T_C=\{0,1\}$, which is simply $f_C: (Q_1\times Q_2) \times (I_1^\epsilon \times I_2^\epsilon) \rightarrow (Q_1^\epsilon\times Q_2^\epsilon) \times (O_1^\epsilon\times O_2^\epsilon)$ with
\begin{multline*}f_C(\tau_C') = \\
{f^{int}_1(q_1(map_1(\tau_C)), in(map_1(\tau_C))), f^{ext}_1(q_2(map_2(\tau_C)), in(map_1(\tau_C))) \choose f^{int}_2(q_1(map_1(\tau_C)), in(map_1(\tau_C))), f^{ext}_2(q_2(map_2(\tau_C)), in(map_1(\tau_C)))}
\end{multline*}
\end{proof}

Please note that the composed system's time step depends only on the logical time steps of the subsystems. If both systems don't perform their step at the same time compared to an external clock then the composed system does not realise its time step in a synchronised manner. If this synchronisation takes place, then there is effectively no difference between the inner and outer time structure of parallel composed systems. 

Again, the extension to more than two parallel working systems or more than two time steps is straight forward.

\subsection{Combining sequential and parallel system composition}
As the application of an operation $f$ preceding two parallel operations $g$ and $h$ is equivalent to the parallel execution of $g$ after $f$ and $h$ after $f$, together with some book keeping on time steps and i/o-states, the following proposition is easy to prove:

\begin{proposition} 
For three systems ${\cal P}_1$, ${\cal P}_2$ and ${\cal S}$ where ${\cal P}_1$ and ${\cal P}_2$ are parallel systems and ${\cal P}_1 || {\cal P}_2$  and ${\cal S}$ are sequential systems, then the two systems $({\cal P}_1 || {\cal P}_2) \circ {\cal S}$ and $({\cal P}_1 \circ {\cal S}) || ({\cal P}_2 \circ {\cal S})$ are functional equivalent, that is $({\cal P}_1 || {\cal P}_2) \circ {\cal S}\cong ({\cal P}_1 \circ {\cal S}) || ({\cal P}_2 \circ {\cal S})$. In other words, the right distribution law holds for parallel and sequential composition with respect to functional equivalence.
\end{proposition}

As sequential system composition is non-commutative, it follows that the left distribution law does not hold:  ${\cal S} \circ ({\cal P}_1 || {\cal P}_2) \ncong ({\cal S} \circ {\cal P}_1) || ({\cal S} \circ {\cal P}_2)$. 

Finally I point out that combining parallel and sequential system composition as ${\cal S} \circ ({\cal P}_1 || {\cal P}_2)$ for stateless systems results in a combination of their system functions as $f_{\cal S} (f_{{\cal P}_1}, f_{{\cal P}_2})$, which is exactly the way, the first rule {\it Comp} for computable functions was defined. Hence, from a formal point of view, for computational systems, sequential together with parallel system composition represent the first level to compose computable functionality. 

%
\section{Recursive systems} \label{s_recursive_systems}
%
If a system operates at least partially on its own output, a recursive system is being created as the rules for creating the system function follow the rules for creating recursive functionality (see section \ref{ss_computable_functionality}).

\begin{figure}[htbp]
  \begin{center}
    \includegraphics[width=10cm]{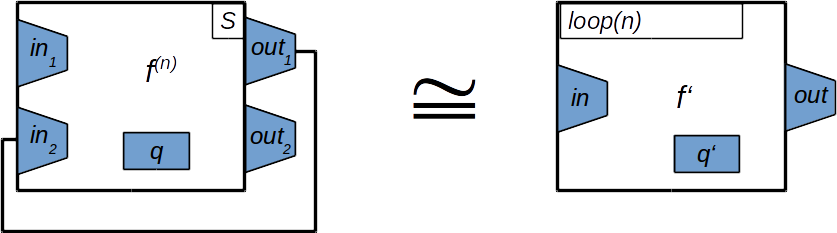}\\
 \end{center}
\caption[]{Diagram of a simple recursive system.} \label{fig_system_composition_loop}
\end{figure}

As I show in Fig.  \ref{fig_system_composition_loop}, a recursive system needs at least two input- and two output-components. Recursive systems are therefore not simple systems in the sense of Def. \ref{def_system_einfach}. To be able to describe them, I introduce the more elaborate system class of ''multi-input-system (MIS)''. 

\subsection{Multi-input-systems}
First, I introduce the ''empty'' character $\epsilon$. By ''empty'' I mean the property of this character to be the neutral element regarding the concatenation operation of characters of an alphabet $A$. As already introduced in the preliminaries, I define $A^\epsilon = A \cup \{\epsilon\}$ and if $A = A_1 \times \dots \times A_n$ is a product set then $A^\epsilon = A_1^\epsilon \times \dots \times A_n^\epsilon$.

From a physical point of view, a state function cannot indicate the empty character. It is supposed to always represent a ''real'' character. However, I use the empty character in my formalism to indicate that certain components of input state functions are irrelevant in a given processing step. 
This means that for an input character vector $i=(i_1, \dots, i_n) \in I^\epsilon_1\times \dots \times I^\epsilon_n$ with $i_k \neq \epsilon$ but all other components $i_{l\neq k} = \epsilon$ (I also write $i=\epsilon[i_k, k]$), the system function $f$ is determined by a function $f_k$ that only depends on characters from $I_k$. That is, in this case $f = f_k:I_k\times Q \rightarrow O \times Q$ with $f(i, p) = f_k(i_k, p)$. 

\begin{definition} \label{def_system_multi-input}
A multi-input system (MIS) ${\mathcal{S}}$ is defined as ${\mathcal{S}} = (T, succ, Q, I, O, q,$ $ in, out, f)$. $Q$, $I$ and $O$ are alphabets, where $I$ has at least 2 components and $I$ and $Q$ have to be non-empty. The state functions $(q, in, out):T\rightarrow Q \times I^{\epsilon} \times O^{\epsilon}$ form a discrete system for the time step $(t, t'=succ(t))$ if they are mapped by the partial function $f:Q \times I^\epsilon \rightarrow Q \times O^\epsilon$ with $f =(f^{int}, f^{ext})$ such that  
\[{q(t') \choose out(t')} = {f^{int}(q(t), i(t)) \choose f^{ext}(q(t), i(t))}\,.\]
\end{definition}

To describe the behaviour of an MIS with an I/O-TS, we have to extend the transition relation in its input- and output characters also by the empty character as follows:

\begin{definition} \label{def_IO-transition_system_mit_epsilon}
The I/O-TS ${\mathcal A} = (I, O, Q, (q_0, o_0), \Delta)_{\mathcal A}$ with $\Delta_{\mathcal A} \subseteq I^\epsilon \times O^\epsilon \times Q \times Q$ describes the behaviour of a discrete system ${\mathcal S}$ from time $t=t_0$, if $Q_{\mathcal S} \subseteq Q_{\mathcal A}$, $I_{\mathcal S} \subseteq I_{\mathcal A}$, $O_{\mathcal S} \subseteq O_{\mathcal A}$, $q_{\mathcal S}(t_0) = q_0$ and $out_{\mathcal S}(t_0) = o_0$ and $\Delta_{\mathcal A}$ is the smallest possible set such that for all times $t \geq t_0$ with $t'=succ(t)$ and for all possible input state functions $in_{\mathcal S}$ holds: $(in(t), out(t'), q(t), q(t')) \in \Delta_{\mathcal A}$.
\end{definition}

If I talk about an I/O-TS, I will usually refer to this more general one. I call transitions with $\epsilon$ as input ''spontaneous''. They cannot occur in transition systems representing a system function, as then, all input characters of its transition relation have to be unequal to $\epsilon$ in at least one component. However, these spontaneous transitions will play a very important role in my description of the behaviour of interactive systems later on. 

To simplify my further considerations, I will consider only systems that operate on input and output characters that are unequal to $\epsilon$ in at most one component. I therefore can restrict the execution mechanism to linear execution chains without any branching.

\subsection{Feedback} \label{ss_Rückkopplung}
In recursive systems, one component of an output state function is also used as an input component. I call this mechanism ''feedback'' and this output/input component a feedback component. It can be represented by an index pair $(k, l)_{\mathcal S}$ where $k$ is the index of the output component and $l$ is the index of the input component of an MIS ${\mathcal S}$.

In accordance with our sequential processing rule that output characters that serve as input characters must be processed next, the calculation of a recursive system is initiated by an external input character and then goes on as long as the feedback component provides a character unequal to the empty character. As with concatenation of simple systems, this results in a different ''internal'' versus ''external'' time scale. 
As a result, we have to modify the calculation rules \ref{def_ausführungsvorschrift} for MIS in the following way:

\begin{definition} \label{def_MIS_execution}
Be ${\mathcal A}$ an I/O-TS representing the behaviour of an MIS ${\cal S}$ with $I_{\mathcal A} = (I_1, \dots, I_{n_I})$ and $O_{\mathcal A} = (O_1, \dots, O_{n_O}) $ and $C$ a set of feedback components. Be further $seq^{(i)} = (i_0, \dots, i_{n-1})$ a (possibly infinite) sequence of external input characters from $I_{\mathcal A}$ of the form $\epsilon[v,k]$, that is,  they have exactly one component not equal to $\epsilon$.

The counter $j$ counts the current external input character. The current values of $i$, $o$, and $q$ are marked with a $*$. The values that are calculated in the current step are marked with a $+$.
Then, we calculate a run as follows: 

\begin{enumerate} 
\item {\bf Initialisation (time $j=0$):} $(q^*, o^*) = (q_0, o_0)_{\mathcal A}$.

\item {\bf Loop:} \label{calc_MIS_loop} Determine for the current state $q^*$ the set of all possible transitions. If this set is empty, stop the computation.

\item {\bf Determine the current input character  $i^* \in I_{\mathcal A}^\epsilon$} according to the following hierarchy:
\begin{enumerate}
  \item Is the current output character (as a result of the last transition) $o^*=\epsilon[v,k]$, that is , $o^*$ has as $k$-th component the value $v\neq\epsilon$ and $o^*$ is part of a feedback component $c=(k,l)$ for the input component $0 \leq l \leq n_I$, then set $i^* = \epsilon[v,l]$ and continue with step \ref{calc_MIS_transition}.
  \item Otherwise, if at time $j$ there is no further input character $i_j$, stop the computation. 
  \item Otherwise set $i^* = i_j$, increment $j$, and continue with step \ref{calc_MIS_transition}.
\end{enumerate}

\item {\bf Transition:} \label{calc_MIS_transition} With the current state value $q^*$ and the current input character $i^*$, choose a transition $t=(i^*, o^+, q^*, q^+)\in\Delta_{\mathcal A}$ and thereby determine $o^+$ and $q^+$. If there is no transition, stop the computation with an exception. 

\item {\bf Repeat:} Set $q^* = q^+$ and $o^* = o^+$ and continue the computation at \ref{calc_MIS_loop}.
\end{enumerate}
\end{definition}  


\subsection{Primary and $\mu$-recursive systems}
An MIS that works on its own output can realise a computation according to the second rule {\it PrimRec}, as the following proposition says. 

\begin{proposition} \label{prop_einfaches_rekursives_system}
Be $g\in F_n$ and $h\in F_{n+2}$ the computable functions of rule {\it PrimRec}. Be ${\mathcal S}$ an MIS with 2 input components $in = (in_1, in_2)$, two output components $out=(out_1, out_2)$, two internal state components $q=(q_1, q_2)$ and the feedback condition $out_1 = in_2$ and the following system function

\begin{eqnarray}
f_{\mathcal{S}}^{int}(in, q) & = & \begin{cases}q_2 = -1:& (in_1, 0)\\q_2 \geq 0:&(q_1, q_2	+1)\end{cases} \\
f_{\mathcal{S}}^{ext}(in, q) & = & \begin{cases}q_2 = -1: & (g(in_1), \epsilon)\\ q_2 = 0 \dots n-1:& (h(q_1, q_2, in_2), \epsilon) \\ q_2 = n:& (\epsilon, h(q_1, q_2, in_2))\end{cases}
\end{eqnarray}

With the initial values $q_2(0)=-1$ and $out(0)=(\epsilon, \epsilon$) it realises the computation according to the second rule {\it PrimRec}. I call such a system also a  primary recursive system.  
\end{proposition}

\begin{proof}
That this system indeed calculates $f(a,n)$ of rule {\it PrimRec} in its $n+1$ time step can be understood with table \ref{tab_recursive_system_execution}.
\end{proof}

\begin{table}
\begin{center}
\begin{tabular}{cc|lllll}
ext.     & int.     & $in_1$     & $in_2 = out_1$                  &          $out_2$ &    $q_1$ &     $q_2$\\
time     & time\\
\hline
 0       & 0        & a          & $\epsilon$                      &                * &        * &       -1 \\
 0       & 1        & $\epsilon$ & $f(a,0) = g(a)$                 &       $\epsilon$ &        a &        0 \\
 0       & 2        & $\epsilon$ & $f(a,1) = h(a, 0, g(a))$        &       $\epsilon$ &        a &        1 \\
 0       & 3        & $\epsilon$ & $f(a,2) = h(a, 1, h(a,0,g(a)))$ &       $\epsilon$ &        a &        2 \\
 0       & $\vdots$ & $\vdots$   & $\vdots$                        &       $\vdots$   & $\vdots$ & $\vdots$ \\
 0       & $n$      & $\epsilon$ & $f(a,n-1) = \dots$              &       $\epsilon$ &        a &    $n-1$ \\
 1       & $n+1$    &          * & $\epsilon$                      & $f(a,n) = \dots$ &        a &      $n$ \\
\end{tabular} 
\end{center}
\caption[]{This table shows the calculation of the system function of a primary recursive system at times 0 to $n+1$, where the final result is presented at its second output component $out_2$. The internal calculation terminates when the coupling signal $out_1 = in_2$ represents the empty character.} \label{tab_recursive_system_execution}
\end{table}

It is interesting to note a couple of consequences of recursive super system formation. First, as with sequential composition, we can distinguish between an internal and an external time scale. Secondly, internally, the system has an internal counting state which is necessary for keeping temporary interim results during the calculation of the system function of the supersystem, but nevertheless the recursively composed supersystem could be stateless with respect to its external time step. And thirdly, the feedback signal $in_2 = out_1$ becomes an internal state signal, as the actual input of the recursive supersystem is $in_1$ and the output is $out_2$. Thus in a sense, the system border changes with the feedback condition, rendering two former external accessible signals into internal ones. 

We can also make the system function of any suitable stateless simple system ${\mathcal S}_1$ to become the function $h$ of the recursion rule {\it PrimRec}, as is illustrated in Fig. \ref{fig_system_composition_recursive}. Then, ${\mathcal S}_1$ has three input components $in_{1,i}$ and one output component $out_{1,1}$. Its system function $h$ is computable and given by $out'_{1,1} = h(in_{1,1},in_{1,2},in_{1,3})$.

\begin{figure}[htbp]
  \begin{center}
    \includegraphics[width=10cm]{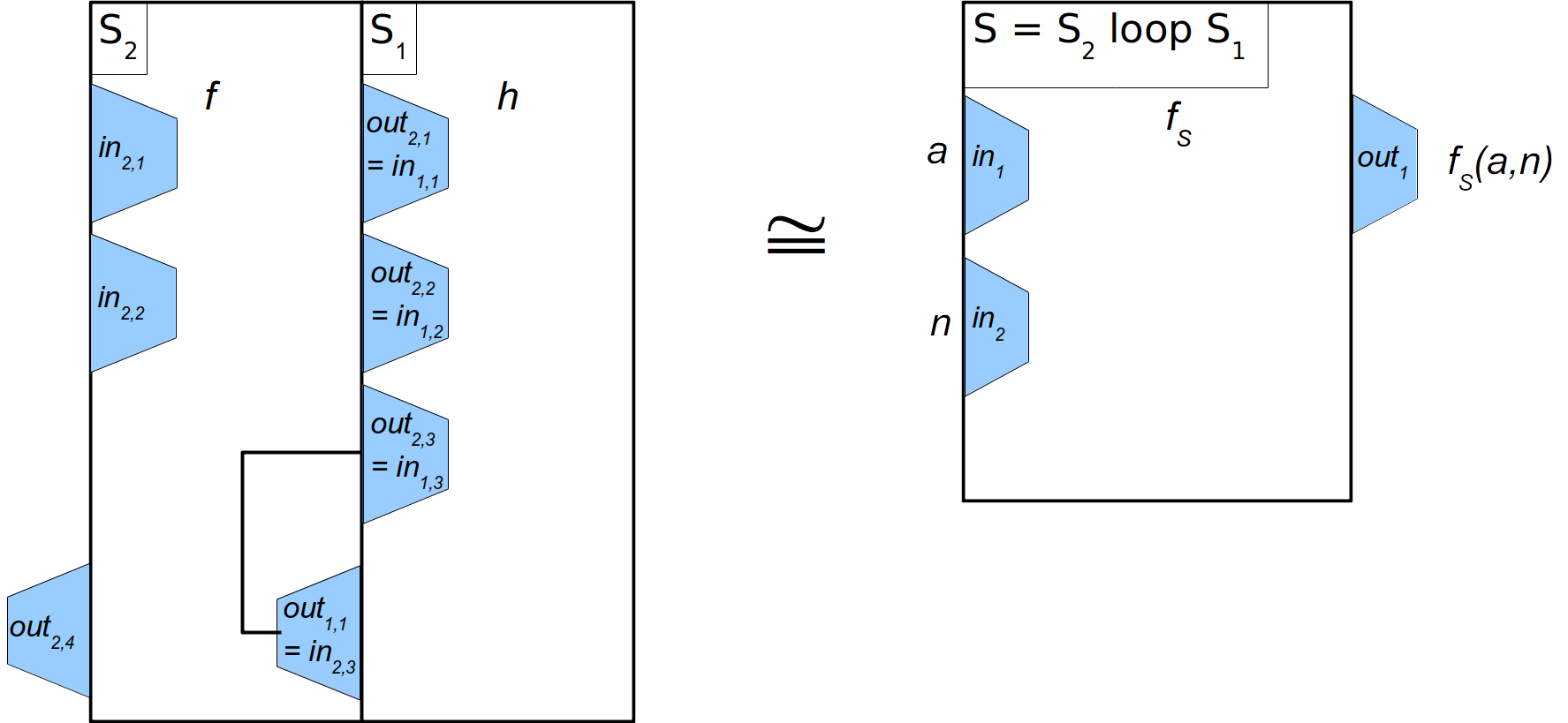}\\
 \end{center}
\caption[]{Diagram of a recursive system, created by composition of a simple system ${\mathcal S}_1$ and an MIS that does the book keeping for the recursion.} \label{fig_system_composition_recursive}
\end{figure}

We define the necessary complementary system ${\mathcal S}_2$, that does the recursive book keeping, as an MIS with three input components $in_{2,i}$, four output components $out_{2,j}$ and an internal state $q$. Actually we tweak the definition of the system a bit by allowing the system function to also operate on output states. This can be justified because these states become internal states in the composed system anyway. 

As shown in Fig. \ref{fig_system_composition_recursive}, ${\mathcal S}_1$ and ${\mathcal S}_2$ couple by the identities $in_{1,1}=out_{2,1}$, $in_{1,2}=out_{2,2}$, $in_{1,3}=out_{2,3}$, and in the other direction $in_{2,3}=out_{1,1}$.

The system function of ${\mathcal S}_2$ routes the value of $a$ through (input state $in_{2,1}$), provides the counter variable (''output'' state $out_{2,2}$) and organises the feedback by the internal coupling $in_{2,3} = out_{2,3}$. By the latter condition, we achieve that the output $out_{2,3}$ of ${\mathcal S}_2$ becomes identical to the output $out_{1,1}$ of system ${\mathcal S}_1$ and thereby it represents the result of the function of ${\mathcal S}_1$ during recursion.
Concretely, given  the initial conditions $out_{2,2}(0) = 0$, $q(0)=max$ the system function of ${\mathcal S}_2$ is:
\begin{align*}
q'         &= \left\{\begin{array}{lcl}out_{2,2}=0 &:& in_{2,2}\\ out_{2,2}>0 &:&  q\end{array}\right.  \text{// stores number of runs of the loop}\\
out_{2,1}' &= \left\{\begin{array}{lcl}out_{2,2}=0 &:& in_{2,1}\\ out_{2,2}>0 &:& out_{2,1}\end{array}\right. \text{// stores $a$}\\
out_{2,2}' &= out_{2,2} + 1 \,\, \text{// counter}\\
out_{2,3}' &= out_{1,1}' = \left\{\begin{array}{lcl}out_{2,2} = 0 &:& g(in_{2,1})\\out_{2,2} > 0 &:& h(out_{2,1}, out_{2,2}-1, out_{2,3})\end{array}\right.\\
out_{2,4}' &= \left\{\begin{array}{lcl}out_{2,2}(t) < q &:& \epsilon\\out_{2,2}(t) = q &:& h(out_{2,1}, out_{2,2}-1, out_{2,3})\end{array}\right. \text{// result}
\end{align*}

In table \ref{tab_system_composition_recursive} I illustrate a run of a composed system for 2 iterations, thus calculating  $h(a,1,h(a,0,g(a)))$. Again, the counter state and the state to hold the interim results became internal states of the composed system.

\begin{table}
\begin{center}
\begin{tabular}{p{0.5cm}|p{0.5cm}p{1.1cm}p{2.5cm}|p{0,5cm}p{1.1cm}p{1.1cm}p{1.1cm}p{2cm}}
time & $in_{2,1}$ & $in_{2,2}$ & $in_{2,3}=out_{1,1}$  & $q$ & $out_{2,1}$  & $out_{2,2}$  & $out_{2,4}$ \\
     &            &            & $=out_{2,3}=in_{1,3}$ &     &$=in_{1,1}$   & $=in_{1,2}$  & \\
\hline
  0  & $a$        & 2          & *                     & max & $*$          & 0            & $\epsilon$\\
  1  & *          & *          & $g(a)$                & 2   & $a$          & 1            & $\epsilon$\\    
  2  & *          & *          & $h(a,0,g(a))$         & 2   & $a$          & 2            & $\epsilon$\\
  3  & *          & *          & *                     & 2   & $a$          & *            & $h(a,1,h(a,0,g(a)))$ \\
\end{tabular} 
\end{center}
\caption[]{The run of the composed system of Fig. \ref{fig_system_composition_recursive} for 2 iterations.} \label{tab_system_composition_recursive}
\end{table}

To realise a $\mu$-recursive computation, the system function can be designed somewhat similar to a recursive loop system. The switch between the two output components indicating termination now results from the detection of the zero of the $g$-part of the system function.

\begin{proposition}
Be $g\in F_{n+1}$ the computable function of rule $\mu$-Rec. Be ${\mathcal S}$ an MIS with 2 input signal components $in = (in_1, in_2)$, two output signal components $out=(out_1, out_2)$, two internal state signal components $q=(q_1, q_2)$, the feedback condition $out_1 = in_2$, and the following system function

\begin{eqnarray}
f_{\mathcal{S}}^{int}(in, q) & = & \begin{cases}
    q_2 = 0: & (in_1, 1)\\
    q_2 > 0: & (q_1, q_2+1)
  \end{cases} \\
f_{\mathcal{S}}^{ext}(in, q) & = & \begin{cases}
    q_2 = 0:                    & (g(in_1, q_2), \epsilon)\\ 
    (q_2 > 0) \wedge (in_2 \neq 0): & (g(q_1, q_2), \epsilon)\\ 
    (q_2 > 0) \wedge (in_2 = 0):     & (\epsilon, q_2-1)
  \end{cases}
\end{eqnarray}.

With the initial values $q_2 = 0$ and $in_2(0) = out_1(0) = \epsilon$ it realises the calculation according to the third rule $\mu$-Rec. I call such a system also a $\mu$-recursive system.  
\end{proposition}

The proof is similar to primary recursive computation. 

\subsection{Implicit recursive systems} \label{ss_implizit_rekursive_systeme}
A primary recursive system is already created, if the output of a simple system is feed back into an MIS that did just provide its input, as I have indicated in Fig. \ref{fig_system_loop_implicit}.

\begin{figure}[htbp]
  \begin{center}
    \includegraphics[width=7cm]{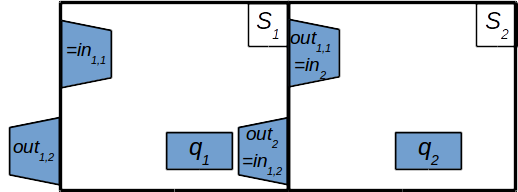}\\
 \end{center}
\caption[]{Diagram of an implicit recursive system.} \label{fig_system_loop_implicit}
\end{figure}

\begin{proposition}
The supersystem of Fig. \ref{fig_system_loop_implicit} is equivalent to a primary recursive system with one recursion step. 
\end{proposition}

\begin{proof} 
Be ${\mathcal S}_1$ an MIS and ${\mathcal S}_2$ a simple system. Then the supersystem  of Fig. \ref{fig_system_loop_implicit} can be constructed as follows:

The initial time of ${\mathcal S}_1$ is $\tau_1$, the initial time of ${\mathcal S}_2$ is $\tau_2$. As I assume $in_{1,2}(\tau_1) = out_2(\tau_2) = in_{1,1}(\tau_1') = \epsilon$, I can drop these components from the computation (e.g. it thereby holds $f_1(in_1(\tau_1), q_1(\tau_1)) = f_1(in_{1,1}(\tau_1), q_1(\tau_1))$, etc.).

I then have the following three computation steps:

\begin{tabular}{l|rcl}
Step \\
\hline
1. & $q_1(\tau_1')$ & $=$ & $f_1^{int}(in_{1,1}(\tau_1), q_1(\tau_1))$\\
   & $in_2(\tau_2) = out_{1,1}(\tau_1')$ & $=$ & $f_1^{ext}(in_{1,1}(\tau_1), q_1(\tau_1))$\\
2. & $q_2(\tau_2')$ & $=$ & $f_2^{int}(in_2(\tau_2), q_2(\tau_2))$\\
   & $in_{1,2}(\tau_1') = out_2(\tau_2')$ & $=$ & $f_2^{ext}(in_2(\tau_2), q_2(\tau_2))$\\
3. & $q_1(\tau_1'')$ & $=$ & $f_1^{int}(in_{1,2}(\tau_1'), q_1(\tau_1'))$\\
   & $out_{1,2}(\tau_1'')$ & $=$ & $f_1^{ext}(in_{1,2}(\tau_1'), q_1(\tau_1'))$ 
\end{tabular}

I summarise these three steps by defining $a:=(in_{1,1}(\tau_1), q_1(\tau_1), q_2(\tau_2))$ and $g(a) := f_1^{ext}(in_{1,1}(\tau_1), q_1(\tau_1))$. Then I have
\begin{eqnarray*}
in_{1,2}(\tau_1') = out_2(\tau_2') & = & f_2^{ext}(in_2(\tau_2), q_2(\tau_2))\\
 & = & f_2^{ext}(f_1^{ext}(in_{1,1}(\tau_1), q_1(\tau_1)), q_2(\tau_2))\\
 & = & f_2^{ext}(f_1^{ext}(a), a)\\
 & =: & h_0(a, g(a)) =: h(a,0,g(a))
\end{eqnarray*}

And with $h_1(x,y) := f_1^{ext}(x, f_1^{int}(y))$ I further have
\begin{eqnarray*}
out_{1,2}(\tau_1'') & = & f_1^{ext}(in_{1,2}(\tau_1'), q_1(\tau_1'))\\
 & = & f_1^{ext}(h_0(a, g(a)), f_1^{int}(in_{1,1}(\tau_1), q_1(\tau_1)))\\
 & = & f_1^{ext}(h_0(a, g(a)), f_1^{int}(in_{1,1}(a))\\
 & = & h_1(h_0(a, g(a)),a) = h(a,1,h(a,0,g(a)))
\end{eqnarray*}

\end{proof}

The counter of this primary recursive system exists only implicitly. 
If the system ${\mathcal S}_1$ ''calls'' further simple systems in the same manner as ${\mathcal S}_2$, then we get more recursive steps. 

This construct is quite relevant for the operation of modern imperative programming languages. In fact, imperative programs that consist of sequences of operations describe implicit primary recursive systems.

%
\section{Interactive systems} \label{s_interactive_systems}
%
In the last sections we investigated the compositional behaviour of systems whose system functions were known. That is, we investigated how supersystems are created by interactions of systems whose transformational I/O-behaviour we knew completely. 

In this section, we want to investigate systems whose behaviour in an interaction is nondeterministic, meaning that their behaviour is not solely determined by the interaction itself, but that they may have a decision-making scope within the interaction. I call these systems ''{\it interactive}''. Just as before, I will analyse these systems according to their composition behaviour. 

As said in the introduction the adequate description of these systems is key to fully understand interaction networks between complex systems as illustrated in Fig. \ref{fig_network_of_relations}. While we could simplify our hierarchical system composition by eliminating the interaction-perspective from our system description, this is no longer the case for interaction networks between interactive systems. Instead, we have to look for a way to describe both, the nodes, that is, the participants, and the edges, that is, the interactions of such an interaction network in a mutually compatible or ''balanced'' way. 

To achieve this goal, the idea is to partition systems in a way that we can use the resulting parts in two different compositions: an external composition, which I call ''cooperation'', and in an internal composition, which I call ''coordination''. I call these versatile system parts ''roles''. Put it a bit differently, the fact, that we are looking for a ''loose coupling scheme'' in the sense of a composition mechanism of systems that does not encompasses the systems as a whole but only their relevant parts in an interaction, logically entails another composition mechanism which composes these ''interaction relevant parts'' back to whole systems. Thus, the question how interactive systems compose directly leads us to the two questions: 
\begin{enumerate}
\item External composition or ''cooperation'': How do different roles of different systems compose by interaction?
\item Internal composition or ''coordination'': How do different roles of the same system compose by coordination?
\end{enumerate}

To express this mathematically, we have to extend our composition relation for systems (\ref{eq_system_composition}) to roles ${\mathcal R}_i$. But what do roles compose to? For the internal composition, it is obviously the system ${\mathcal S}$ all roles ${\mathcal R}_i^{\mathcal S}$ belong to. We therefore can state mathematically:

\begin{equation} \label{eq_role_composition_coordination}
{\mathcal S} = comp^{Coord}_{\mathcal R}({\mathcal R}^{\mathcal S}_1, \dots, {\mathcal R}_n).
\end{equation}

But, what is the resulting entity of the external composition? The concept of this composition entity is well known in informatics: it's named ''protocol''.  According to Gerald Holzmann \cite{Holzmann1991} this term was first used by R.A. Scantlebury and K.A. Bartlett \cite{Scantlebury1967}. We therefore can say that a protocol ${\mathcal P}$ in this sense is the composition result of roles ${\mathcal R}_{\alpha_i}^{{\mathcal S}_i}$ of different interactive systems ${\mathcal S}_i$, interacting in a nondeterministic way. And we can say that the roles are their interfaces in these compositions. Mathematically we write

\begin{equation} \label{eq_role_composition_interaction}
{\mathcal P} = comp^{Interact}_{\mathcal R}({\mathcal R}_{\alpha_1}^{{\mathcal S}_1}, \dots, {\mathcal R}_{\alpha_n}^{{\mathcal S}_n}).
\end{equation}

Obviously, both compositions are emergent compositional as the composition of the roles is not a role but either a system or a protocol. In the following I will delineate these concepts in detail.

\subsection{The definition of interactive systems} \label{ss_def_interaktive_systeme}
The program to derive the concept of interactivity consists of several steps. First, I introduce the notion of an interaction partition of an MIS. It partitions the behaviour of the MIS into different, interaction related behaviours, that I also call ''roles''. I show that these roles together with a mechanism of ''internal coordination'' is functionally equivalent to the original MIS. If this equivalent system formulation additionally satisfies the interactivity condition, then the original system is ''interactive''. In Fig. \ref{fig_system_inner_coupling} I illustrate this idea of a system that coordinates multiple roles in several, distinct interactions.

\begin{figure}[htbp]
  \begin{center}
    \includegraphics[width=5cm]{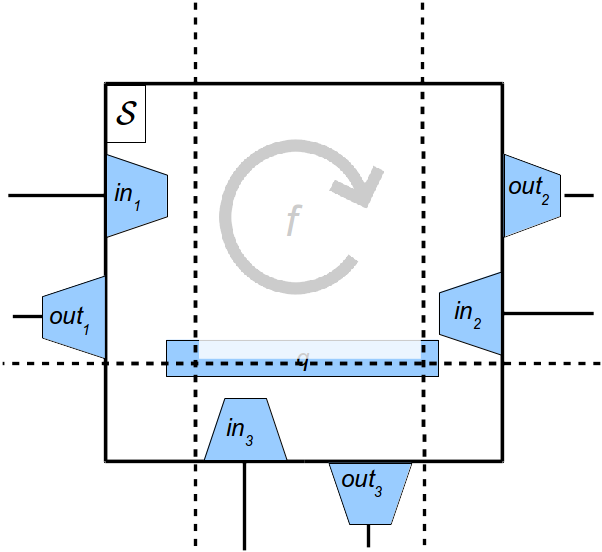}\\
 \end{center}
\caption[]{A single system ${\mathcal S}$ that coordinates multiple roles in different interactions. It shows the need to describe the inner coupling of the different roles, symbolised as a circular arrow over the greyed out central system parts.} \label{fig_system_inner_coupling}
\end{figure}

To describe an interaction we start by determining the relevant input and output components of the involved systems. In a first step, we must rearrange the input and output components of an MIS into interaction related pairs. For example, an MIS ${\mathcal S}$ has 4 input components with the input alphabet $I_{\mathcal S} = I_1 \times I_2 \times I_3 \times I_4$ where the input components 1, 2, and 4 are assigned to a first interaction and the input component 3 is assigned to a second interaction, then $I_1' = I_1 \times I_2 \times I_4$ and $I_2' = I_3$. The output alphabets have to be assigned similarly. Both, the input as well as the output alphabet of an interaction might remain empty. 

How are the transitions divided? If both, the non-empty component of the input as well as the output character of a transition belongs to an interaction, then the whole transition with its start and target state value belongs to the interaction. If the non-empty component of the input and the output character of a transition belongs to two different interactions, then this transition has to be divided into two separate transitions, one for each interaction together with a coordinating rule which relates the receiving transition of the one interaction to the sending transition of the other interaction. 

Now we have to deal with the following problem: Assume that we have several transitions from start state $p$ to target state $q$ mapping different input characters uniquely to some output characters, say $t_1 = (i_1, o_1, p, q)$ and $t_2 = (i_2, o_2, p, q)$.
Then the proposed division of the transitions results in an information loss. Once the first interaction arrives in state value $q$ under any input, but no output, the information which output should be generated in the ensuing transition of the second interaction is lost. If the resulting role-based product automaton should represent the same behaviour as the original system, then the inner state of the first interacting role has to be extended by its input alphabet. See Fig. \ref{fig_inner_coupling}.

\begin{figure}[htbp]
  \begin{center}
    \includegraphics[width=6cm]{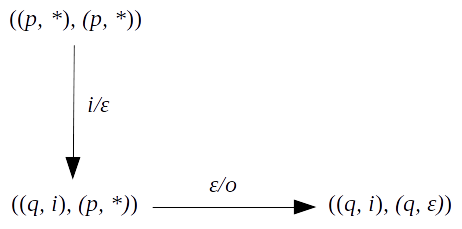}
 \end{center}
\caption[]{A transition $t=(i,o,p,q)\in \Delta$ whose input and output character are assigned to two different interactions is generally split into two transitions, each belonging to two orthogonal transition relations. In the first transition, $(i, \epsilon, (p, *), (q, i))$ of the first interaction relation (vertical dimension), the input character is processed and stored. In the ensuing second transition, $(\epsilon, o, (p, *), (q, \epsilon))$ of the second interaction (horizontal dimension), the output character is produced.} \label{fig_inner_coupling}
\end{figure}

As an example, let us consider a system that has two inputs and two outputs and simply reproduces input 1 at output 1 and input 2 at output 2, illustrated in Fig. \ref{fig_two_systems_with_different_IO-relations}. Is such a system interactive? That depends on the combination of inputs and outputs we assign to the roles. If we combine input 1 and output 1 as well as input 2 and output 2 into two roles each, then we have in fact two trivial stateless, independent deterministic systems without any internal coordination. It just holds $out_1' = in_1$ and $out_2' = in_2$ or in the automata-notation $p \stackrel{c/c}{\rightarrow}_{\Delta} p$ for every character $c\in I_i = O_i$, $i=1,2$. Clearly such a system would not be interactive. 

\begin{figure}[htbp]
  \begin{center}
    \includegraphics[width=6cm]{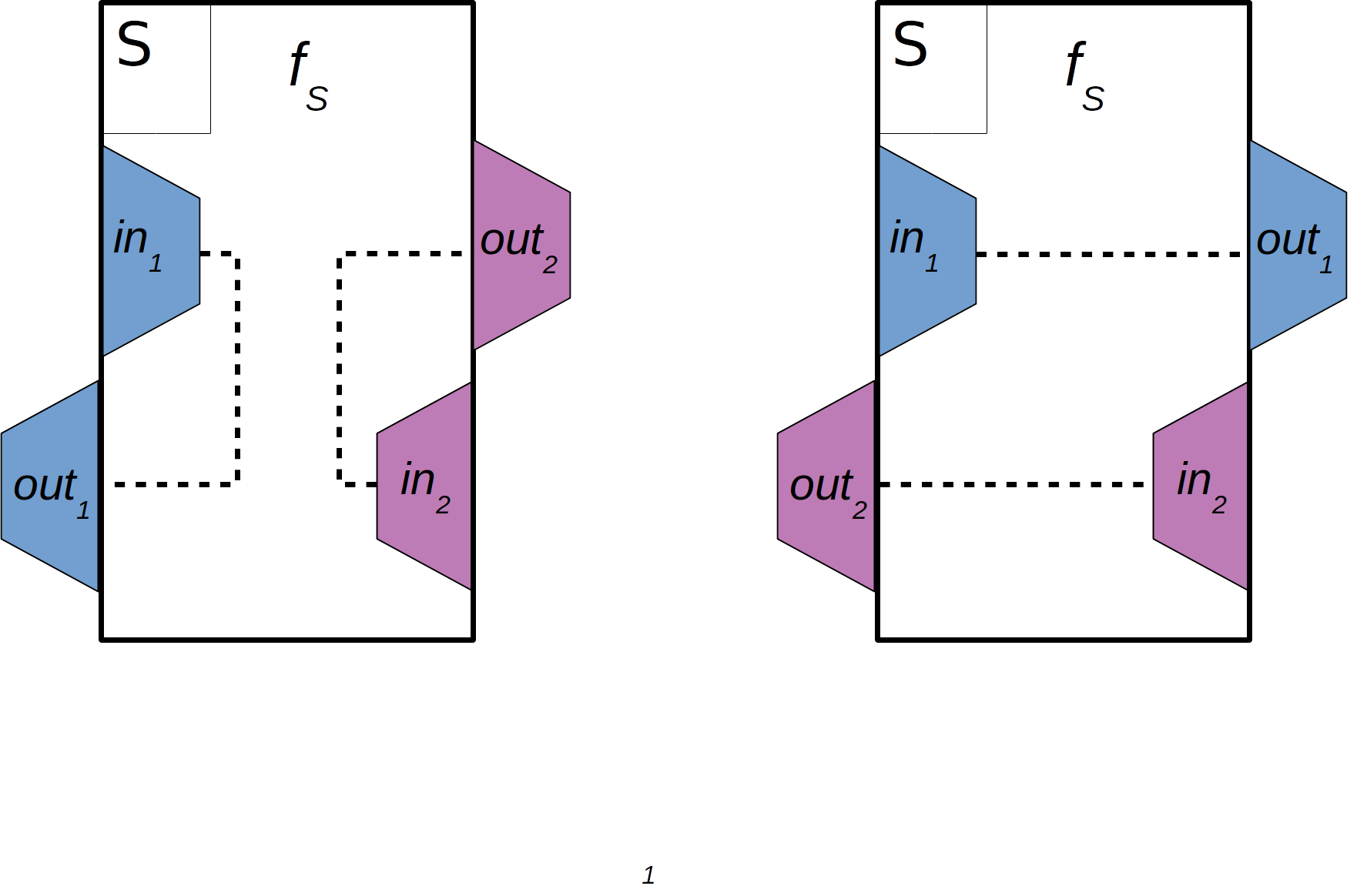}
 \end{center}
\caption[]{A single simple MIS with an identity relation between its two input and output states, symbolised by the dashed lines and colours, is drawn two times: on the left, the related input and output states are involved in a single interaction; on the right, they are involved in separate interactions} \label{fig_two_systems_with_different_IO-relations}
\end{figure}

If, however, we combine input 1 with output 2 to form a role and input 2 with output 1 to form the other role, then things look different. Then this system becomes a ''man-in-the-middle'' in which role 2 reproduces the behaviour of role 1 and vice versa. To create an equivalent role-based description of this system, we have to partition its transition relation according to the desired role structure, store the inputs temporarily and introduce the coordination. Receiving a character $c$ in state $p$, all we know is that the role does not create any immediate output. Thus, for every incoming character $i\in I$ we have $(p,*) \stackrel{i/\epsilon}{\rightarrow} (p, i)$. And if the role outputs a character, it seems to be produced spontaneously: $(p,c) \stackrel{\epsilon/o}{\rightarrow}_{\Delta} (p, \epsilon)$. Both roles are internally coordinated in the sense that any receiving transition in one role is inherently connected to an ensuing corresponding sending transition of the other role.  
	
This way an excellent chess playing system can be realised by playing simultaneously against two grandmasters, playing black against the first and white against the other one. This system only has to reproduce the moves that the Grand Master in white plays against its role 1 in its role 2 and the answer of the Grand Master in black in his role 1. I would call such a system interactive. Interestingly, the interaction description of each role does not allow to decide whether this system knows the ''rules of the game'', as in this case, this property obviously depends on the behaviour of each single grandmaster.  

Formally, I construct an interaction partition of an MIS as follows:

\begin{definition} \label{def_interaktionspartition}
Let ${\mathcal A}$ describe the behaviour of a system according to Def. \ref{def_IO-transition_system_mit_epsilon} with the input alphabet $I_{\mathcal A} = I_1\times \dots \times I_m$ and the output alphabet $O_{\mathcal A} = O_1\times \dots \times O_n$.

[Determination of the structure] An  ''interaction partition'' ${\mathcal P}$ of ${\mathcal A}$ is a pair ${\mathcal P} = (\{{\mathcal A}_1, \dots, {\mathcal A}_{N}\}, Coord)$ consisting of a set of interaction related behavioural descriptions $\{{\mathcal A}_1, \dots, {\mathcal A}_N\}$, also called ''roles'' and a coordination set $Coord$ representing the inner coupling of the roles. $Coord$ is a set of 4-tuple $({\mathcal A}_{1,j}, t_{1,j}, {\mathcal A}_{2,j}, t_{2,j})$, stating that when the transition $t_{1,j}$ is executed, processing a true input character without output in its interaction role ${\mathcal A}_{1,j}$, then the transition $t_{2,j}$ of the interaction role ${\mathcal A}_{2,j}$, which is without input but with a true output character, has to be executed. I also call $Coord$ a ''coordination set of type 1''.

[Determination of the number of roles and assigned characters] The input and output alphabets are combined to $N$ interaction-related pairs: $(I'_1, O'_1), \dots (I'_N, O'_N)$, which represent the starting point of the $N$ roles. Both the input and output components can also be empty. I introduce the index-functions $ind^{(I)}$ and $ind^{(O)}$ to assign each alphabet component of ${\mathcal A}$ its role number $\alpha$ (or $\beta$) and its new index  $j'$ (or $k'$) within its role respectively: $(\alpha, j') = ind^{(I)}(j)$ and $(\beta, k') = ind^{(O)}(k)$. The partial functions $\pi^{(I)}_\alpha:I_{\mathcal A} \rightarrow I_\alpha$ and $\pi^{(O)}_\beta:O_{\mathcal A} \rightarrow O_\beta$ map the characters of the input and output of ${\mathcal A}$ uniquely to the input and output characters of the roles. 

[The set of state values of a role:] The set of state values of a role $Q_{\alpha}$ is the smallest possible subset of $Q_{\mathcal A} \times I_{_\alpha}$ to fulfil the to be constructed transition relation of ${\mathcal A}_\alpha$.

[Inductive splitting of the transitions] Be $t=(i,o,p,q)\in \Delta_{\mathcal A}$ with $i=\epsilon_{I_{\mathcal A}}[c^{(I)}, j] \in I_{\mathcal A}$, $o=\epsilon_{O_{\mathcal A}}[c^{(O)}, k] \in O_{\mathcal A}$, and $p,q\in Q_{\mathcal A}$ and be $(\alpha, j') = ind^{(I)}(j)$ and $(\beta, k') = ind^{(O)}(k)$.

Let further be $i_\alpha = \pi^{(I)}_\alpha(i) = \epsilon_{I_\alpha}[c^{(I)}, j']$ and $o_\beta = \pi^{(O)}_\beta(i) = \epsilon_{O_\beta}[c^{(O)}, k']$ the component-wise reordered input and output characters with respect to the component-wise reordered alphabets.

The beginning of induction starts at the initial state. Initially, the part of the inner state that refers to the input characters is irrelevant and therefore can be arbitrarily chosen. In the induction step, this part is equal to the last input character. To abbreviate the description, I use a $*$ at this position.

In case $\alpha = \beta$ the input and output character are assigned to the same role. Then $(p, *), (q, \epsilon) \in Q_\alpha$ and the transition $t' = (i_\alpha, o_\alpha, (p, *), (q, \epsilon)) \in \Delta_\alpha$.   

In case $\alpha \neq \beta$, input and output take place in two different roles. Then $(p, *), (q, i_\alpha) \in Q_\alpha$, and $(p, *), (q, \epsilon) \in Q_\beta$, as well as the transitions $t_\alpha = (i_\alpha, \epsilon, (p, *), (q, i_\alpha)) \in \Delta_\alpha$ and $t_\beta = (\epsilon, o_\beta, (p, *), (q, \epsilon)) \in \Delta_\beta$. In addition, the coordination relationship applies that after $t_\alpha$ $t_\beta$ is to be executed, i.e. $({\mathcal A}_\alpha, t_\alpha, {\mathcal A}_\beta, t_\beta) \in Coord_{\mathcal P}$.

$Coord$ contains all data corresponding to the decomposed transitions. This means that the information about the coordination of the roles is, in fact, the information about the cross-role transitions. 

\end{definition}
 
With this notion of an interaction partition, we can now define a system as being interactive: 

\begin{definition} \label{def_interaktives_verhalten}
Be ${\mathcal A}$ a behavioural description of a system ${\mathcal S}$ according to Def. \ref{def_IO-transition_system_mit_epsilon} and ${\mathcal P}$ an interaction partition of ${\mathcal A}$ with at least 2 roles.
${\mathcal A}$ (and thereby ${\mathcal S}$) is ''interactive'' with respect to ${\mathcal P}$, if the transition relation of each role of ${\mathcal P}$ is nondeterministic and stateful and the coordination set of type 1 $Coord$ is non-empty. 
\end{definition}

\subsection{Decisions} \label{ss_entscheidungen}
Obviously, in a nondeterministic interaction, the interaction does not (completely) determine the actions of the systems involved. The processing of the received characters is therefore not completely determined by the interaction itself. A second approach for the completion of the role description to a system function, beside coordination, is the concept of decisions: decisions determine the behaviour, i.e. the transitions, where it would otherwise be indeterminate. 

According to the two mechanisms that give rise to nondeterminism of transitions, we can distinguish two classes of decisions: spontaneous decisions that determine the spontaneous transitions and selection decisions that determine a selection. 

Decisions in this sense can therefore be seen as a further, ''inner'' input alphabet $D$, which complements the input alphabet $I$ of an I/O-TS ${\mathcal A}$ according to Def. \ref{def_IO-transition_system_mit_epsilon} to $I' = I\times D$ such that a complemented transition relation $\Delta'$ becomes deterministic.

Decisions in this sense are very similar to information. They are enumerated by an alphabet and their names are relevant only for their distinction. In contrast to ordinary input characters, whose main characteristic is to appear in other output alphabets and that are allowed to appear in different transitions, we name all decisions of a corresponding transition system differently and different from all input and output characters, so that we can be sure that they really do determine all transitions.	

\begin{definition} \label{def_decision} 
Be ${\mathcal A}$ an I/O-TS and $D$ an alphabet. We call the transition system ${\mathcal A}'$ a ''decision system'' to ${\mathcal A}$ and the elements of $D$ ''decisions'', if $I\cap D = \emptyset$, $O\cap D = \emptyset$, $Q\cap D = \emptyset$,  and $\Delta' \subseteq (I_{\mathcal A}^\epsilon \times D) \times O_{\mathcal A}^\epsilon \times Q_{\mathcal A}\times Q_{\mathcal A}$ with $((i, d), o, p, q) \in \Delta'$ if $(i, o, p, q)\in \Delta$ and for $d$ applies:

\[
  d=\begin{cases}
    \epsilon, & \text{if there's no further transition $(i^*, o^*, p^*, q^*)\in \Delta$}\\
              & \text{with $(i,p) = (i^*, p^*)$.}\\
    \mbox{so selected} & \text{that $\Delta'$ is deterministic, i.e. $\Delta'$ determines the function}\\
              & \text{$f': I^\epsilon \times D \times Q \rightarrow O^\epsilon \times Q$ with $(o, q) = f'(i, d, p)$.}\\
              & \text{For two transitions $t'_1, t'_2 \in \Delta'$ it holds $t'_1 \neq t'_2 \Rightarrow d_1 \neq d_2$.}\\
              & \text{Additionally, $\Delta'$ is the smallest possible set.}\\
  \end{cases}
\]
\end{definition}

Obviously, the set of decisions for an already deterministic I/O-TS is empty. I.e. a system completely determined in an interaction, which only executes its function, makes no decisions in this sense. We can therefore say that the decision concept serves to describe the interaction partners completely despite incomplete knowledge in the sense of an I/O-mapping. 

In a further step one can now consider how a system comes to its decisions. For example, whether one can calculate the decisions, what role other interactions may play or any heuristics. This is where game theory and computer science meet. 

\subsection{External composition or ''cooperation''}\label{ss_cooperation}
Until now, our starting point still has been the system as a whole. Now we change our perspective and investigate the interaction between the systems as is illustrated in Fig. \ref{fig_system_cooperation}. We now look in more detail into how cooperation in the sense of this article works. What makes cooperation special? Can we view every interaction as a cooperation? 

\begin{figure}[htbp]
  \begin{center}
    \includegraphics[width=8cm]{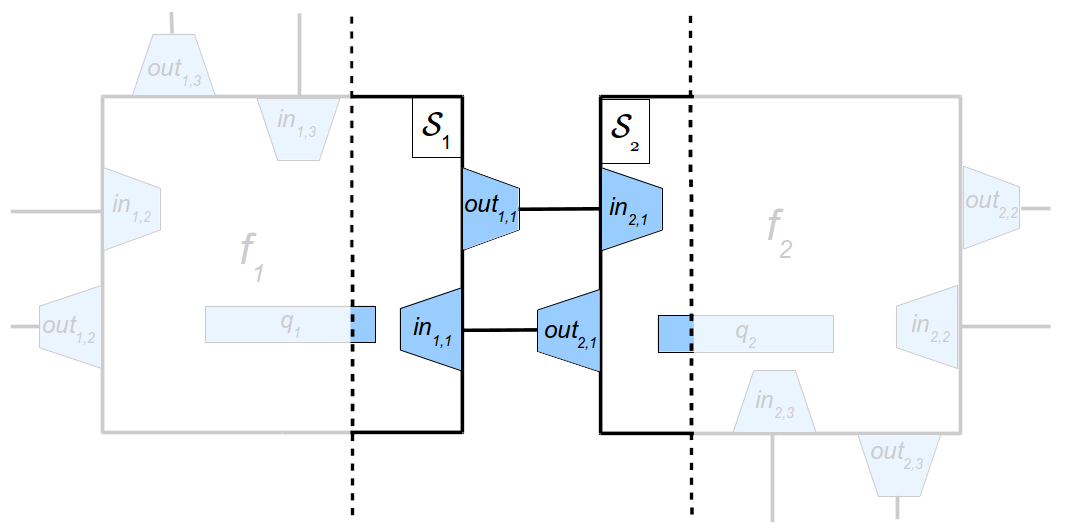}\\
 \end{center}
\caption[]{Two interactive systems ${\mathcal S}_1$ and ${\mathcal S}_2$ interact with each other and with multiple other systems. Within an interaction only the projection of the system functionality on the interaction channels becomes directly visible.} \label{fig_system_cooperation}
\end{figure}

\subsubsection{Outer coupling: channel based restriction (CBR)\label{ss_outer_coupling}}
It is to be expected that exploring our new area, we have to modify our tools, namely the transition system according to Def. \ref{def_IO-transition_system_mit_epsilon}. This is indeed the case.


The first thing we have to note is that in nondeterministic interactions -- in contrast to hierarchical system composition, where the functions of the subsystems naturally composed to functions of the supersystem - there is no longer any genuine purpose. But, as we still want to have a criterion for its success, we have to introduce it, leading us from I/O-TSs to I/O-automata (I/O-A)\footnote{Another name in the literature is ''transducer'', because this machine translates an input string into an output string. See for example \cite{Sakarovitch2010} for an extensive introduction.}

\begin{definition} \label{def_NFIOA} 
An I/O-automaton (I/O-A) is an I/O-TS with an additional acceptance component $Acc$. I.e. an I/O-A is a tuple ${\mathcal A} = (Q, I, O, (q_0, o_0), Acc, \Delta)$ with $Q$ being a nonempty set of state values, $I$ and $O$ are possibly empty input and output alphabets,  $(q_0, o_0)$ is the initial state, $Acc$ is the acceptance component and $\Delta \subseteq I^\epsilon\times O^\epsilon \times Q \times Q$ is the transition relation. If $Q$, $I$ and $O$ are finite, then we call the automaton also finite (fI/O-A).

The acceptance component depends on the success model. For finite calculations with a desired end, $Acc_{finite}$ consists of the set of final state values. For infinite calculations of a finite automaton there are differently structured success models. One of them is the so-called Muller acceptance, where the acceptance component is a set of subsets of the state value set $Q$, i.e. $Acc_{Muller} \subseteq \wp(Q)$. A run is considered accepted whose finite set of infinitely often traversed state values is an element of the acceptance component (e.g. \cite{DBLP:conf/dagstuhl/Farwer01}).
\end{definition}

In the case that for every $(p, i) \in Q \times I^\epsilon$ at most one transition $(i, o, p, q) \in \Delta$ exists, $\Delta$ specifies a function $\delta: Q \times I^\epsilon \rightarrow Q \times O^\epsilon$ with $(q,o) = \delta(p, i)$. If no $\epsilon$ is used, i.e. $\delta: Q \times I \rightarrow Q \times O$, then we have a deterministic I/O automaton (dI/O-A). A deterministic finite I/O-A is also called a Mealy automaton \cite{Mealy1955} in the literature.

The next point we have to consider is the actual coupling mechanism. In hierarchical system composition the fact that we never left the context of our system of consideration, the created supersystem, implied the possibility to eliminate the interaction mediating states. But now, we deal with the coupling of at least partly independent systems. For the interaction of interactive systems, the coupling states --- or Shannon channels --- take a special significance as they restrict the selectable transitions of the receiver accordingly. I call this the ''channel-based restriction (CBR)'' of the receiver's transition relation (See fig. \ref{fig_coupling_by_interaction}). In the case of in- and output with multiple components between two systems ${\mathcal U}$ and ${\mathcal V}$, a Shannon channel is characterised by two indices $(k_{\mathcal U},l_{\mathcal V})$, where the first indexes the output component of ${\mathcal U}$ and the second the input component of ${\mathcal V}$. 
We must take this -- external -- coupling into account and have to add a set of Shannon channels to our automata structure. I name an I/O-A together with a set of Shannon channels a CBR-automaton, or CBR-A.

We also have to take into account the new possibility of entirely spontaneous state transitions. If the automaton offers a spontaneous transition, this can be processed instead of a current input character. This means that used as an input character, the empty character is not offered from the outside, but marks transitions that can take place independently of external inputs. 

To consider these two new requirements, we have to adapt our calculation rule of Def. \ref{def_MIS_execution} for our CBR-As in the selection procedure of the character $i^*$ to be currently processed. 

\begin{definition} \label{def_CBR_execution}
Let ${\mathcal P} = ({\mathcal A}, C)$ a CBR-A, consisting of an I/O-A  ${\mathcal A}$ with $I_{\mathcal A} = (I_1, \dots, I_{n_I})$ and $O_{\mathcal A} = (O_1, \dots, O_{n_O}) $ and a set ${\mathcal C}$ of coupling signals that connect different output components with input components of ${\mathcal A}$. Furthermore, be $seq = (i_1, \dots, i_n)$ a (possibly infinite) sequence of input characters from $I_{\mathcal A}$ from the ''outside'', which are processed in sequence and are all of the form $\epsilon[v,k]$, i.e. they are  unequal to the empty character $\epsilon$ in exactly one component.

The variable $j = 0\dots n-1$ counts the sequence of the external input character. Due to the possibility of spontaneous transitions as well as the feedback of self-generated output the (not counted) number of the executed transitions can deviate from $j$.
The current values of $i$, $o$ and $q$ are indicated by a '$*$', the values calculated in the current step by a '$+$'.

\begin{enumerate} 
\item {\bf Initialisation (time $j=0$):} $(q^*, o^*) = (q_0, o_0)_{\mathcal A}$.

\item {\bf Loop:} \label{calc_IOA_loop} Determine for the current state $q^*$ as start state the set of all possible transitions. If this set is empty, end the calculation.

\item {\bf Determine the current input character $i^* \in I_{\mathcal A}^\epsilon$:} \label{destine_input_marks} Proceed after the following sequence:
\begin{enumerate}
  \item If the current output character has the value $v\neq\epsilon$ in its $k$-th component, i.e. $o^*=\epsilon[v,k]$, and $o^*$ is part of a feedback signal $c=(k,l)$ to the input component $0 \leq l \leq n_I$, then set $i^* = \epsilon_I[v,l]$.
  \item If there is no input character $i_j$ at the $j$-th time, but there is a spontaneous transition for $q^*$, select $i^* = \epsilon$ as the current input character.
  \item If there is an input character $i_j$ and a spontaneous transition to $q^*$ exists, you can choose between $i_j$ and $\epsilon$ for $i^*$. 
  \item If there is only an external input character without a spontaneous transition to $q^*$, $i^* = i_j$ must be selected. 
  \item If there is neither an external input character nor a spontaneous transition, then finish the calculation.
\end{enumerate}

\item {\bf Transition:} \label{calc_IOA_transition} With $q^*$ as current state value and $i^*$ as current input character select a transition $t=(i^*, o^+, q^*, q^+)\in\Delta_{\mathcal A}$ and so determine $o^+$ and $q^+$. If there is no possible transition at this point, terminate the calculation with an error.

\item {\bf Repetition:} If an external input character was chosen, go over to the next (external) time $j+1$. Set $q^* = q^+$ and $o^* = o^+$ and jump back to \ref{calc_IOA_loop}
\end{enumerate}
\end{definition}

In the deterministic case, an input sequence is considered accepted if the acceptance condition for the calculated run is fulfilled. In the nondeterministic case, an input sequence is considered accepted if the acceptance condition is fulfilled for at least one of the runs possible with this input sequence.  

\subsubsection{Protocols}\label{ss_protocols}
Now we want to describe the interaction of several interactive systems through their roles. In order to express this system coupling with our concept of the CBR-A, we have to consider all roles together as a product automaton in a first step and then restrict their transition relation by enforcing the coupling Shannon channels. 

I call the product of the I/O-As of all roles also a ''weakly synchronised product automata'' as this product automaton only realises the ''joint consideration'' of all interaction components as a whole in the sense of the reaction to a common input character, a synchronisation of the times steps, the production of a common output character and a common acceptance condition.

\begin{definition}\label{def_weakly_synchronized_product} 
The {\it weakly synchronised product} of a set of $n$ I/O-As $\{{\cal A}_1, \dots, {\cal A}_n\}$ is defined by the I/O-A ${\cal B} = (Q, I, O, (q_0, o_0), Acc, \Delta)_{\cal B}$, with 
\begin{itemize}
\item $Q_{\cal B} = \mbox{\Large $\times$} Q_k$.
\item $I_{\cal B} = \{i\in\mbox{\Large $\times$} I^\epsilon_k \,|\, \text{if}\, \epsilon\neq i_j \in I_j$ is the $j$-th component of $i$, then $i=\epsilon[i_j, j]$\}.
\item $O_{\cal B}$ corresponding to $I_{\cal B}$.
\item ${q_0}_{\cal B} = ({q_0}_1, \dots, {q_0}_n)$, ${o_0}_{\cal B} = \epsilon$.
\item The common acceptance condition is the AND-combination of the single conditions, relating to the product set of the single acceptance components: $Acc_{\mathcal B} = \mbox{\Large $\times$} Acc_k$.
\item $\Delta_{\cal B} = \{(p, q, i, o)|$ $p=(p_1, \dots, p_n)$ is an achievable state value and  ${\cal A}_k$ has a transition $(p_k, q_k, i_k, o_k)$ so that $q = p \left[\frac{q_k}{p_k}, k\right]$ and $i =\epsilon[i_k, k]$ and $o = \epsilon[o_k, k]$ $\}$. 
\end{itemize}
I also write ${\cal B} = \bigotimes_{i=1}^n {\cal A}_i$.
\end{definition}

This construction ensures that each output character of the product automaton is again in at most one component unequal to $\epsilon$.

Adding Shannon channels creates transition chains that are triggered either by external characters or spontaneously. The most interesting case, where the CBR-A no longer processes a character ''from the outside'' at all, leads us to the definition of a protocol in the sense of a closed interaction:  

\begin{definition}
A CBR-A in which all input and output states are connected via Shannon channels is called ''{\it closed}'', otherwise it is called ''{\it open}''. A closed CBR-A is also called a ''{\it protocol}''. 
\end{definition}

Thus a protocol is the desired description of the nondeterministic, asynchronous and stateful interaction of different interactive systems, which are only visible through their roles. 

Due to the closed nature of the protocol, each output character is also an input character. As long as the transition relation of the receiving roles provides an output character to an input character, the calculation of a run automatically proceeds to its next step. I call these sections of a protocol run, which are started by spontaneous transitions and are ended by transitions without output characters, a ''{\it chain of interaction}''.

It is not clear from the beginning that a calculation of a protocol run produces something useful. In fact, terrible things can happen: characters are sent that cannot be received (ill-formed). Unwanted states can occur such as when everyone is waiting for someone to do something (deadlock). Or unwanted infinitely long inner chains or endless loops (livelock) occur. I therefore define: 

\begin{definition}
A protocol is called ''{\it well-formed}'' if for each transition with a sent character $o$, which is unequal to $\epsilon$ in at least one component, a corresponding receiving transition exists. Otherwise it is called ''{\it ill-formed}''.
\end{definition}
        
This is a typical safety property in the sense that something undesirable will never happen. 

\begin{definition}
A well-formed protocol is called ''{\it interruptible}'' when every chain of interaction is finite. 
\end{definition}

\begin{definition}
A well-formed, interruptible protocol is called ''{\it consistent}'', if for each attainable state value the acceptance condition is either fulfilled or at least a continuation exists, so that the acceptance condition can be fulfilled.
\end{definition}

This is a typical liveness property in the sense that something desired will eventually happen.

According to the definition, a consistent protocol neither has deadlocks, i.e. non-final states without continuation, nor livelocks in the sense of periodically reached states without continuation according to the acceptance condition.  

Finally, I would like to point out that protocols and games are very closely related to each other by the already introduced concept of decision. The protocol roles can be complemented with decisions. And a decision automaton corresponding to a consistent protocol can be regarded as a ''game in interaction form'' in the sense of the game theory \cite{Reich2020_A_Theory_of_Interaction_Semantics}.

\subsubsection{Example: The single-track railway bridge}
Now I want to illustrate the protocol concept with the simple example of the single-track railway bridge drawn from \cite{Alur2015_Principles}, a simple version of the general problem of a shared single resource \cite{Holzmann1991}. As is shown in Fig. \ref{fig_trains}, two trains, ${\mathcal Z}_1$ and ${\mathcal Z}_2$, must share the common resource of a single-track railway bridge. For this purpose, both trains interact with a common controller ${\cal C}$, which must ensure that there is no more than one train on the bridge at any one time.

\begin{figure}[htbp]
\centering
\includegraphics[width=8cm]{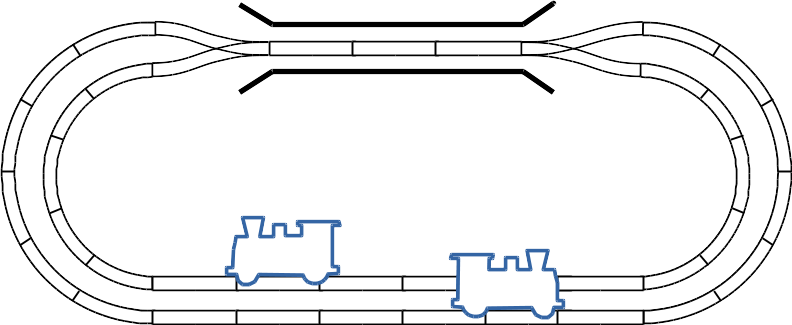} 
\caption{A single-track railway bridge crossed by two trains. To avoid a collision on the bridge, both trains interact with a central controller.}
\label{fig_trains}
\end{figure}

The interaction between each train and the controller is described by a protocol. For this we need to describe both the train and the controller in terms of the role they play in the interaction. For both trains and the controller I choose a model of 3 state values, which I call $Q_{Z_{1,2}/C}=\{away, wait, bridge\}$. The input alphabet of the trains $I_{{\cal Z}_{1,2}} = \{go\}$ is the output alphabet of the controller $O_{\cal C}$ and the output alphabet of the trains $O_{{\cal Z}_{1,2}} = \{arrived, left\}$ is the input alphabet of the controller $I_{\cal C}$. 

\begin{figure}[h!]
\centering
\includegraphics[width=8cm]{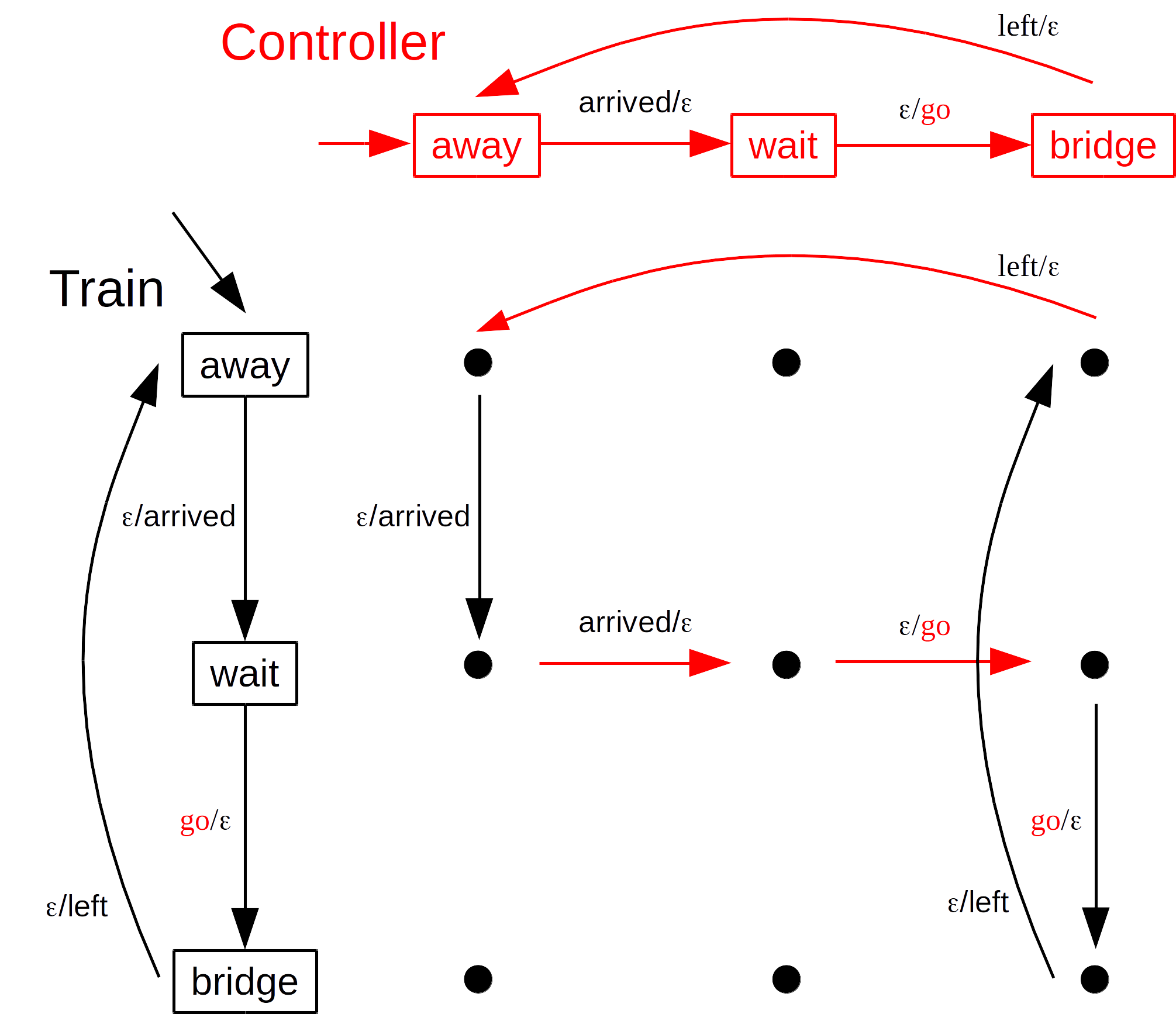}
\caption{Presentation of the protocol between train and controller for the problem of the single-track railway bridge.  Initially, both controller and train are in the $away$ state. When a train arrives, it signals $arrived$ to the controller. As a CBR automaton this character must now be processed by the controller, the controller in turn changes to its $wait$ state. The controller releases the track with $go$ and the train signals the controller with $away$ that it has left the bridge again.  }
\label{fig_trains_protocol}
\end{figure}

In Fig. \ref{fig_trains_protocol} the protocol is shown as a CBR-A of the two roles of train and controller, each in one dimension. Initially, both controller and the train are in the state $away$. When a train arrives, it signals $arrived$ to the controller. The controller releases the track with $go$ and the train signals the controller with $away$ that it has left the bridge again. The interaction is successful when both the train and the controller go through their three states infinitely often. 

The protocol between train and controller has all the characteristics we want. It is complete as no further external characters occur. It is well formed as for each sent character there is a processing transition at the right time. And finally, it is consistent as it has only finite interaction chains and it fulfils its acceptance condition. 

The correctness, we could also say the truth, of the representation of the state of the train in the controller depends on the correctness of the protocol. And, this train-controller protocol does not guarantee our original goal that at most one train uses the bridge at a time. This requires correct coordination of the controller.

\subsection{Internal composition or ''coordination''} 
Now I want to explore the second, the ''inner'' coupling of different roles of one system into that system, which I call ''{\it coordination}'' and which I already illustrated schematically in Fig. \ref{fig_system_inner_coupling}. This will also result in a method for synthesising interactive systems from their different roles.
And it will also pave the way for our considerations on the compositional behaviour of consistent protocols, or, in other words, how interactions can be composed into new interactions.

\subsubsection{Inner coupling: coordination as transition elimination} \label{ss_composition_of_interactions}
With the concept of interactive systems we already explored a coordination concept by partitioning a given MIS into its roles, namely the interaction partition together with the coordination set of type 1. This coordination set of type 1 contained the information about the cross-role transitions of the original system.

The concept of role coupling I will present now works differently and starts from the combination of all roles of one system as their weakly synchronised product automaton, still having all possible transition of the unrestricted product.
Then we impose our coordination by dropping all the transitions which do not contribute or even jeopardise the coordination goal, but leave the success criteria of the interactions invariant. I also call the involved roles the ''coordinated roles'' and the remaining transitions of the product automaton the ''coordination set of type 2'' and the omitted transitions the complement.

Transitions from a reachable state with a given input character can only be eliminated if other transitions to with this input character exist from that state. Otherwise, we would risk to ill-form the protocol. The same holds for spontaneous transitions without input character. Thus, as also the example in the next section will show, the nondeterminism of the interaction and especially the spontaneous transitions are of special importance for the creation of the inner coordination from the transition relations of the individual roles.

However, it is essential for the creation of coordination through transition elimination that even under the restricted transition relation of the product automaton the acceptance conditions of the individual roles, which continue to result as projections of the product automaton, can be fulfilled. In order to formalise these concepts I first introduce the projection of an I/O-A.

\begin{definition} \label{def_projected_automaton}
Be ${\mathcal B}$ an I/O-A and $\pi = (\pi_Q, \pi_I, \pi_O): Q_{\mathcal B}^\epsilon \times I_{\mathcal B}^\epsilon\times O_{\mathcal B}^\epsilon \rightarrow Q_{\mathcal B}^\epsilon \times I_{\mathcal B}^\epsilon \times O_{\mathcal B}^\epsilon$ a projection function.
Then the projected I/O-A ${\mathcal A} = \pi({\mathcal B})$ is determined by $Q_{\mathcal A} = \pi_Q(Q_{\mathcal B})$, $I_{\mathcal A} = \pi_I(I_{\mathcal B})$, $O_{\mathcal A} = \pi_O(O_{\mathcal B})$, ${q_0}_{\mathcal A} = \pi_Q({q_0}_{\mathcal B})$, $Acc_{\mathcal A} = \pi_Q(Acc_{\mathcal B})$, $\Delta_{\mathcal A} = \{(p',q',i',o')| (p, q, i, o) \in \Delta_{\mathcal B} \enspace\mbox{and}\enspace p'=\pi_Q(p),\enspace q'=\pi_Q(q),\enspace i'=\pi_I(i),\enspace o'=\pi_O(o)\}$.
\end{definition}

The projection of the condition of acceptance is to be understood such that the set on which the condition of acceptance is based is projected accordingly. Thus I introduce a ''coordinated automaton (CA)'', which differs from the original weakly synchronised product automaton only with respect to its transition relation, which is a true subset of the original one, but still fulfils ${\mathcal A}$'s acceptance condition. Please note that we have to include not just the coordinated roles, but the complete protocols into this definition, as we needed to evaluate the success criterion for the interactions for the correctness of our construction. Formally:  

\begin{definition} \label{def_coordinating_automaton}
Let $P=\{{\mathcal P}_1, \dots, {\mathcal P}_n\}$ be a set of $n$ consistent protocols, where the protocol ${\mathcal P}_i$ contains the stateful role ${\mathcal A}_i$.  
Let ${\mathcal A}= \bigotimes_{i=1}^n {\mathcal A}_i$ be a weakly synchronised I/O-A with the $n$ roles $\{{\mathcal A}_1,\dots, {\mathcal A}_n\}$.
$Z\subset \Delta_{\mathcal A}$ is the non-empty set of excluded transitions with at least one transition originated from each role transition relation $\Delta_i$.
And $\{\pi_1, \dots \pi_n\}$ are $n$ projection functions that map $\mathcal A$ to the respective roles ${\mathcal A}_i$.
Then ${\mathcal C}$ is a ''{\it coordinating automaton (CA)}'' to ${\mathcal A}$ with respect to the transition relation $\Delta_{\mathcal C} = \Delta_{\mathcal A}\backslash Z$ with nonempty $Z$ if otherwise $Q_{\mathcal C} = Q_{\mathcal A}$, $I_{\mathcal C} = I_{\mathcal A}$, $O_{\mathcal C} = O_{\mathcal A}$, $Acc_{\mathcal C} = Acc_{\mathcal A}$ and all roles ${\mathcal A}_i$ are preserved as a projection of ${\mathcal C}$, i.e. $\pi_i({\mathcal C}) = {\mathcal A}_i$, and the collective acceptance condition is also fulfilled by ${\mathcal C}$.
\end{definition}

The coordinating automaton (CA) used for the internal coupling has inputs and outputs that do not refer to each other, but serve the other roles of the protocols under consideration - in contrast to the CBR automaton that represented the external coupling.

Since the CA by definition leaves the properties of the subprotocols invariant, the overall interaction it controls is again a consistent protocol.

\begin{proposition} \label{prop_protocol_composition}
Let $P=\{{\mathcal P}_1, \dots, {\mathcal P}_n\}$ be a set of $n$ consistent protocols, where the protocol ${\mathcal P}_i$ contains the stateful role ${\mathcal A}_i$ and let ${\mathcal C}$ be a coordinating automaton (CA) to ${\mathcal A}= \bigotimes_{i=1}^n {\mathcal A}_i$ such that $\Delta_{\mathcal C} \subset \Delta_{\mathcal A}$. Then the CBR-A ${\mathcal P}$, consisting of the complex role ${\mathcal C}$ together with all remaining roles of the protocols ${\mathcal P}_i$ and channels is again a consistent protocol.  
\end{proposition}

\begin{proof}
Being a protocol, each of the ${\mathcal P}_i$ is a product IO-A together with a set of Shannon channels. All these protocols are independent and in particular they do not share any Shannon channels, that is their Shannon channel sets are pairwise disjoint.

Taking one role out of each protocol and creating a product automaton where we restrict the transition relation such that the CA ${\mathcal C}$ results, means that we then have again an I/O-A with the roles ${\mathcal A}_i$ which still fulfils all the acceptance conditions of its roles ${\mathcal A}_i$. Hence, together with the set of Shannon channels of all protocols ${\mathcal P}_i$, this is again a CBR-A and it is closed, which makes it a protocol, and it is also consistent, as the acceptance conditions of all roles are still fulfilled.  
\end{proof}

In fact, the procedure described for coordinating different interactions provides guidance on how to synthesize systems from their interaction roles (to a certain extent). This becomes clear in the next example of the controller for the operation of the single-track railway.

\subsubsection{The relation between coordination sets of type 1 and 2}

Now we can  have a closer look at the relationship between the coordination sets of type 1 and 2. The coordination set of type 1 contained the information of all transitions where the input and output of the transition of the original MIS were attributed to different roles. As its starting point was a full system, namely an MIS, it could not contain any nondeterminism. 

In contrast, the coordination sets of type 2, introduced in this section, still leaves room for nondeterminism. Thus we can express all coordination sets of type 1 as type 2 but not vice versa. The definition of a CA (\ref{def_coordinating_automaton}) therefore can be viewed as an extension of the definition of an interactive system (\ref{def_interaktives_verhalten}). All roles in a CA are stateful per requirement and must be nondeterministic, as otherwise no transition can be eliminated that originates from this role, which is required by the coordination set of type 2. 

To get a fully qualified system with a general CA, we must complement its remaining nondeterminism by decisions. Thus I state the following proposition:

\begin{proposition}
A system whose behaviour can be described by a decision automaton of a CA is interactive. 
\end{proposition}

\begin{proof}
To prove this proposition, we have to show that for such a system we can find an interaction partition with at least two nondeterministic and stateful roles together with a non-empty coordination set of type 1. The interaction partition is given by the roles of the CA. It remains to be shown that we can construct a coordination set of type 1 from the given coordination set of type 2 together with the decisions of the decision automaton. I leave this as an exercise. 
\end{proof}

\subsubsection{Example: Controller for the operation of the single-track railway bridge}
The controller, which coordinates the movements of the trains across the single-track bridge, is to be constructed from its roles of the respective interactions with the individual trains. This is shown in Fig. \ref{fig_trains_state_chart_controller}. 

\begin{figure}[htbp]
\centering
\includegraphics[width=8cm]{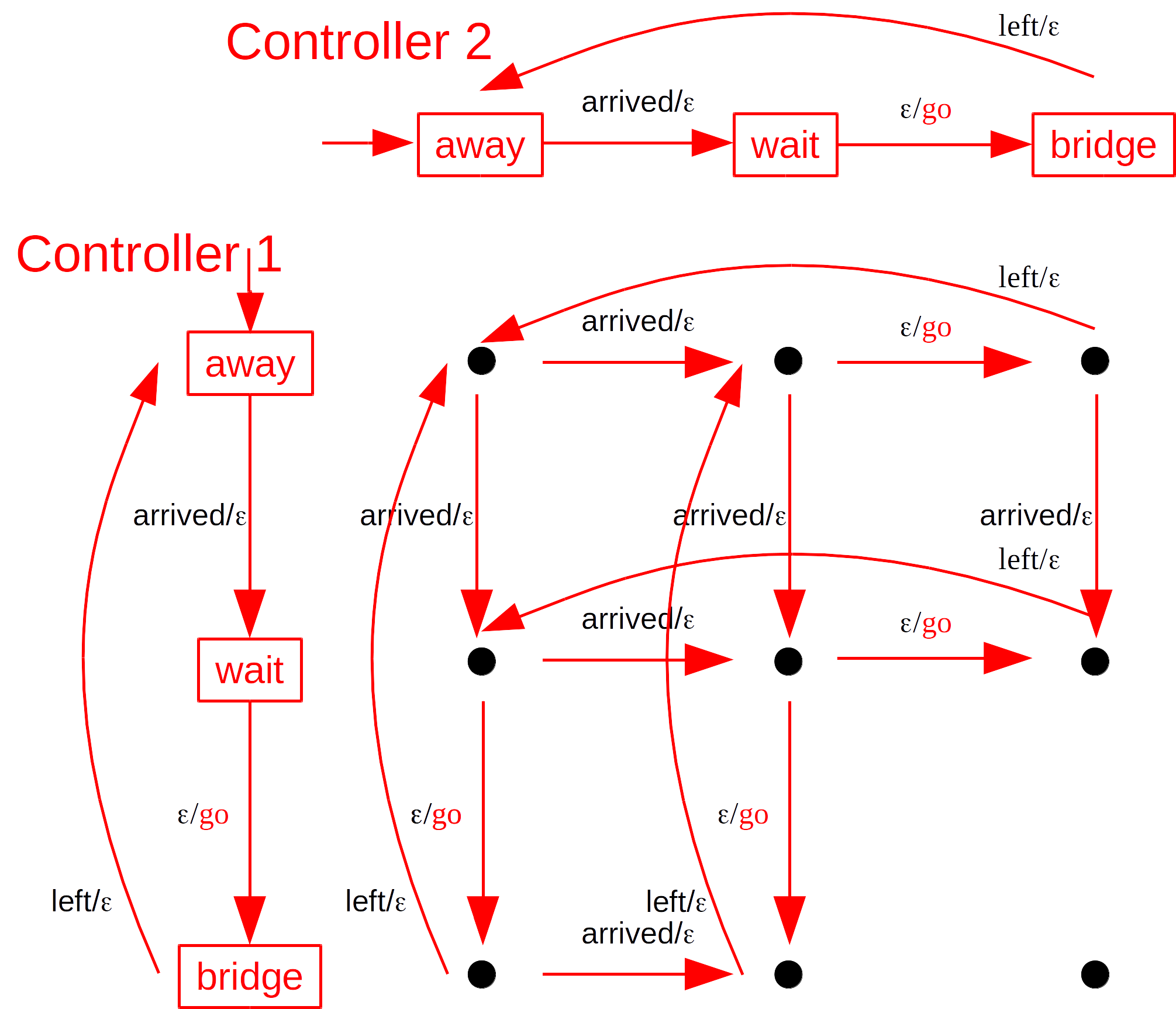} 
\caption{State diagram of the central controller that's supposed to coordinate two trains. It shows the inner coupling between two roles of the controller by elimination of the transitions towards the state $(bridge, bridge)$. Please compare it to Fig. \ref{fig_trains_protocol}, which shows in contrast the protocol between train and controller.}
\label{fig_trains_state_chart_controller}
\end{figure}

The coordination rule is, that both trains are prohibited to be on the bridge at the same time. This means that the internal state of the controller must not take the value $(bridge, bridge)$. Obviously, the state transitions that contain the $(bridge, bridge)$ state as the target state must be eliminated. This is possible without ill-forming the protocol, since these are only spontaneous transitions. 

Please note that the internal representation of the controller $(bridge, bridge)$ represents the fact that both trains are on the bridge only if the interaction between controller and bridge actually ensures that this is the case. 

In Fig. \ref{fig_trains_state_chart_controller} the controller in the state $(wait, wait)$ has to make a selection decision as to which of the two trains it gives permission to continue. Thus, the controller is not ''fully coordinating''. Removing either $(wait, wait) \rightarrow (wait, bridge)$ or $(wait, wait) \rightarrow (bridge, wait)$ breaks a symmetry and creates an automatic preference of one train over the other in case both trains are waiting to cross the bridge. This could lead to one of the trains never being let over the bridge - a starvation/fairness problem. 

The controller must make a decision. The question is how to calculate it. One possibility would be to introduce an additional third state indicating which train is allowed to enter the bridge. Its state value could be alternating to guarantee fairness or it could be set to a new random value, each time it is queried.

This example illustrates very well that nondeterminism in an interaction is important in order to coordinate it with other interactions. Additional interactions could restrict existing nondeterminism by putting on additional coordination constraints - but it could also increase the opportunities to take decisions. In the end, the remaining nondeterminism within a coordinated interaction has to be determined by decisions. Symmetry in an interaction can be broken by rules in the sense of coordination, but this can lead to the loss of important properties of the overall interaction, such as non-starvation.

\subsection{Snowflake topology of the protocol composition}

\begin{figure}[htbp]
  \begin{center}
    \includegraphics[width=8cm]{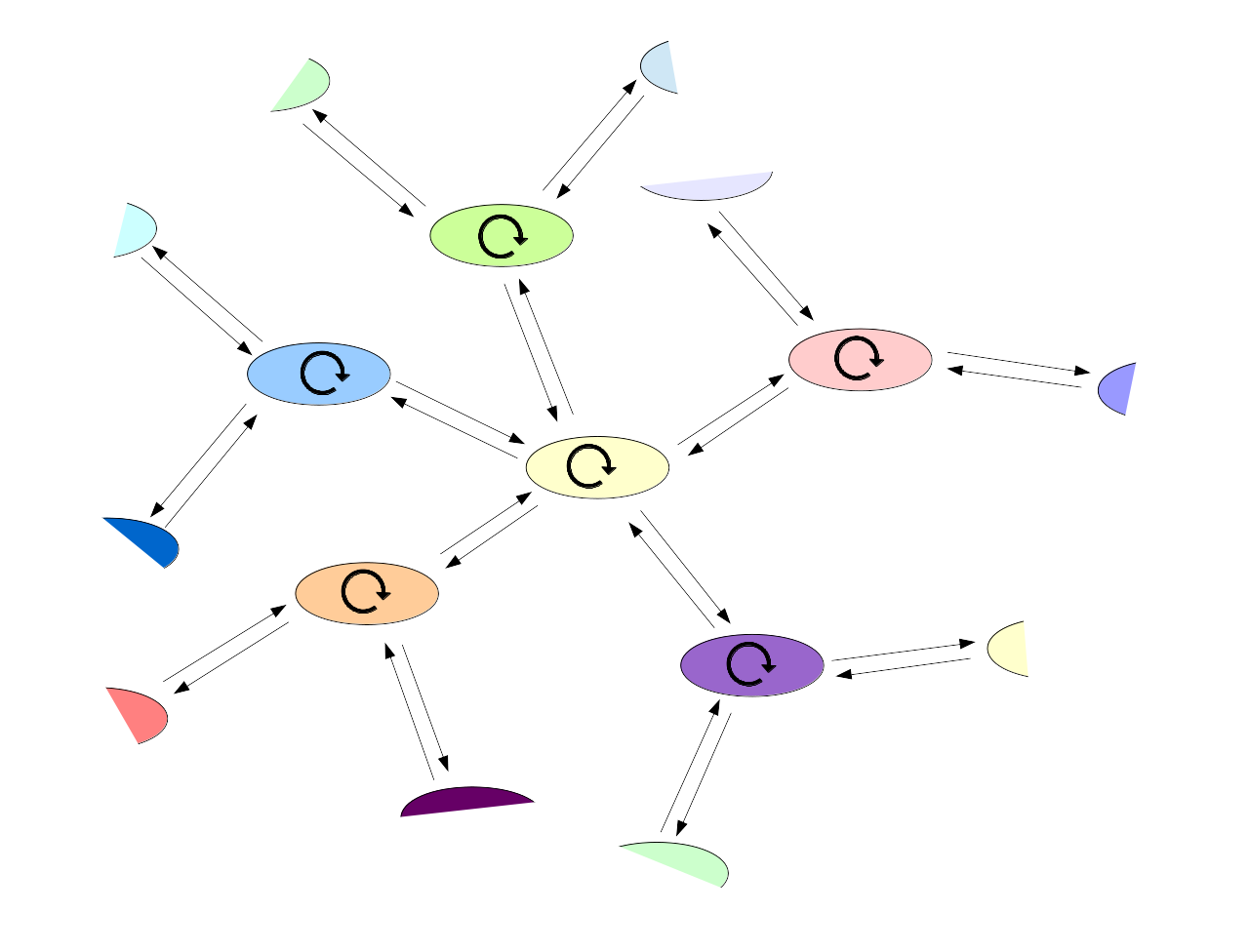}\\
 \end{center}
\caption[]{Snowflake topology of the compositional relationship for consistent protocols. The coordinating roles are drawn as complete ellipses with a circular arrow and represent interactive systems, while the marginal roles are represented as ellipse sections to emphasise the projection relationship of the roles to the systems they represent.} \label{fig_star_like_topology}
\end{figure}
 
In section \ref{ss_composition_of_interactions} I have stated with proposition \ref{prop_protocol_composition} a simple composition rule for protocols: If we select two roles from two different consistent protocols and coordinate these selected roles with each other in such a way that the respective acceptance conditions remain fulfilled, then this composite structure again results in a consistent protocol.

This rule implies a snowflake topology without any cycles, illustrated in Fig. \ref{fig_star_like_topology}

What happens if we abandon the snowflake topology and also allow other roles to coordinate with each other? In fact, the nondeterminism of participants can be lost as can be demonstrated with our example of the railway controller. As long as the trains act completely independently of each other, it can happen that both trains stand in front of the bridge at the same time. The controller then has to decide which of the two to let pass --- which is the justification of its existence in the first place. If, however, the trains coordinate themselves beforehand in the sense that they already ensure that only one train arrives at a time to cross the bridge, they actually eliminate $(wait, wait)$ as an attainable state of the controller and thereby effectively eliminate its decision.   

As you can see, this inner coupling of both trains also ''solves'' the actual problem that both trains do not come together on the bridge --- only now not loosely coupled via a well-defined protocol, but via an inner coupling mechanism, which, if one wanted to realise it, would again have to be implemented via a train-train interaction. 

Thus, a snowflake topology is a sufficient but not a necessary condition for retaining the consistency of protocols under invariant acceptance condition composition. In reality, we usually do not have this guarantee. But it is interesting to see the importance of the assumption of the interaction partners to act independently.

\begin{figure}[h!]
  \begin{center}
    \includegraphics[width=6cm]{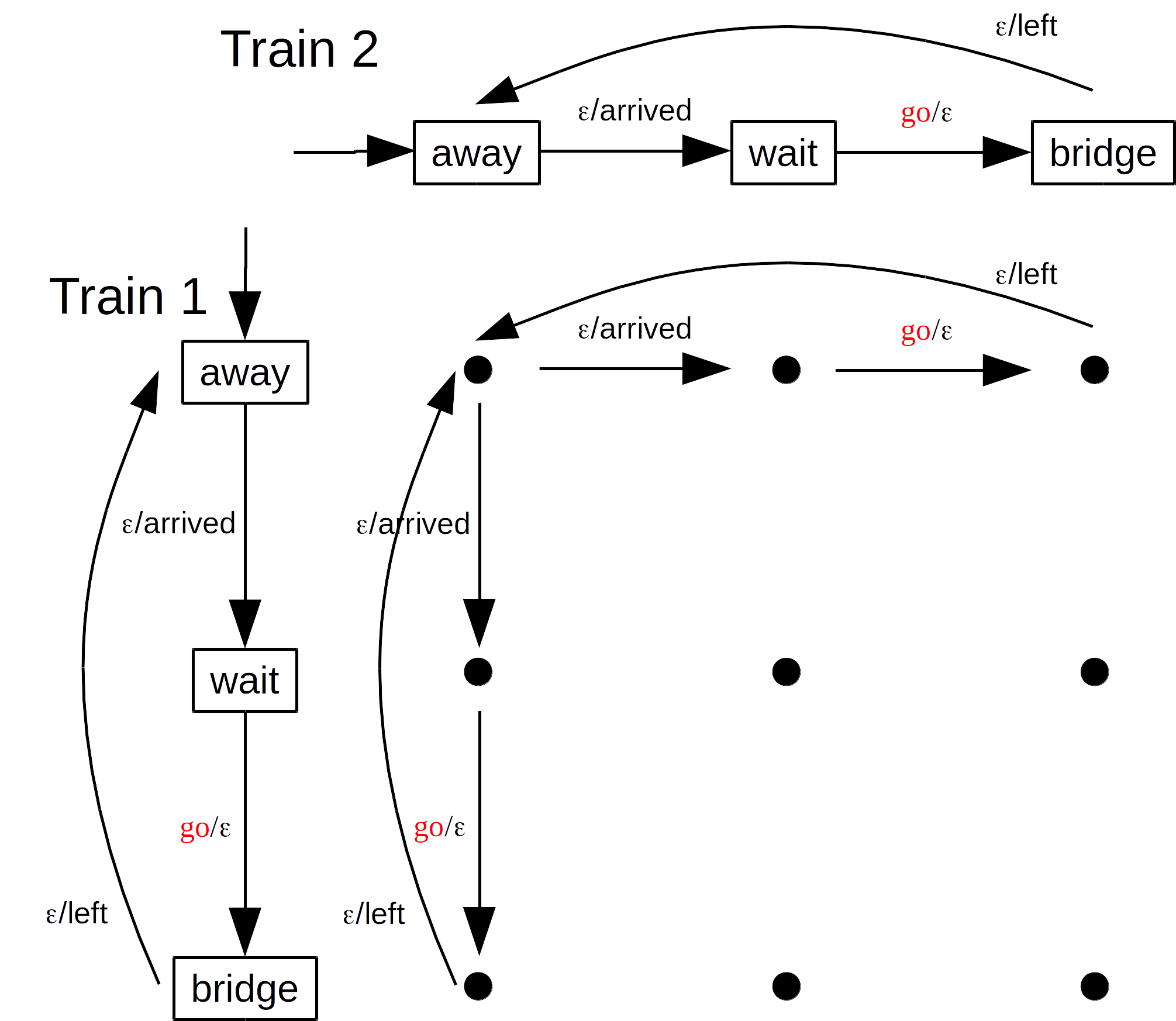}\\
 \end{center}
\caption[]{Coordination of two trains to avoid simultaneous arrival at the bridge. To this end, the transitions leading to the common state $(wait, wait)$ must be avoided. Thus, both trains have to coordinate their arrival at the bridge.} \label{fig_trains_coordination}
\end{figure}

\subsection{Interaction-oriented architecture of interactive systems}
When building applications, software engineering is exposed to the tension of describing systems via functions on the one hand, but on the other hand only being able to integrate these systems into the desired interaction networks via non-deterministic, stateful interactions. If one ignores this tension and actually describes interactive systems with an explicitly formulated system function, then one unfortunately loses any guarantee that small changes in individual interactions will also result in only small changes in the structure of the application.

This leads to the obvious requirement to find a balance by an ''{\it interaction-oriented architecture}'' here. I think that the  partitioned structure of an interactive system with its three parts, one that describes its interactions, a second part that describes its coordination and eventually a third part that describes its decisions does provide such a balance.

Given such a structure, the effort required to change the system description in case of changes in the interactions should be kept to a minimum by not endangering the consistency guarantees for the other roles of an interactive system when a role is modified.   

Since the freedom in the roles, i.e. the possibility of spontaneous and selection decisions, plays an essential role for their coordination possibilities, the protocol design aims at designing the protocol roles such that they can be integrated as versatile as possible. 

Please note that a deterministic transition relation is not necessary for automatic executability. It is sufficient, if no more selection decisions are to be made, i.e. there must be exactly one receiving transition from each attainable state to an input character or otherwise at most one spontaneous transition. The then still necessary spontaneous decision can be regarded as decisions ''to become active''. I therefore define:

\begin{definition}
An I/O-TS (as well as its transition relation), which has at most one transition from each attainable state for each possible input as well as to the empty character, is called executable ''{\it without alternative}''.
\end{definition}

\begin{definition}
A CA which is executable without alternative is called a ''{\it fully coordinating automaton}'' (fCA) or also ''process''.
\end{definition}

\begin{proposition}
A fCA whose input alphabet has been extended by a signal with the input alphabet $I_{tic} = \{tic\}$ that triggers all spontaneous transitions is an MIS.
\end{proposition}

\begin{proof}
To show this, we need to specify the corresponding system function of the fCA that is extended by $I_{tic}$ as described. Since the fCA has no other nondeterministic structures beyond its spontaneous transitions, this is possible 
\end{proof}

An obvious and very interesting research question, which I will not pursue further here, is,  under which conditions such an fCA exists and can be constructed effectively.

A further, in my opinion also very interesting research question refers to the obvious and also intended similarity of the introduced role concept with the role concept of colloquial language and also of sociology as for example elaborated by Erwin Goffman \cite{Goffman1959}. The question is whether this mechanisms of flexibilisation can also be found in nature. It could be imagined that nervous structures in biological organisms correspond to role representations and internal coordination mechanisms. It is obvious that humans are characterised by the ability to combine many different interactions and to adapt to considerable changes. Some of us are father, employee, son, husband, works council member, researcher, etc. at the same time. If a new task is added or one is omitted, this is usually possible with little effort, provided there are sufficient degrees of freedom.  

\subsubsection{Why explicitly specified processes can become very inflexible}
If we have created an fCA from its roles by defining the coordination rules and determining the remaining decisions, then we could try to create a traditional process, i.e. an MIS, with an explicit system operation by further summarising all still existing spontaneous transitions and eliminating intermediate states, to completely determine the system, so to speak, like this:
\[p_1 \stackrel{i/o_1}{\rightarrow} p_2 \stackrel{\epsilon/o_2}{\rightarrow} p_3 \, \Rightarrow \, p_1 \stackrel{i/(o_1, o_2)}{\rightarrow} p_3\]

It should first be noted that it is not possible to find an equivalent MIS for every fCA. In particular, the combination of spontaneous transitions with transitions with a true input characters, as shown, may lead to a transition that no longer outputs only a single character, but a sequence of characters (a word).  

But beyond that, complete determination of the transition relation has two further considerable disadvantages: 

\begin{enumerate}
\item The elimination of spontaneous transitions changes the atomicity of the time steps. In an environment that carries out several processes in parallel, this can lead to a change in the behaviour of the entire system, for example in the sense of a starvation. 
\item The projection relation between the restricted product automaton and its constituent role automata is lost. In particular, the acceptance component would have to be changed accordingly if state values were eliminated, since it refers directly to a set of state values. And therefore the invariance of the properties of the roles in their interaction gets lost.
\end{enumerate}

Why is the invariance of roles so important in software engineering? From a synthesis perspective, the problem arises of how strongly changes in interactions affect changes in interacting processes. It is desirable that minor changes in the interactions result in only minor changes in the processes. Especially since complex processes participate in many different interactions, an imbalance in the adaptation effort quickly leads to the inability of process synthesis to keep pace with the change dynamics of the interactions, even if individual interactions change only comparatively rarely and little. 

If one formulates the process functions each time explicitly in the sense of a calculable function, there is no reason to assume that minor changes of the interactions always lead to the fact that the internal structure of these functions, that is the sequence and kind of the computation steps, will also change comparatively little. On the other hand, my conjecture is that in role- or interaction-oriented process synthesis, due to the invariance of all unchanged roles, the change effort is of the same order of magnitude as the change in the roles themselves. Or it grows perhaps, at least in the predominant number of practically relevant cases, only benignly with the number of roles to be coordinated --- an interesting research question. 

The importance of this problem of the inflexibility of explicitly stated processes is also acknowledged by other authors, such as Nirmit Desai, Ashok U. Mallya, Amit K. Chopra and Munindar P. Singh \cite{DBLP:journals/tse/DesaiMCS05}. In my opinion its relevance stretches far beyond computer science. It is the deeper cause why approaches in the field of organisational sciences, such as Lean \cite{RoosWomackJones1991_Machine}, which make exactly this explicit formulation of the business processes a prerequisite for a successful management of process changes, clearly lag behind their promises. So, the suggestion for organisational science would be to think more in terms of interactions or rules then in processes. It seems that the area of compliance takes exactly this direction \cite{IDW_PS_980}.

\subsection{Composition of interactive systems}
In section \ref{ss_def_interaktive_systeme} I had said that interactive systems interact nondeterministically and therefore there is no formation of supersystems by interaction. Now we want to explore the question of how interactive systems compose --- isn't this a contradiction?

We can resolve this contradiction in two quite different ways. On the one hand, we just cannot say that the interaction of interactive systems always leads to a superordinate interactive system. This follows directly from the snowflake topology of the compositional rules for consistent protocols.  A sufficient condition is that  the topology of the network under consideration is without loop. Thus the question whether we can view two mutually interacting interactive systems as a interactive supersystem depends on the additional information on how all their other roles are linked. 

Assuming no further link between the involved systems beside the protocol under consideration the question of composition of interactive systems boils down to whether two CAs that are linked solely by a protocol can again be described as a single CA without the two roles of the protocol - an interesting open question that I pursue not further in this version of this article.

A further resolution of the apparent contradiction to investigate the composition of interactive systems lies in the changed semantics of this composition compared to simple or recursive systems. 
Interactive systems are in general not representable as MIS. In order to be able to assign a system function to interactive systems, I introduced the concept of decisions as a means to describe the interaction partners completely despite incomplete knowledge and to decouple them at the same time in the consideration from their other interactions. 

The question of the composition of interactive systems, i.e. whether the interaction of two interactive systems via a consistent protocol produces another interactive system, is therefore actually the question of the composition of the assigned decision systems and thus the question of whether the decision-making ability of these systems is retained under the composition. 

By modifying the system concept by means of the additional ''inner'' input alphabet of the decisions, the question of the composition of these systems has thus acquired a completely different meaning compared to simple or recursive systems. 

A last remark in this section: In the area of ''Business Process Management (BPM)'' people often speak of ''end-to-end'' processes and usually mean the executable functionality of an MIS in our sense (e.g. \cite{Hammer2015_BPM,Jeston2014_BPM}). If we identify the nodes of these networks with the processes, then we see that the meant ''end-to-end'' functionality does not exist in a well-defined sense, but that the ends of the processes are always the role definitions of ''the others''. Otherwise we could not clearly state the meaning of the state values of the processes. Remember the single-track railway example: the truthfulness of the controller's state depended on the consistency of the train-controller protocol. 

%
\section{Partitioning the I/O-TSs by equivalence relations} \label{s_partitioning}
%
So far we have always limited our presentation to individual characters and state transitions. This quickly becomes confusing for more complex applications. 
More complex use cases can be handled if they allow to consider classes of transitions to be somehow equivalent. Then, in mathematical terms, we can divide the previously considered transition relation $\Delta$ into equivalence classes $\{\Delta_l\}_{l\in L \subseteq \mathbb{N}}$ with the help of the respective equivalence relation $\sim$ while preserving its essential properties.

In general, an equivalence relation $\sim$ on a set $R$ is defined by the three properties for three elements $a, b, c \in R$:
\begin{enumerate}
\item {\it Reflexivity}: Each element is equivalent to itself: $a \sim a$. 
\item {\it Symmetry}: Equivalence is always valid in both directions: $a\sim b\ \Rightarrow \ b\sim a$.
\item {\it Transitivity}: If $a$ is equivalent to $b$ and $b$ is equivalent to $c$, $a$ is equivalent to $c$ as well: $a\sim b$ and $b\sim c$ $\Rightarrow$ $a\sim c$.  
\end{enumerate} 

Such a partition implies a function $part:\Delta \rightarrow L$ assigning the index of the equivalence class to each element of the transition relation. 

In fact, I present two different equivalence class constructs, one for the deterministic case and one for the nondeterministic case. 

\subsection{Objects in the sense of object orientation} \label{ss_objects}
In the case of a deterministic I/O-TS, we can break down the transition relation $\Delta$ into subtransition relations $\Delta_l$, which in turn are deterministic. Thus each of them defines a function $\delta_l: I\times Q \rightarrow O\times Q$ with $(o,q) = \delta_l(i,p)$ for $(p,q,i,o) \in \delta_l$.

The following applies: $(o, q) = \delta_l(i,p) = \delta_l|_{p}(i)$. In fact, these are two equations:

\begin{eqnarray}
 o & = & \delta_l^{(o)}|_{p}(i) \label{eq_object_I}\\\ 
 q & = & \delta_l^{(q)}|_{p}(i) \label{eq_object_II}
\end{eqnarray}
   
The first equation stands for the object-oriented way \cite{DahlNygard1966,Jacobson1992} of representing a change of state of a so-called ''object'' by representing the signal --or variable, as the signal is often called -- that represents the state with a name, the ''object name''. This signal is then assigned the operation, such as: outputParameters = objectName.operation(inputParameters). The second equation is only implicit in the world of object orientation.

We see that ''objects'' in the sense of object orientation are actually systems in our sense. In addition to the system character, objects also represent a certain compositional behaviour. Namely to feed their output back to the system where they received their input from. Thus, a composed object always represents the composed super system.     

In this context it is interesting to note a certain ambivalence in the development of object orientation between nondeterministic interactions and compositional functionality. On the one hand the very influential programming language Simula was actually designed by Ole-Johan Dahl and Kristen Nygaard for the purpose of simulating physical systems \cite{DahlNygard1966}. On the other hand, the programming paradigm of Smalltalk \cite{Kay1996}, another very influential programming language in the beginnings of object orientation, conceptually consisted more of message processing in the sense that a Smalltalk object could hold a state, receive messages, and, when processing a message, send messages to itself or other objects. 
But without a clear concept of nondeterministic interaction, the relationship to the processing functionality was ultimately stronger.

If we additionally divide the inner state into a main component, which I call ''mode'', and the remainder $Q = Q_{mode} \times Q_{rest}$, we get three equations:

\begin{eqnarray}
 o & = & \delta_l^{(o)}|_{p_{mode}}(i, p_{rest}) \\ 
 q_{mode} & = & \delta_l^{(mode)}|_{p_{mode}}(i, p_{rest}) \\ 
 q_{rest} & = & \delta_l^{(rest)}|_{p_{mode}}(i, p_{rest})
\end{eqnarray}

The determination of a mode state leads to the so-called ''state pattern'' of the object-oriented world \cite{Gamma1995}. With its help you can define generic events $(p_{mode}, q_{mode})$. 

\subsection{Documents and main states} \label{ss_documents_and_main_states}
In the nondeterministic case, the equivalence class construction consists of dividing the character sets into ''document classes'' and the state value set into more or less ''relevant'' values. 

The key idea of this equivalent class construction is that the partition function depends not directly on the single characters $i$ and $o$ and the state values $p$, and $q$ of the transitions, but it depends on a combination of coarser criteria given by document classes, the most relevant state value components as well as some additional conditions. 

First, I introduce a function $parse: I\cup O \rightarrow DocCls \times Param$ with $(docCls, param) = parse(a)$, which assigns each character $a\in I\cup O$ a document class $docCls(a)$ and a value of a parameter set $param(a)$. Both, the set of document classes and the (possibly multidimensional) values of the parameter set are alphabets and can therefore be processed by operations. We can thereby regard each character as an instance of a document of a given document class that represents certain parameter values. 

I also split the state of the I/O-TS into a mode component and the rest: $Q = Q_{mode} \times Q_{rest}$, where $mode(q)$ denotes the $mode$-part and $rest(q)$ denotes the $rest$-part of $q$.

And I introduce a set of conditions $cond$ as functions $cond_{mode(p), docCls(i)}: Q_{rest(p)} \times param(i) \rightarrow \{true, false\}$, which test whether values of the remaining state and the parameters of the input documents evaluate a condition as true or false. 

With these tools I now define the desired equivalence relation: 

\begin{definition}
Two transitions $t,t'\in\Delta$ are equivalent, symbolically $t \sim_{ND} t'$, if the values of the partitioning function $part: \Delta\rightarrow \mathbb{N}$ with 
\begin{eqnarray*}
  part(i, o, p, q) & = part(& docCls(i), docCls(o), mode(p), mode(q),\\
                   &        & cond_{mode(p), docCls(i)}(rest(p), param(i)))
\end{eqnarray*}
for both transitions are identical, i.e. $part(t) = part(t')$.
\end{definition}

Please note that $rest(q)$ and $param(o)$ are ignored for the construction. 
Instead of $\Delta_l$ we also write:
\begin{equation} \label{eq_equiv_class_transition}
 mode(p)\stackrel{docCls(i), cond_{mode(p), docCls(i)}(rest(p), param(i))/docCls(o)}{\xrightarrow{\hspace*{5cm}}} mode(q)
\end{equation}

An example is a seller who receives a purchase order $Order$ from a customer and passes accordingly from mode state value $listening$ to $ordered$, provided that the customer is assessed as $trustworthy$. Accordingly, she returns a confirmation as $Confirmation$:  
\[ listening\stackrel{Order, isTrustworthy(Customer)/Confirmation}{\xrightarrow{\hspace*{5cm}}} ordered\]
The equivalence relation abstracts away all the nitty gritty details of which item was ordered, how many items were ordered, the identity and address data of the customer, etc. but views all such transitions as equivalent with respect to this representation. 

Thus, it is important to keep in mind, that unlike the $p \stackrel{i/o}{\rightarrow} q$ notation, which denotes a single transition, the term in Eq. (\ref{eq_equiv_class_transition}) denotes a whole class of equivalent transitions. 

As a consequence, any protocol role can be specified as a table of 5 columns denoting the items of 
Relation (\ref{eq_equiv_class_transition}), where each row contains the information of one transition equivalence class.

%
\section{Other system models in the literature \label{s_related_work}} 
%
There is a vast amount of existing literature on formally worked out system models. Any extensive review would be an ambitious endeavour on its own. Its quite interesting to see, how differently this subject has been treated. I restrict myself to a couple of other system models that have gained some prominence in my view. I mainly look at how they represent the system function, whether they reproduce the compositionality of the system function in the sense of computability, and how balanced their approach is with respect to the two concepts of interaction and coordination. 

This investigation suggests that these criteria do provide valid assessment criteria. To some extent, the expressiveness of the investigated system models does not suffice to represent the differences between parallel versus sequential composition or even recursion in a simple enough way or do not make a clear enough distinction between the  deterministic versus nondeterministic interaction paradigms. Some lack the compositionality of their system function. And some other do not provide the formal means to distinguish between external and internal coupling.

\subsection{Reactive systems of David Harel and Amir Pnueli}
In \cite{HarelPnueli1985_Reactive} David Harel and Amir Pnueli distinguish between  ''{\it transformational}'' and ''{\it reactive}'' systems. They define: ''{\it A transformational system accepts input, performs transformation on them and produces outputs}'' - which is consistent with our definition \ref{def_system_einfach} of a simple system. They further define ''{\it Reactive systems ... are repeatedly prompted by the outside world and their role is to continuously respond to external inputs ... A reactive system, in general, does not compute or perform a function, but is supposed to maintain a certain ongoing relationship, so to speak, with its environment.}''

The reason, they put so much weight on this distinction is their opinion, that the usual way to design software systems from high level descriptions to more concrete structures is essentially transformational in nature, fitting very well to transformational systems. In a reactive system, in contrast, ''it is not clear if or how complex behaviour can at all be decomposed beneficially into chunks''. They propose the state chart method for specifying the behaviour of complex reactive systems as a combination of Moore and Mealy automata. It consists of describing the system's behaviour in terms of states, events, activities and conditions where the combination of events and conditions trigger transitions between the states and activities. The system's inputs are the (external) events and its outputs are the (external) activities.   

In my opinion David Harel and Amir Pnueli made a very good point in realising that the composition behaviour of systems which represent functions in their interface differs fundamentally from the composition behaviour of what they call reactive systems, and I prefer to call interactive systems. But as I tried to show, it is false to conclude from the non-functional relation of David Harel's and Amir Pnueli's reactive systems to their interaction partners to their somehow non-functional functioning - or ''mode of operation''. Even reactive systems work in a step wise mode and at least for technical systems the simple fact holds that for being contractible, there has to be a function to be implemented - the system function. 

The main difference between ''transformational'' and ''reactive'' systems in the sense of David Harel and Amir Pnueli is not their functioning but - as they correctly point out - their relation to their environment and their compositional behaviour. For ''transformational'' systems their interfaces towards the interaction provides access to the full system function and thereby provides deterministic interactions together with a hierarchical composition according to the definition of sub-/super-systems in section \ref{s_simple_systems}. By contrast, as I proposed already in \cite{Reich2012_PRI}, the interfaces of ''reactive'' systems provide access to at least two different projections of the systems, which I call ''roles'' and thereby only provide nondeterministic interactions towards the environment.

\subsection{The system model of Manfred Broy}
Manfred Broy describes his component model in several articles (e.g. \cite{DBLP:books/sp/cstoday95/Broy95,Broy2009_Relating_Time,Broy2010_Logical}). It is an example of a model of interactive systems that is based on input and output histories alone and does without the explicit description of an internal state. 

A major problem with this approach is that the input and output histories of composite components do not result exclusively from the input and output histories of the subcomponents \cite{BrockAckerman1981_Scenarios}. Instead, further assumptions about causality in the sense of a partial ordering of the events of the histories among themselves are necessary, which Manfred Broy also has to introduce accordingly. 

He defines the behaviour of a system by a set of ''input'' and ''output'' processes, whereby a process in his sense is a finite or infinite set of discrete events. The event concept is not entirely clear to me, as on the one hand, an event may represent (among others) a state change or some action, but on the other hand he states that an event represents a point in time and thus has no time duration. Also, it remains unclear, whether event sequences like $a,\epsilon, b$ and $a, b$ are semantically equivalent or not. 

A key concept in his theory is what he calls a {\it timed stream} of messages $m\in M$ which is an infinite sequence of finite sequences of messages. Each finite sequence $s\in M^*$ can be attributed to a time interval $t\in \mathbb{N}$ as $s(t)$. The set of timed streams is denoted as $(M^*)^\infty$. If we assume $M=\{0,1,\dots,9\}$, then $M^*$ becomes the set of all natural numbers $\mathbb{N}$ and $(M^*)^\infty$ is the set of all infinite sequences of arbitrary natural numbers. As can be seen from the example, the amount of information to be processes in one time interval is finite, but can be arbitrary large.

The messages are transported by typed channels $c\in C$. A {\it channel valuation} is a mapping which assigns a typed timed stream to a channel. The set of all valuations to a given set of channels $C$ is then defined as $\vec{C} = \{x:C \rightarrow (M^*)^\infty:\forall c\in C: x(c)\in (type(c)^*)^\infty\}$. Each channel valuation defines a communication history $x:C \rightarrow (\mathbb{N} \rightarrow M^*)$ for the channels in $C$.

He calls the set of (typed) input and output channels a ''syntactic interface'' and the relation between the input and output histories a ''semantic interface'' of a component. He describes the latter by a function $f$ mapping the set of the histories of the input channels $\vec{I}$ onto the power set of histories of the output channels $\wp(\vec{O})$ of the component: $f:\vec{I} \rightarrow \wp(\vec{O})$. 

Manfed Broy then explicitly introduces the properties of causality as well as realisability in the sense that all output histories can indeed be computed.

Composition is defined by a single composition operator with feedback. 
Given two interface behaviours with disjoint set of output channels $O_1 \cap O_2 = \emptyset$: $F_1:\vec{I}_1 \rightarrow \wp(\vec{O}_1)$, $F_2:\vec{I}_2 \rightarrow \wp(\vec{O}_2)$.
Then with $I= (I_1 \cup I_2)\setminus (O_1 \cup O_2)$, $O=(O_1 \cup O_2) \setminus (I_1 \cup I_2)$, and $C=I_1 \cup I_2 \cup O_1 \cup O_2$ Manfred Broy defines the interface behaviour of the composed interface $F = F_1\otimes F_2: \vec{I} \rightarrow \wp(\vec{O})$ by $F(x) = \{(y\in \vec{C})|O: y|I = x|I \wedge y|O_1 \in F_1(y|I_1) \wedge y|O_2 \in F_2(y|I_2)\}$ where $x|Y$ means the restriction of the valuation $x$ to the channels in $Y$.

According to Manfred Broy, the composition operator $\otimes$ preserves realisability.

Due to the behavioural approach which focuses on observable behaviour, the definition of black box behaviour eschew internal states - but not for composed systems. The definition of a composed system allows for hidden channels which become, according to my understanding, nothing else than internal states in the sense of time dependent attributes of the composed system. Lets look at two simple systems $i=1,2$ each with two input channels $I_i = \{I_{i,1}, I_{i,2}\}$, two output channels $O_i = \{O_{i,1}, O_{i,2}\}$ and a deterministic system function $f_i = (f_{i,1}, f_{i,2}) = (I_{i,1} \wedge I_{i,2}, \neg(I_{i,1} \wedge I_{i,2}))$. The system composition is recursive with $O_{1,2} = I_{2,1}$ and $O_{2,1} = I_{1,2}$. By this equalities, we see that the channels do not just represent the channel names as Manfred Broy states, but the channels themselves. Otherwise the equality of two channel variables relating to two channels with two different names would not make much sense.

Following Manfred Broy's composition rule we then have $I= \{I_{1,1}, I_{2,2}\}$, $O=\{O_{1,1}, O_{2,2}\}$, $C=\{I_{1,1}, I_{1,2}, I_{2,1}, I_{2,2}, O_{1,1}, O_{1,2}, O_{2,1}, O_{2,2}\}$ and $F(x) = \{(y\in \vec{C})|O: y|I = x|I \wedge y|O_1 \in F_1(y|I_1) \wedge y|O_2 \in F_2(y|I_2)\}$.

This seems to be pretty simple and straight forward. So one is tempted to ask where the recursion has been hidden. The complexity is hidden in the fact, that without coupling, the streams $I_{1,2}$ and $I_{2,1}$ can take arbitrary values at each point in time, while in the coupled case, only their first value can be arbitrarily chosen and all other values depend in a complex way on these two initial values together with the complete history of input values of the remaining input streams  $I_{1,1}$ and $I_{2,2}$. 

So, the complexity of recursive system relations cannot be avoided. But recurring to the complex structure of timed streams, my impression is that Manfred Broy adds a substantial amount of unnecessary additional complexity whereas it is difficult for me to find the additional benefit. 

\subsection{The BIP component framework of Joseph Sifakis et al.}
The group around Joseph Sifakis has developed the component model ''Behaviour, Interaction, and Priority (BIP)'' \cite{BazuBozgaSifakis2006,BasuEtAl2011_Rigorous,BasuEtAl2011_BIP,BensalemBozgaQuilbeufSifakis2012_Knowledge} to describe interactive systems. This model, as I understand it, is quite similar to the model for interactive systems presented in this work with its partitioning of state and its rule orientation. But it does not distinguish sufficiently clearly between information transport and processing, because computations are performed in the components as well as in the connectors.    

BIP-components can be atomic or composed. An {\it atomic component} is a labelled finite transition system characterised by the tuple $(State, state_0, P, G, F, Trans)$. 
$State = S \times V$ gives the possible state values. $S$ is called control state, $V = V_1 \times \dots \times V_n$ is called data state\footnote{the authors use $V$ to denote the variables which represent the data states syntactically. Here the index $i$ relates to the $i$-th data state [variable].}. $state_0$ is the initial state.
$P$ is the set of ports, which are transition labels [or action names] used for synchronous state transitions of multiple components. 
$G$ is a set of ''guard'' conditions operating on $V$, that is each $g\in G$ is a function $g:V\rightarrow \{true, false\}$.
$F$ is a set of functions operating on V, that is for each $f\in F$, $f:V\rightarrow V$. 
$Trans$ is the set of transitions with $trans = State \times P \times G \times F \times State$. A transition $t = (state, p, g_p, f_p, state')$ represents a step from state value $(s, v)$ to $(s', v')$ with $v' = f_p(v)$. It can be executed if some interaction including port $p$ is offered by the execution environment and the guard $g_p(v) = true$. 

Components are glued together by connectors defining interactions and by priorities. For simplicity we assume components with disjoint sets of names of ports, state values and transitions.

A {\it connector} specifies interactions between its components $\{C_1, \dots, C_n\}$. It is characterised by a tuple $(V, P, I, G, F, Trans)$.
$V = V_1 \times \dots \times V_n$ is called data state\footnote{Now the index $i$ relates to the $i$-th component.}.
$P$ is the set of ports the connector relates to.
$I\subseteq \mathcal{P}(P)$ is the set of feasible interactions, where each $i\in I$, $i \subseteq P$. 
$G$ is a set of ''guard'' conditions operating on $V$, that is each $g\in G$ is a function $g:V\rightarrow \{true, false\}$ having access to the complete data state of all components.
$F$ is a set of functions operating on V, that is for each $f\in F$, $f:V\rightarrow V$, also having access to the complete data state of all components and providing ''data exchange'' between the components. 
$Trans$ is the set of transitions with $trans = V \times I \times G \times F \times V$. A transition $t = (v, i, g_i, f_i, v')$ represents a step from state value $v$ to $v' = f_i(v)$. It can be executed if its interaction becomes feasible (that is, the respective components have reached a state such that each provides transitions labelled with the respective port) and the guard $g_p(v) = true$.

A {\it priority} selects among possible interactions. It is characterised by a tuple $(c, G, I, <)$.
$c$ is the component the priority belongs to, which composes $\{C_1, \dots, C_n\}$ with data state $V = V_1 \times \dots \times V_n$.
$G$ is a set of ''guard'' conditions operating on $V$.
$I$ is a set of interactions.
$< \subseteq I \times I$ is a strict partial order relation providing the priority.
When the condition holds and both interactions are enabled, only the higher one is possible. 

A {\it compound component} is then given by its components, its connectors and its priorities. Neglecting priorities, it is finally characterised by a tuple $(State, P, I, G, F, Trans)$. 
$State = S\times V = S_1 \times \dots \times S_n \times V_1 \times \dots \times V_n$.
$P = \bigcup P_i$, $I = \bigcup I_i$, $G = \bigcup G_i$, and $F = \bigcup F_i$ is the set of functions, operating on $V$.
$Trans \subseteq State \times I \times G \times F \times State$ is the set of transitions where a transition is given by $t=(state, \alpha, g, f, state')$.
$\alpha$ is a feasible interaction associated with a guard $g_\alpha \in G$ and a function $f_\alpha \in F$ such that there 
exists a subset $J\subseteq\{1\dots n\}$ of atomic components with transitions $\{(state_j, p_j, g_j, f_j, state_j')\}_{j\in J}$ and $\alpha = \{p_j\}_{j\in J}$.
$g = \left(\bigwedge_{j\in J} g_j\right) \wedge g_\alpha$.
$f = (f_1, \dots, f_n) \circ f_\alpha$.
$v' = f(v)$.
$s'(j) = s_j'$ if $j\in J$; otherwise $s'(j) = s_j$. That is, the states from which there are no transitions labelled with ports in $\alpha$, remain unchanged.

Without priorities, a move $(s, v) \stackrel{\alpha}{\rightarrow} (s', v')$ is possible if there exists a transition $(state, \alpha, g, f, state')$, such that $g(v) = true$.

Most importantly, the functionality of BIP components is not compositional in the sense of definition \ref{def_compositionality}. There is always the global function operating on all states of a composed component at once first. And only afterwards, the components' genuine functions are applied to the individual states, violating the encapsulation and thereby making this model of reactive systems complex. The purpose of this connector-function is probably to represent data exchange --- but obviously it is not limited to reproduce states, but it is allowed to make arbitrary manipulations. So in my opinion, this model does not distinguish clearly between transport and processing of information.  

\subsection{The system models of Rajeev Alur}
In his excellent book ''Principles of Cyber-Physical Systems'' \cite{Alur2015_Principles} Rajeev Alur also distinguishes between functional and reactive components in the tradition of David Harel and Amir Pnueli.  

Focusing on the timing behaviour, he develops several models of computation for reactive components. In the synchronous model all components execute in lock-step, similar to what I would call ''clocked'' systems. In the asynchronous model all processes execute at independent speeds, and there is an unspecified delay between the reception of inputs and the production of outputs by a process. In the timed model processes rely on a global physical time to achieve a loose form of synchronisation.   

\subsubsection{Synchronous reactive components (SRC)} \label{s_synchronous_reactive_systems}
In the interconnection of digital systems, such as logic gates or flip-flops, input and output may refer to ''the same'' point in time, and only the update of any internal state refers to the next point in time. 
\begin{eqnarray*}
   q(t') & = & f^{int}(in(t), q(t))\\
   out(t) & = & f^{ext}(in(t), q(t))
\end{eqnarray*}
This system model is the basis of synchronously reactive systems  \cite{BenvenisteCaspiEdwardsHalbwachsLeGuernicDeSimone2003} and has found its expression in quite a number of hardware description languages, such as Lustre, Esterl, Signal or VHDL/Verilog.

Rajev Alur models a {\it synchronous reactive component (SRC)} $C$ by a tuple $C = (I, O, S, Init, React)$, mixing semantic elements like sets, states, and values with syntactic elements like variables and descriptions. $I, O, S$ are sets of typed input, output and state variables, each variable taking values only from the sets $Q_I, Q_O, Q_S$. $Init$ is the description of the initialisation defining the set $[\![Init]\!] \subseteq Q_S$ of initial states, and {\it React} is a description of the reaction defining the transition relation $[\![React]\!] \subseteq Q_S \times Q_I \times Q_O \times Q_S$.

An SRC $C$ is said to be {\it event triggered}, if a subset $J \subseteq I$ exist, where all the variables in $J$ also can take the value $\bot$ (meaning ''absent''); every output variable is either latched (the value of the output variable is the updated value of a state variable) or can also be absent; and for a reaction with absent input, the non-latched output remains also absent. 

The reaction $React = (L, {\cal A}, \succ)$ is given by a set $L$ of task local variables and a set ${\cal A}$ of tasks and a binary precedence relation $\succ \subseteq {\cal A} \times {\cal A}$ (The latter two represent a {\it task-graph}). Each task $A$ has a read-set $R \subseteq I \cup S \cup O \cup L$, a write-set $W \subseteq O \cup S \cup L$, and an update description $Update$ with $[\![Update]\!] \subseteq Q_R \times Q_W$ such that (1) $\succ$ is acyclic; (2) each output variable belongs to the write-set of exactly one task; (3) if an output or a local variable $y$ belongs to the read-set of a task $A$, then there exists a task $A'$ such that $y$ is in the write-set of $A'$ and $A' \succ^+ A$; and (4) if a state or a local variable $x$ belongs to the write-set of a task $A$ and also to either the read-set or write-set of a different task $A'$, then either $A \succ^+ A'$ or $A' \succ^+ A$.

As he says, his approach to ensuring well-behaved composition relies on the syntactic decomposition of the reaction description into tasks given by the designer. He introduce the notion of an interface for $I$, $O$, $\succ \subseteq I \times O$ of a component to ensure consistent composition.
Two components $C_1$ and $C_2$ with $I_i$, $O_i$, and $\succ_i$ $(i=1,2)$ are said to be {\it compatible} if (1) $O_1$ and $O_2$ are disjoint and (2) the relation $\succ = \succ_1 \cup \succ_2$ is acyclic. Actually, the task graph over the set of tasks of $C_1$ and $C_2$ obtained by  retaining the precedence edges in the individual components and adding cross-component edges from a task $A_1$ of one component to a task $A_2$ of another component whenever $A_1$ writes a variable read by $A_2$, is acyclic.

For composition, we assume that the input/output variables are named in a way to establishes the intended communication pattern and the internal variables are named to avoid naming conflicts. Let $C_i = (I_i, O_i, S_i, Init_1, React_i)$, $i=1,2$ be two compatible SCR. Suppose that $React_i = (L_i, {\cal A}_i, \succ_i)$. Then the {\it parallel composition} $C = C_1|| C_2$ is again an SCR such that $S = S_1 \cup S_2$, $O = O_1\cup O_2$, $I = I_1 \cup I_2 \setminus O$, the initialisation for a state variable $x_i$ is given by $Init_i$, $(i=1,2)$; the reaction description of $C$ uses the local variables $L=L_1\cup L_2$ and is given by the task graph such that (1) ${\cal A} = {\cal A}_1 \cup {\cal A}_2$ and (2) $\succ = \succ_1 \cup \succ_2$ together with the task pairs $(A_1, A_2)$, such that $A_1$ and $A_2$ are tasks of different components with some variable occurring in both, the write-set of $A_1$ and the read-set of $A_2$.

The commutativity of parallel composition of SCRs indicates that the formalism does not express the different composition schemata described in section \ref{s_simple_systems}. With the compatibility condition recursion is excluded.

\subsubsection{Asynchronous processes} In the asynchronous case, Rajev Alur talks about processes and channels instead of components and input/output variables. 

An {\it asynchronous process (AP)} $P$ is defined again by the tuple $P = (I, O, S, $ $Init, React)$. However, the reaction is defined differently with input task $A_x$ for every input channel $x\in I$, an output tasks $A_y$ for every output channel $y$ and a set of internal tasks ${\cal A}$ as $React = (\{{\cal A}_x|x\in I\}, \{{\cal A}_y|y\in O\}, {\cal A_i})$. Each task consist of a guard condition $Guard$ over $S$ and an update rule $Update$. The update rule of each $A_x$ has a read-set $S \cup \{x\}$ and a write-set $S$, each $A_y$ has a read-set $S$ and a write-set $S \cup \{y\}$, and each $A$'s read- and write-set is $S$. 
  
The {\it parallel composition} $P_1|P_2$ of two AP $P_i = (I_i, O_i, S_i, Init_i, React_i$, $(i=1,2)$ is again an AP if $O_1$ and $O_2$ are disjoint and $S = S_1 \cup S_2$, $O = O_1 \cup O_2$, $I = I_1 \cup I_2 \setminus O$, and $Init = Init_1; Init_2$. 
For each input channel $x\in I$, (1) if $x\notin I_2$, then the set of input tasks ${\cal A}_x$ is ${\cal A}_x^1$; (2) if $x\notin I_1$, then the set of input tasks ${\cal A}_x$ is ${\cal A}_x^2$; and (3) if $x\in I_1 \cap I_2$, then for each task $A_1 \in {\cal A}_x^1$ and $A_2 \in {\cal A}_x^2$, the set of input tasks ${\cal A}_x$ contains the task described by $Guard_1 \wedge Guard_2 \rightarrow Update_1;Update_2$. 
For each output channel $y\in O$, (1) if $y\in O_1\setminus I_2$, then the set of output tasks ${\cal A}_y$ is ${\cal A}_y^1$; (2) if $y\in O_2 \setminus I_1$, then the set of output tasks ${\cal A}_y$ is ${\cal A}_y^2$; (3) if $y\in O_1 \cap I_2$, then for each task $A_i\in {\cal A}_y^i$ (i=1,2) the set of output tasks $A_y$ contains the task described by $Guard_1 \wedge Guard_2 \rightarrow Update_1;Update_2$; and (4)  if $y\in O_2 \cap I_1$, then for each task $A_i\in A_y^i$ (i=1,2), the set of output tasks ${\cal A}_y$ contains the task described by $Guard_2 \wedge Guard_1 \rightarrow Update_2;Update_1$;
The set of internal tasks is ${\cal A} = {\cal A}_1 \cup {\cal A}_2$.

\subsubsection{State machines}
However, to introduce more advanced matter like safety or liveness requirements, Rajeev Alur turns to state machines, very similar to the I/O-automata in this article. 

Both, SCRs (p.66) and APs (p.141) have naturally associated transitions systems by declaring the input and output variables to be local variables. With  a pair of states $(s, t)$ he identifies $s \stackrel{i/o}{\rightarrow} t$ as the reaction for some input $i$ and some output $o$ for SCRs - which is exactly the notation I used. For APs he adds the undefined value $\perp$ to the possible values of input and output variables in case they are not involved in an action - which is one of the essential ideas distinguishing nondeterministic from deterministic I/O transition systems.

So, from an automata perspective, the difference between SRCs and APs is mapped onto the deterministic versus nondeterministic issue - very similar to the approach of this article. However, what is lacking is recursion in the deterministic case and most importantly the role concept and therewith the distinction between the inner and outer coupling, the essence of the process model of this article.
 
\subsection{Approaches based on named actions}

There are many approaches to describe processes which are based on named actions. Among others, these are Robin Milner's calculus of communicating systems (CCS) \cite{Milner1989, Milner1992},  Charles A. R. Hoare's calculus of communicating sequential processes (CSP) \cite{Hoare1985} and Jan A. Bergstra's and Jan W. Klop's algebra of communicating processes (ACP) \cite{BergstraKlopp1987_Universal}. Also, Richard Mayr's \cite{Mayr1999} process rewrite systems can be subsumed her. For a recent overview, see the book of Jos C. M. Baeten, Twan Basten, and Michel Reniers \cite{BaetenBastenReniers2010_ProcessAlgebra}.

They all have in common the view of a process as a structure producing events by actions and - this is the main point - are based on providing names for these events as their 'alphabet' they operate on. So, the semantics of these calculi is provided by transition systems where each transition has a label, the name of the action causing the state change. An interaction becomes the simultaneous execution of actions.

Although, it is generally possible to find a mapping from the i/o-pairs of the IO-automata used in this article to some set of unique transition labels, the coupling mechanism thereby gets lost. As described in section \ref{ss_outer_coupling} the outer coupling was based on identical names of output and input characters: that the output of one transition is the input of another transition (Shannon channel). As described in (section \ref{ss_composition_of_interactions}), the basis for the inner coupling are the $\epsilon$ transitions, as these can be eliminated to represent the desired causal relation between the role-transitions.

As a result, an important difference between any approach with named actions and the presented approach with anonymous actions but named I/O-characters is their different support for describing composition behaviour. As the names of the actions are arbitrary, they do not help to express the fact that some transitions may refer to the same action, but only from a different perspective/projection.

\subsubsection{Christel Baier's and Joost-Pieter Katoen's system model}
A typical representative of the named action approach can be found in the very good textbook of Christel Baier and Joost-Pieter Katoen \cite{BaierKatoen2008} on model checking. 
To describe systems, they use the program graph concept whose semantics is defined by a labelled transition system. 

They define a program graph (PG) over a set of typed variables as a digraph whose edges are labelled with conditions of these variables and actions (pp. 32). Thus, a PG is not restricted to a prescription on how to calculate a computable function but it can also represent more generally a process in my terminology.

\begin{definition}
A program graph (PG) over a set $Var$ of typed variables is a tuple $(Loc, Act, Cond, \text{\it Effect}, \hookrightarrow, Loc_0, g_0)$ where
\begin{itemize}
\item $Loc$ is a finite set of locations, specifying which of the conditional transitions is possible.
\item $Cond$ is a set of conditions, mapping the evaluation of the variables to the set of Boolean values $\{true, false\}$.
\item $Act$ is a set of actions whose effect is formalised by the effect function.
\item $\text{\it Effect}: Act\times Eval(Var) \rightarrow Eval(Var)$ indicates how the evaluation $\eta$ of variables is changed by performing an action.
\item $\hookrightarrow \subseteq Loc\times Cond(Var) \times Act \times Loc$ is the conditional transition relation.
\item $Loc_0 \subseteq Loc$ is a set of initial locations.
\item $g_0 \in Cond(Var)$ is the initial condition. 
\end{itemize}
\end{definition}

The notation $l \stackrel{g:\alpha}{\hookrightarrow} l'$ is used as shorthand for $(l, g, a, l') \in \hookrightarrow$. The condition $g$ is also called thee ''guard'' of the conditional transition. If $g$ is tautological true, then the condition may be dropped.

Executing a program graph, we nondeterministically chose between all transitions $l \stackrel{g:\alpha}{\hookrightarrow} l'$ which satisfy $g$ in the current evaluation $\eta$ (i.e., $\eta \models g$). The execution of action $\alpha$ changes the evaluation of the variables $\eta$ according to $\text{\it Effect}(\alpha, .)$. Subsequently, the systems changes into location $l'$. If no such transition exists, the system stops.

They define the semantics of a program graph by means of a labelled transition system such that whenever $l \stackrel{g:\alpha}{\hookrightarrow} l'$ is a conditional transition in the PG, and the guard $g$ holds in the current evaluation $\eta$, then there is a transition from state $l, eta$ to state $(l', \text{\it Effect}(\alpha, \eta))$:

\begin{definition}
The transition system $TS(PG)$ of a program graph $PG = (Loc, Act, Cond, \text{\it Effect}, \hookrightarrow, Loc_0, g_0)$ over a set $Var$ of variables is the tuple $(S, Act, \rightarrow, I, AP, L)$ where
\begin{itemize}
\item $S = Loc \times Eval(Var)$ is the set of state values.
\item $\rightarrow \subseteq S \times Act \times S$ is defined by the following rule, formulated in structural operational semantics (SOS)-notation 

\[\frac{l \stackrel{g:\alpha}{\hookrightarrow} l' \hspace{0.2cm} \wedge \hspace{0.2cm} \eta \models g}{(l,\eta) \stackrel{\alpha}{\rightarrow} (l', \text{\it Effect}(\alpha, \eta))}\].
\item $I = \{(l, \eta) | l \in Loc_0, \eta \models g_0 \}$ is the set of possible initial state values.
\item $AP = Loc \cup Cond(Var)$  is the set of atomic propositions used later on to define system properties.
\item $L((l, \eta)) = \{l\} \cup \{g\in Cond(Var) | \eta \models g\}$.
\end{itemize}
 
\end{definition}

To describe the interactions between PGs, they first define their nondeterministic interleaving. 

\begin{definition} 
Let $PG_i = (Loc_i, Act_i, Cond_i, \text{\it Effect}_i, \hookrightarrow_i, Loc_{0,i}, g_{0,i})$, for $i=1,2$ be two program graphs over the variables $Var_i$. Program graph $P = PG_1|||PG_2$ over $Var_1 \cup Var_2$ is defined by
$P = (Loc_1 \times Loc_2, Act_1 \sqcup Act_2, Cond_1 \sqcup Cond_2, \text{\it Effect}, \hookrightarrow, Loc_{0,1} \times Loc_{0,2}, g_{0,1} \wedge g_{0,2})$ 

\noindent
where $\hookrightarrow$ is defined by the rules   
\[ 
\frac{l_1 \stackrel{g:\alpha}{\hookrightarrow_1} l_1'}{(l_1, l_2) \stackrel{g:(\alpha, 1)}{\hookrightarrow} (l_1', l_2)}
\hspace{0.2cm}\text{and}\hspace{0.2cm}
\frac{l_2 \stackrel{g:\alpha}{\hookrightarrow_2} l_2'}{(l_1, l_2) \stackrel{g:(\alpha, 2)}{\hookrightarrow} (l_1, l_2')}
\]
and $\text{\it Effect}$ is defined by:
\[ 
\text{\it Effect}(\alpha, \eta)(v) = \begin{cases}
  \text{\it Effect}_i\left((\alpha, i), \eta|_{Var_i}\right)(v) & \text{if}\, v\in Var_i\\
  \eta(v) & \text{otherwise}
\end{cases}
\]
where $\eta|_{Var_i}$ denotes the function $\eta$ restricted to $Var_i$.
\end{definition}
  
For sets $X$, $Y$, $X\sqcup Y$ is defined by $\{(x,1)|x\in X\} \cup \{((y,2))|y\in Y\}$. Variables can be classified as ''global'' or ''shared'' and ''local''. Both PGs share the variables $Var_1 \cap Var_2$, the other variables are local. Actions of the composed $PG$ accessing shared variables are denoted as ''critical'' as the resolution of a contention between these actions is an additional source of nondeterminism.

Then PGs are extended by communication actions $c!v$ and $c?x$ operating on a channel $c$ as FIFO-buffers with a given capacity $cap(c)\in \mathbb{N}\cup\{\inf\}$ where 

\begin{tabular}{ll}
$c!v$ & transmits the value $v$ along the channel $c$, and\\
$c?x$ & receives a message via channel $c$ and assigns it to variable $x$.
\end{tabular}

Is is assumed that the type of the variables are compatible with the values transported in the channel $dom(c)$. With $\text{\it Chan}$ as the set of channels and $\text{\it Var}$ as the set of variables, the set of communication actions is then given by

\[Comm = \left\{c!v, c?x \,|\, c\in \text{\it Chan}, v\in dom(c), x\in Var \,\text{with}\, dom(c) \subseteq dom(c) \right\}\]

The evaluation of the channels is provided by the function $\xi$. The PG's definition is extended by extending the action alphabet by the communication actions. 
The transition system semantics of a channel system is defined as follows:

\begin{definition}
Let $CS = [PG_1|\dots|PG_n]$ be a channel system  over $(\text{\it Chan}, Var)$ with $PG_i = (Loc_i, Act_i, \text{\it Effect}_i, \hookrightarrow_i, Loc_{0,i}, g_{0,i})$, for $0 < i\leq n.$
The transition system of $CS$, denoted as $TS(CS)$, is the tuple $(S, Act, \rightarrow, I, AP, L)$ where 
\begin{itemize}
\item $S = (Loc_1\times\dots\times Loc_n) \times Eval(\text{\it Var}) \times Eval(\text{\it Chan})$,
\item $Act = \bigcup_{0<i\leq n} Act_i \cup \{\tau\}$,
\item $\rightarrow$ is defined by the following rules
  \begin{enumerate}
  \item Interleaving for $\alpha \in Act_i$:

\[\frac{l_i \stackrel{g:\alpha}{\hookrightarrow} l_i' \hspace{0.2cm} \wedge \hspace{0.2cm} \eta \models g}{(l_1,\dots,l_i,\dots, l_n,\eta,\xi) \stackrel{\alpha}{\rightarrow} (l_1,\dots,l_i', \dots, l_n, \eta', \xi)}\]

 where $\eta' = \text{\it Effect}(\alpha, \eta)$. 
  \item Asynchronous message passing for $c\in \text{\it Chan}$, $cap(c) > 0$:

\noindent
$\bullet$ Receive a value along channel $c$ and assign it to variable $x$:

\[\frac{l_i \stackrel{g:c?x}{\hookrightarrow} l_i' \hspace{0.2cm} \wedge \hspace{0.2cm} \eta \models g \hspace{0.2cm} \wedge \hspace{0.2cm} len(\xi(c)) = k>0 \hspace{0.2cm} \wedge \hspace{0.2cm} \xi(c) = v_1\dots v_k}{(l_1,\dots,l_i,\dots, l_n,\eta,\xi) \stackrel{\tau}{\rightarrow} (l_1,\dots,l_i', \dots, l_n, \eta', \xi')}\]

 where $\eta' = \eta[x:=v_1]$ and $\xi' = \xi[c:=v_2\dots v_k]$. 

\noindent
$\bullet$ Transmit value $v\in dom(c)$ over channel $c$:

\[\frac{l_i \stackrel{g:c!v}{\hookrightarrow} l_i' \hspace{0.2cm} \wedge \hspace{0.2cm} \eta \models g \hspace{0.2cm} \wedge \hspace{0.2cm} len(\xi(c)) = k<cap(c) \hspace{0.2cm} \wedge \hspace{0.2cm} \xi(c) = v_1\dots v_k}{(l_1,\dots,l_i,\dots, l_n,\eta,\xi) \stackrel{\tau}{\rightarrow} (l_1,\dots,l_i', \dots, l_n, \eta, \xi')}\]

 where $\xi' = \xi[c:=v_1v_2\dots v_kv]$. 

  \item Synchronous message passing over $c\in \text{\it Chan}$, $cap(c) = 0$:

\[
\frac{l_i \stackrel{g_1:c?x}{\hookrightarrow} l_i' \hspace{0.2cm} \wedge \hspace{0.2cm} \eta \models g_1 \hspace{0.2cm} \wedge \hspace{0.2cm} \eta \models g_2 \hspace{0.2cm} \wedge \hspace{0.2cm} l_j \stackrel{g_2:c!v}{\hookrightarrow} l_j' \hspace{0.2cm} \wedge \hspace{0.2cm} i\neq j}
{(l_1,\dots,l_i,\dots, l_j, \dots, l_n,\eta,\xi) \stackrel{\tau}{\rightarrow} (l_1, \dots, l_i', \dots, l_j', \dots, l_n, \eta', \xi)}
\]

 where $\eta' = \eta[x:=v]$. 
  \end{enumerate}
\item $I=\left\{(l_1,\dots,l_n, \eta, \xi_0) \,|\, \forall 0<i\leq n. \,(l_i\in Loc_{0,i} \wedge \eta \models g_{0,i})\right\}$,
\item $AP = \bigcup_{0<i\leq n} Loc_i \cup Cond(\text{\it Var})$,
\item $L((l_1,\dots,l_n, \eta, \xi_0)) = \{l_1, \dots, l_n\} \cup \{g\in Cond(\text{\it Var}) \,|\, \eta \models g\}$.
\end{itemize}

\end{definition}
\subsubsection{Comparison to I/O-TSs with the alternating bit protocol (ABP) example}
To compare the approach of Christel Baier and Joost-Pieter Katoen with the my approach of I/O-TSs, I use their example 2.31 (pp. 54) of the alternating bit protocol (ABP) 

\begin{figure}[ht]
\begin{center}
\includegraphics[width=7cm]{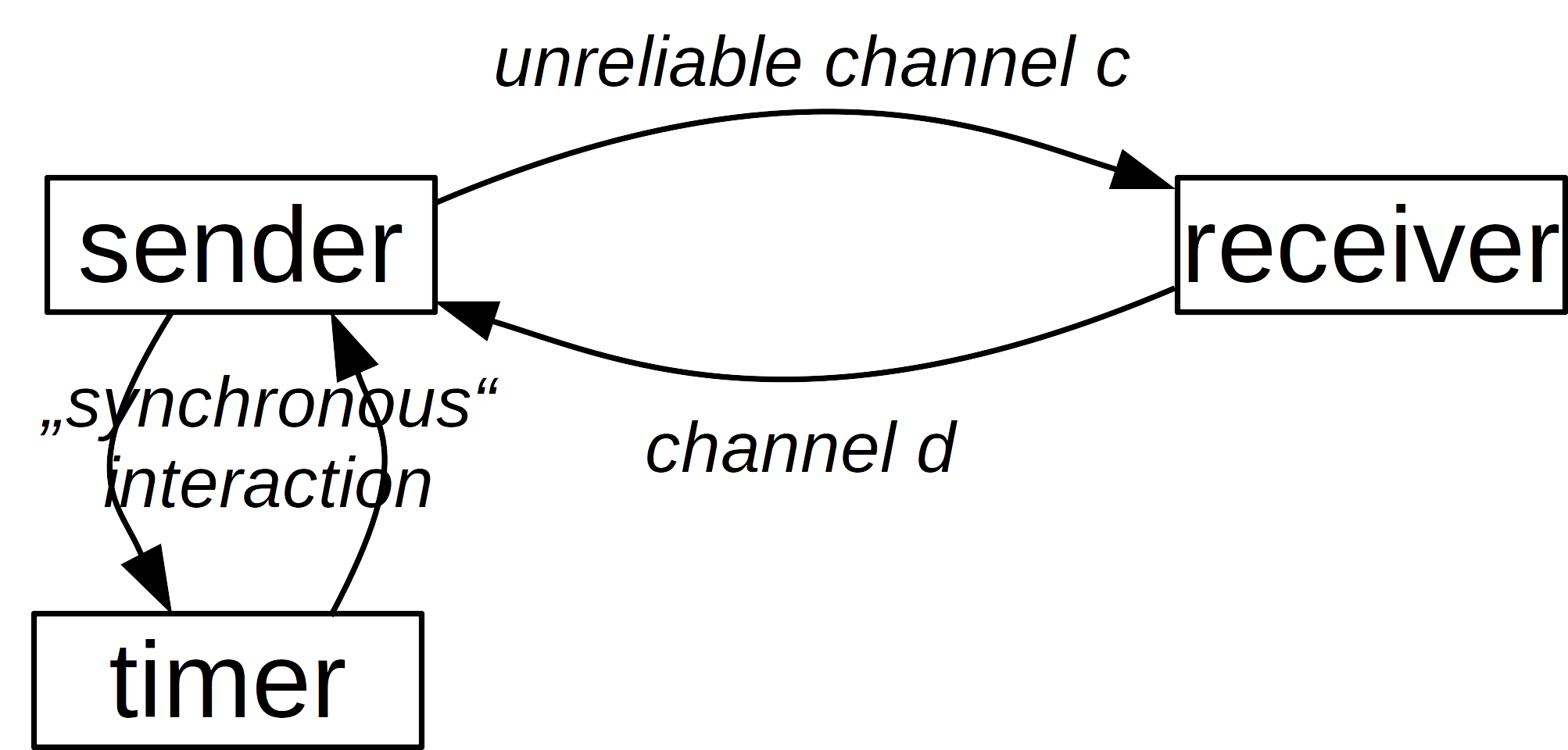}
\end{center}
\caption[]{\label{fig_alternating_bit_protocol_systems} The alternating bit protocol (ABP) specified by three program graphs (PGs) for sender, receiver and timer}
\end{figure}

The goal of the ABP is to transmit messages $m_0, m_1, \dots$ from a sender system ${\mathcal S}$ to a receiver system ${\mathcal R}$ over an unreliable channel $c$ and to assure that no message gets lost (see Fig. \ref{fig_alternating_bit_protocol_systems}). To make the example simpler, we assume that the back channel $d$ is reliable. To achieve this goal, sender ${\mathcal S}$ transmits messages one by one together with an alternating control bit and retransmits the message if the receiver ${\mathcal R}$ returns the wrong bit. If ${\mathcal S}$ has to wait too long for the acknowledgement, it timeouts and retransmits its pair of message and control bit $(m_i, b_i)$. Thus, the alternating control bit $b_i$ is used to distinguish retransmissions of $m_i$ from transmissions of subsequent and previous messages.

\begin{figure}[ht]
\begin{center}
\includegraphics[width=13cm]{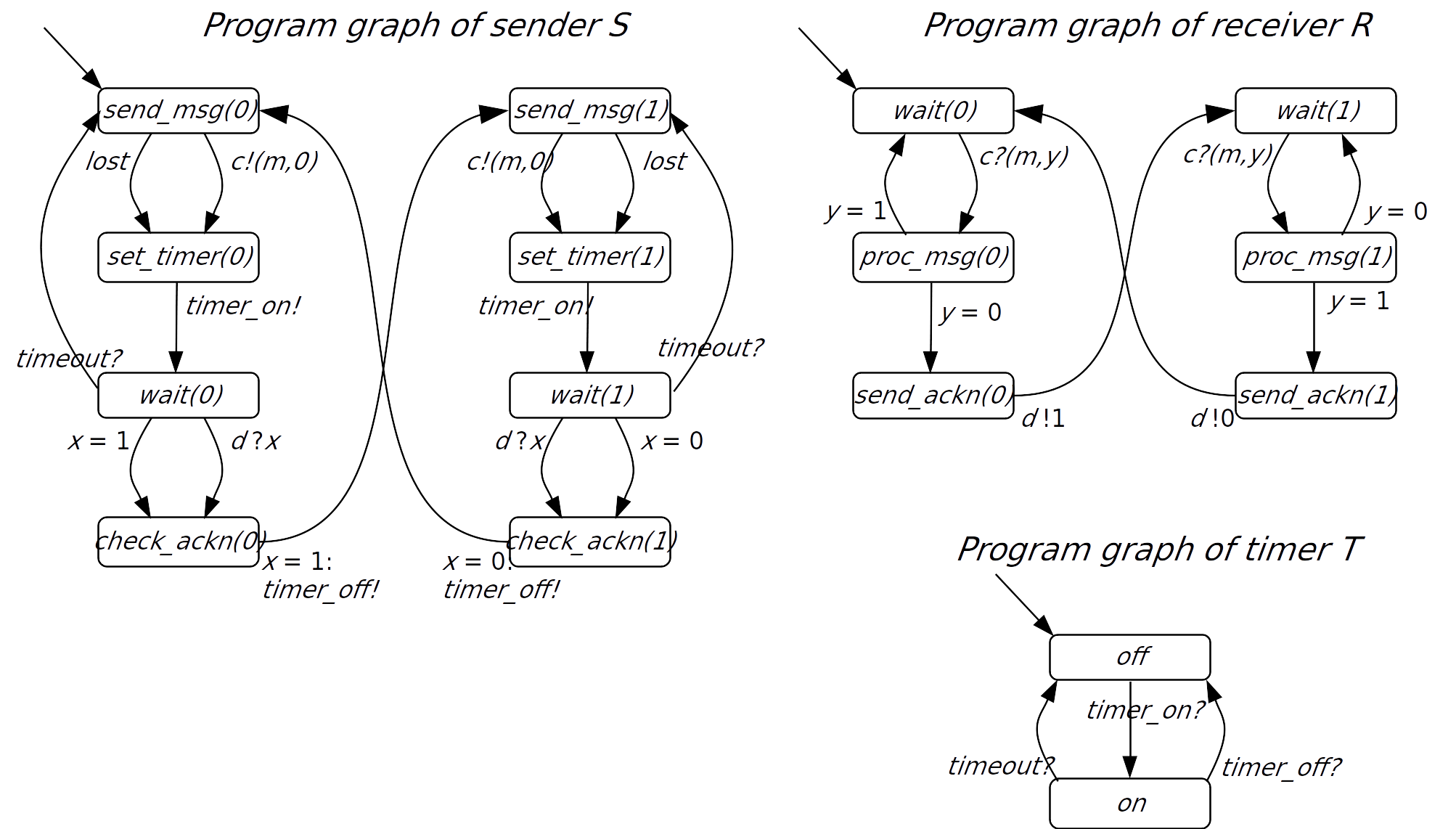}
\end{center}
\caption[]{\label{fig_alternating_bit_protocol_PGs} The alternating bit protocol (ABP) specified by three program graphs (PGs) for sender, receiver and timer}
\end{figure}

The timeout mechanism of ${\mathcal S}$ is modelled by a timer ${\mathcal T}$. ${\mathcal S}$ activates ${\mathcal T}$ with the $timer\_on$-message and stops it on receipt of an acknowledgement. When raising a timeout,  ${\mathcal T}$ signals to ${\mathcal S}$ that a retransmission should be initiated.

I show the program graphs of the protocol in Fig. \ref{fig_alternating_bit_protocol_PGs}. The number of states of the underlying transition system TS(ABP) depends on the capacity of channel $c$.

In Fig. \ref{fig_alternating_bit_protocol_IOTs} I model the ABP with its sender, receiver and timer with I/O-TSs. Now, the transitions are described anonymously with input and output characters. 

Obviously, both approaches can model the same classes of interactions. However, in my opinion, the main advantage of the IOTs is that a transition naturally represents the action of the system and the state values represent the system ''in between''. With PGs or the derived transitions systems, reception and sending of messages become extra ''actions'' and necessitates the introduction of extra variables. First, this leads to unnecessary many state values. And secondly and even more importantly, it blurs the interface of the systems with respect to their transformational I/O-behaviour.

\begin{figure}[ht]
\begin{center}
\includegraphics[width=13cm]{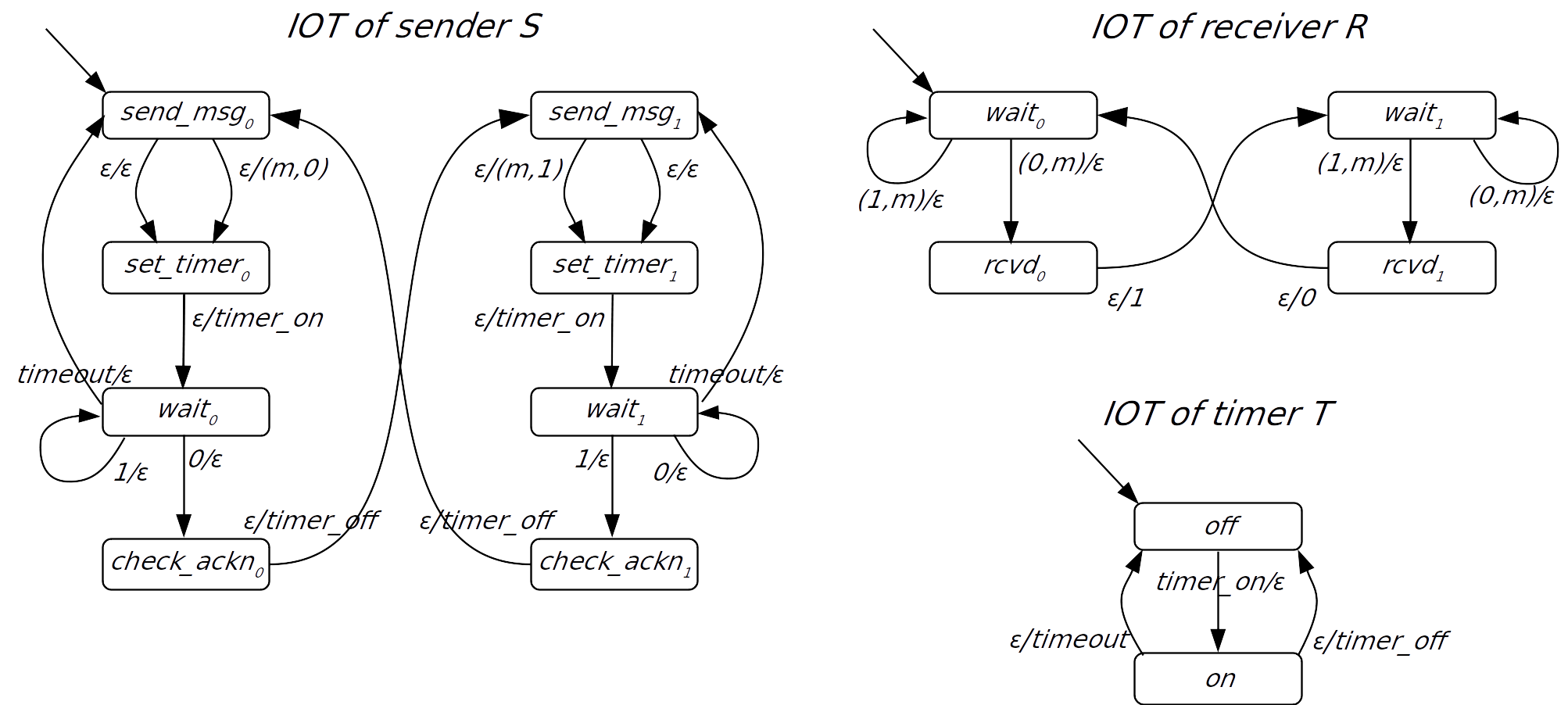}
\end{center}
\caption[]{\label{fig_alternating_bit_protocol_IOTs} The alternating bit protocol (ABP) specified by three I/O-transition system (IOTs) for sender, receiver and timer}
\end{figure}

With the I/O-TSs it is easy to see that they can be simplified by consolidating consecutive transitions without output with transition without input into one transition. This is shown in Fig. \ref{fig_alternating_bit_protocol_IOTs_simplified}

\begin{figure}[ht]
\begin{center}
\includegraphics[width=13cm]{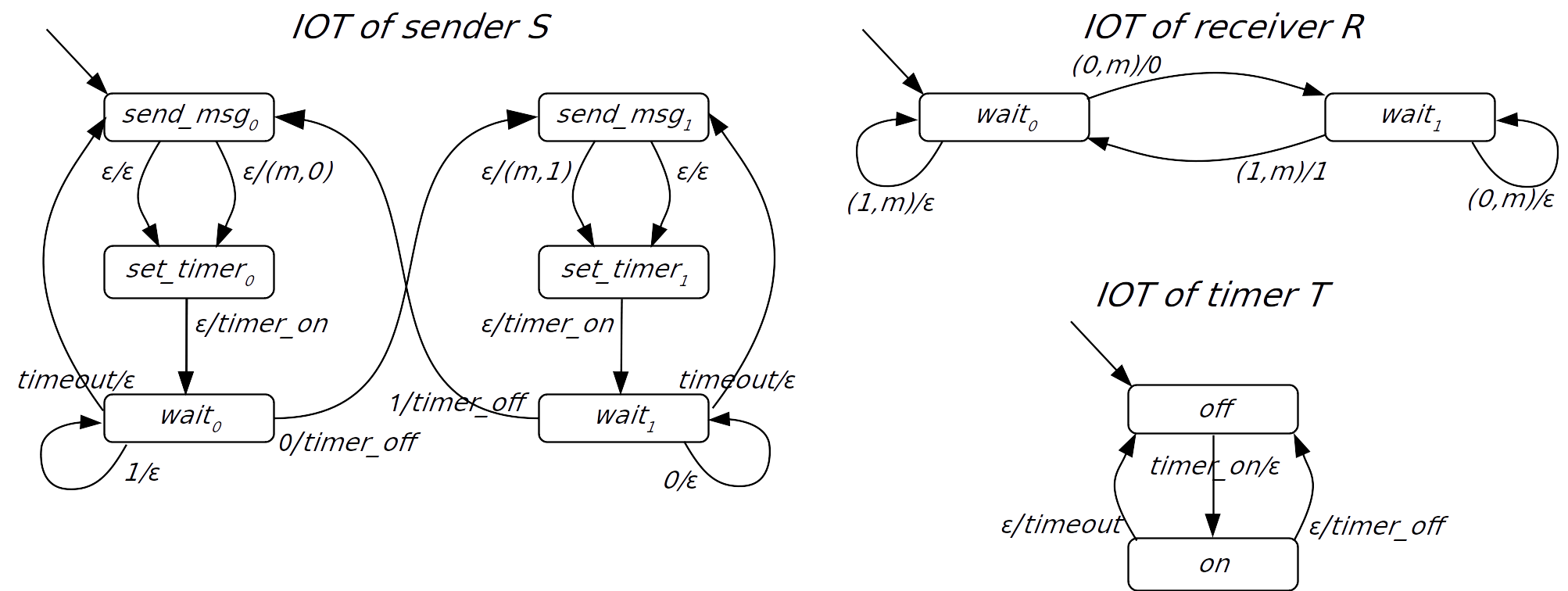}
\end{center}
\caption[]{\label{fig_alternating_bit_protocol_IOTs_simplified} The alternating bit protocol (ABP) specified by three simplified I/O-transition system (IOTs) for sender, receiver and timer}
\end{figure}

%
\section{IT system architecture} \label{s_it_system_architecture}
%
It is consensus in computer science to understand as the core of the architecture concept of an IT system its structure in terms of the compositional relationships of its interacting components. In this sense, structure is not left to the arbitrary consideration of humans, but is inherent in the system as a whole. In my understanding structure is actually the concept with which we can give ''wholeness'' an intelligible meaning \cite{Reich2001}.  
 \cite{Fowler2002_ApplArch,RechtinMaier2009,ISO42010-2011,Sommerville2009_SE}. 
Thus, the concept of composition of IT systems not only determines a meaningful framework for the discourse of their interoperability, but, by definition, it also determines the discourse about the architecture of IT systems.

\subsection{Components and Interfaces}
To quote Clemens Szyperski \cite{Szyperski2002} (p.36): ''Components are for composition''. Thus, components are intended as building blocks that easily fit together \cite{LauWang2007_Software}. But as already David Harel and Amir Pnueli \cite{HarelPnueli1985_Reactive} observed, the compositional behaviour of systems differs fundamentally, depending on whether they represent functions in their interaction or not.

Thus, it's the composition concept which necessitates the distinction between systems in general and components in particular and suggest a clear concept of an interface. How do we name the structure that comprises everything that we have to know from a system or its relevant parts to apply a composition operator? \cite{Tripakis2016_Compositionality,AlfaroHenzinger2001_InterfaceTheories} propose to use the term ''{\it interface}'' for this, a suggestion I happily endorse. A ''{\it component}'' can then be understood as a system that is intended for a particular composition and therefore has correspondingly well-defined interfaces that, by definition, express the intended composition.

The composition approach to components and interfaces is also necessary to fully understand an even simpler, but more phenomenological approach to classify component interactions correctly, as it was proposed by Tizian Schröder and me \cite{Reich2015_kvsi,ReichSchroeder2019_RefMod_IASem}. This approach is based on a classification of information transport between and information processing within interacting systems with respect to dimensions that are known to influence the gestalt of interfaces: information transport can occur uni- or bidirectional and information processing can be stateful/stateless, synchronous/asynchronous and deterministic/nondeterministic. 

I advocate to classify interfaces into unilateral and multilateral. Unilaterality denotes the property of an interface to be composable with arbitrary other such interfaces within an appropriate hierarchical composition context. This is typical for operations as we can use them in almost arbitrary superordinated contexts. 

In contrast, multilateral interfaces can compose only with their complementary partner interfaces. This is typical for a set of roles, belonging together and composing to a certain protocol, being the base for horizontal component relations. 

\subsubsection{Unilateral interfaces}
For simple systems the interface depends on whether we compose them as pipes or implicitly recursive. However, in both cases we use the transformational behaviour of the systems to create a superordinated super system. 

From a purely phenomenological point of view, we could be tempted to say that within a pipe, because of the unidirectional information flow, a simple system in a pipe has two interfaces, one upstream and one downstream. All what is known at both system borders is that the sender system's output alphabet must be a subset of the receiver system's input alphabet and that the receiver system is providing an operation on this value set. Then, the interface would be just the data type\footnote{A data type in my sense is an alphabet and a known set of elementary operations on this set \cite{Reich2018_data}. A good example of such a data type is the IEEE float datatype, which specifies, among other things, how 8 bytes must be processed using the elementary operations +, -, *, / so that we can regard them as floating point numbers.} of sender and receiver system and the system function does not appear explicitly. 

But to determine the pipe as a supersystem, that results form this form of composition, we obviously need to know the system functions of the pipe-subsystems. Thus the imagination of a pipe as a ''data flow'' is only one aspect of the pipe composition, neglecting the transformational aspects. This is the reason why, in the phenomenological model, we have to view the processing of each of the pipe's subcomponents as being deterministic, indicating that we have to know all their subsystems' functions to compose it.

Another case is implicitly recursive system composition as presented in section \ref{ss_implizit_rekursive_systeme} where the output of a simple system is feed back into the system which provided the input beforehand. I would like to illustrate what happens here with a simple example, shown in Fig. \ref{fig_system_vs_object_II}, of an MIS ${\mathcal S}_1$ and the two simple systems ${\mathcal S}_2$ and ${\mathcal S}_3$, which together form the (super)system ${\mathcal S}$ with the supersystem function $f_{\mathcal S}(x) = 2x + 5$.

\begin{figure}[ht]
\begin{center}
\includegraphics[width=10cm]{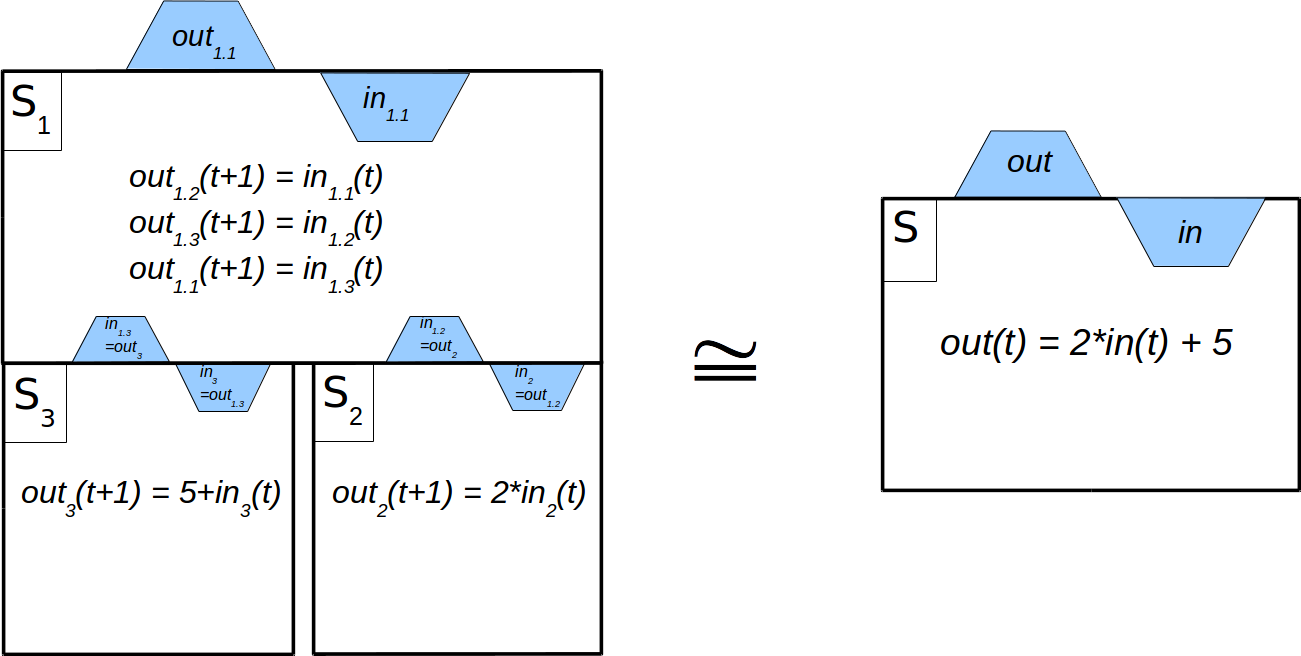}
\end{center}
\caption[]{\label{fig_system_vs_object_II} The MIS ${\mathcal S}_1$ and the two simple systems ${\mathcal S}_2$ and ${\mathcal S}_3$ compose to a supersystem ${\mathcal S}$ with $f_{\mathcal S}(x) = 2x + 5$. The input states of the supersystem at time t are mapped by the $f_{\mathcal S}$ to the output state at time $t_{\mathcal S}'=t_{\mathcal S}+1$.}
\end{figure}

Interestingly, two different sorts of hierarchy are created by this interaction. First, the deterministic interaction between the subsystems creates a hierarchy between ${\mathcal S}_1$ and the two ''dependent'' systems ${\mathcal S}_2$ and ${\mathcal S}_3$. The system ${\mathcal S}_1$ provides the input to systems ${\mathcal S}_2$ and ${\mathcal S}_3$ and thus determines their transitions entirely but not vice versa. 

But the interaction also creates the supersystem ${\mathcal S}$ with the system function $f_{\mathcal S}(x) = 2x + 5$ by composition, which makes ${\mathcal S}_1$, ${\mathcal S}_2$, and ${\mathcal S}_3$ to its subsystems. Please note that there is no interaction between the supersystem and its subsystem whatsoever, instead their relation is a ''is-part-of''-relation. And we attribute the operation with the mapping of the function $f_{\mathcal S}(x) = 2x + 5$ to the supersystem ${\mathcal S}$, and not to system  ${\mathcal S}_1$. These two different hierarchies between the same systems are shown in Fig. \ref{fig_system_vs_object_II_hierarchy}

\begin{figure}[ht]
\begin{center}
\includegraphics[width=10cm]{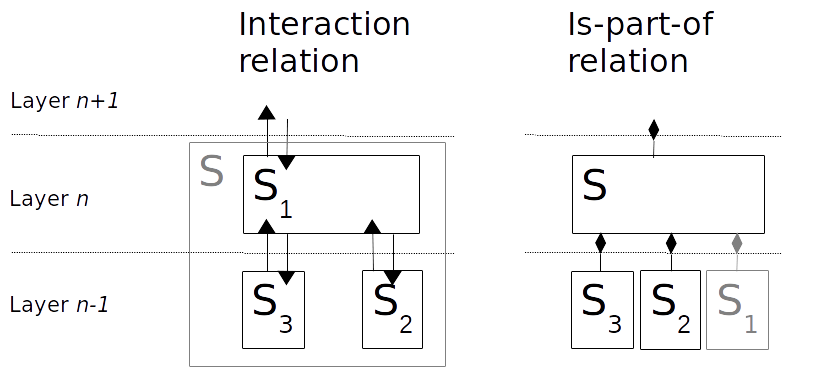}
\end{center}
\caption[]{\label{fig_system_vs_object_II_hierarchy} Due to their interaction, the systems of Fig. \ref{fig_system_vs_object_II} can be ordered in two different ways. On the left side they are ordered according to their asymmetric interaction relation. The arrows represent the information flow. The supersystem ${\mathcal S}$ is only shown in grey to show that it actually comprises all layers of its subsystems in this diagram. On the right side they are arranged according to their ''is-part of'' relation. Now it is the supersystem ${\mathcal S}$ that explicitly appears in a superordinate layer. With the ''is-part-of''-relation, there is no flow of information between the layers.}
\end{figure}

In the programming language C, this composition would look like:

\begin{verbatim}
int s2(int n) {return(2*n);}
int s3(int n) {return(n+5);}
int s(int n) {return(s3(s2(n)));}
\end{verbatim}

In this description, the system ${\mathcal S}_1$ is not mentioned explicitly which is why I have greyed it out in the right part of Fig. \ref{fig_system_vs_object_II_hierarchy}.

Objects in the object oriented sense are obviously systems intended for a hierarchical compositional behaviour, such that they always present the supersystem operation as their interface. The power of this concept is reflected in its dissemination as hierarchically composing component models are now ubiquitous. Their expressiveness lies in the fact that the interfaces of the component to be addressed always represent ''the whole system''. For many programming languages there are now package managers, each of which manages more than 100,000 packages, including Maven Central (Java), NuGet(.NET), Packagist (PHP), PyPI(Python), RubyGems(Ruby), etc.\cite{Cox2019}. 

\paragraph{Exceptions}
But then, what are exceptions, as we can declare them for example in the interface of a JAVA operation? Exceptions indicate in fact an undesired nondeterministic behaviour. For example, our operation to write something to our disk might work well, but only until the disk is full. Then this operation changes its behaviour, which means that its I/O-mapping is not uniquely determined by its input, but - undesirably - also by its internal state. 

We can decompose any nondeterministic transition relation into a deterministic part and a remainder. The deterministic part can then represent the desired function and the nondeterministic part represents the exceptions. With the try-catch-formalism of modern programming languages, exceptions obviously change the ''control flow'' of the operation, that is, its structure. The exception-mechanism allows us to foresee this situation and to define a substitute-function for the exceptional case.

\paragraph{Events}
Unilateral interfaces can be supplemented by events in the sense that changes of actually internal state values become externally observable. This makes sense if objects are used by several users and those that are not currently in charge of control want to be informed about important state changes, especially in connection with the state pattern (see section \ref{ss_objects}. For example it is a good idea to let my wife know if the car we use together is still roadworthy after I have driven it.

\subsubsection{Multilateral interfaces}
I called the interaction of interactive systems ''cooperation'' and their part which appears in the protocol composition a ''role''. Thus, ''roles'' are the interfaces for cooperative interactions. In general each role of a protocol is stateful, asynchronous and nondeterministic. As a protocol consists of all its roles, all roles are in a certain sense complementary. I therefore propose to speak of a multilateral interface in contrast to the unilateral interfaces of the hierarchical case, where the interface is agnostic against eventual other interfaces in the composition it is involved.

Taking into account the equivalence class construction for protocol roles of section \ref{ss_documents_and_main_states} it should be possible to describe all real world protocols in a standardised 5-tuple form.

\subsubsection{Components and recursion}
As we have seen, there are compositions that result in recursion. This raises the question, how does the interface of a system intended for recursion may look like. 

The operation in general represents an interface for implicit one step recursion. We have also seen that an operation with the special signature to represent a simple system with three appropriate input states and one output state, as illustrated in Fig. \ref{fig_system_composition_recursive}, can be composed together with a special bookkeeping system to create a recursive supersystem.

So, recursion does not extend our family of interface classes. In other words, we cannot indicate by the interface that a system is intended specially for a general recursion composition. As every programmer knows, we can easily create recursive program structures by mutual operation calls, for example like the following two faculty operations:

\begin{verbatim}
int faculty1(int n) {
  if (n == 1) return 1;
  else return n*faculty2(n-1);
}

int faculty2(int n) {
  if (n == 1) return 1;
  else return n*faculty1(n-1);
}
\end{verbatim}

If we take our credo seriously that components are intended as building blocks that {\it easily} fit together, we must avoid recursion on a component level. There is no special interface class to express a recursive composition and recursion makes computation complex. Thus, in fact, components should mark a computational complexity border, where recursion happens inside and the relation towards the outside is established by role-interfaces. 

\subsubsection{Summary}

\begin{table}[ht]
\begin{tabular}{p{1.3cm}|p{1cm}|p{1cm}|p{0.7cm}|p{3.5cm}|p{1.8cm}}
System class & Infor\-mation flow & Deter\-minism & State\-ful & Compositional behaviour & Interface class\\
\hline \hline
Simple systems & $\rightarrow$ & yes & no & Supersystem composition to pipes, possibly parallel & In- and out\-put as data types.\\
\hline
Simple systems & $\leftrightarrow$ & yes & no & Supersystem composition with implicit recursion & Operation represents system\\
\hline
Simple systems & $\leftrightarrow$ & yes & yes & Supersystem composition with implicit recursion & Object represents system\\
\hline
Recursive systems & * & * & * & Show same composition behaviour as simple systems if recursion calculation is viewed as internal. & *\\
\hline
MIS & $\leftrightarrow$ & yes & no & Supersystem composition with implicit recursion & Operation represents supersystem\\ 
\hline
MIS & $\leftrightarrow$ & yes & yes & Supersystem composition with implicit recursion & Object represents supersystem\\ 
\hline
MIS/IS & $\leftrightarrow$ & no & yes & No supersystem composition & In- and output as protocols
\end{tabular}
\caption[]{\label{tab_interface_classes}In this table, the different interface classes are summarised, based on the information flow to and from the systems and their transformational behaviour being deterministic vs. nondeterministic and stateful vs. stateless.}
\end{table}

In Tab. \ref{tab_interface_classes}, I summarise the given interface classification. It would be interesting to examine component models whether they adhere to the given interface classification and if lack of adherence leads to more complex component models. 

In their overview of component models, Ivica Crnkovic et al. \cite{CrnkovicSentilles2011_Classification} distinguish between ''operation-based'' and ''port-based'' interface support, demonstrating that many currently important component models actually do not support protocol declarations. 

\subsection{Component models in the literature}
According to Gerard J. Holzmann \cite{Holzmann2018_Software_Components}, the term software component was coined at the 1968 NATO Conference on Software Engineering in Garmisch by Doug McIlroy \cite{McIlroy1968_Software_Components}.  Tassio Vale et al. gives an overview of current trends in component-based software engineering \cite{ValeEtAl2016_28years_CBSE}.

The tight relation between the component and the interface concept makes it well suited for an evaluation of component models. Hierarchically composing component models are meanwhile ubiquitous. Package managers now exist for many programming languages, each managing more than 100,000 packages, such as Maven Central (Java), NuGet(.NET), Packagist (PHP), PyPI(Python), RubyGems(Ruby), etc. \cite{Cox2019}. 

But what about support for horizontal interactions? In their overview of component models, Ivica Crnkovic et al. distinguish. \cite{CrnkovicSentilles2011_Classification} distinguish between ''operation-based'' and ''port-based'' interface support, showing that many currently important component models do not in fact support protocol declarations.
 
\subsubsection{Components as distributed objects}
Distributed object models were developed under the idea that the encapsulation of the internal state by a beforehand defined set of operations in the sense of an abstract data type '' hides'' this state against the external world, providing some ''autonomy'' and thereby achieving a '' loose'' coupling between IT systems \cite{ChinChanson1991}. 

Well-known component models, which pursue this approach, are the Common Object Request Broker Architecture (CORBA), the Distributed Component Object Model (DCOM) or also the Open Platform Communications Unified Architecture (OPC-UA).

In fact, however, according to our definition only the compositional structure can be hidden behind an object-oriented interface, but not the logic of the respective mapping. Thus, from a logical perspective, remote objects become just a part of the ''one IT-system'' as local objects do and of ''loose coupling'' cannot be spoken. We simply cannot reach out of a system by a call of an operation.

In addition, the problem now arises that communication to remote objects is orders of magnitude more unreliable than communication to local objects, so that in particular interrupted state-changing operations too often lead to inconsistencies. State changes via operation semantics over unreliable communication is, as the problem of the Byzantine generals \cite{LamportShostakPease1982} vividly illustrates, generally not a good idea in more complex information processing contexts. 

Also, there is the problem of the social distance of the teams developing the various, possibly complex subcomponents, which can be well illustrated by the example of exception handling. The engineer of the superordinated component must worry about all exceptions of a subordinated component. That means that the engineer must know the regular operating condition of the subordinated component in detail for the design of its own component in anticipation, since she must align the composition thereupon. This is often a challenge, especially for strongly state-dependent remote components.

While for intentionally hierarchical relations between components these restrictions may still make sense in a compromising way due to other considerations, in the case of distributed objects the engineer faces an unsolvable dilemma in her effort to express horizontal relations: If she takes the object-oriented paradigm seriously, she must let objects call each other reciprocally --- which inevitably results in very complex, because recursive relations. Alternatively the engineers could agree on artificial, common objects, which would only have a de facto transport semantics, or, directly use the offered technology only for the definition of this transport semantics and do thus as far as possible without the support of the component model.   

\subsubsection{Service oriented architecture (SOA)}
The idea of service oriented architecture (SOA) goes back to R.W. Schulte and Y. V. Natis of the Gartner Group. \cite{SchulteNatis1996}. OASIS defines an SOA quite unspecifically as a ''paradigm for organising and utilising distributed capabilities that may be under the control of different ownership domains.'' \cite{OASIS2006_SOA_RM_1_DE}. A ''service'' is defined as ''The performance of work (a function) by one for another.'' and as a ''mechanism by which needs and capabilities are brought together''. A SOA is currently being propagated, for example, for Industry 4.0 \cite{DIN_SPEC_91345:2016-04} or in NATO \cite{NATO_ADatP-34_2020} (Vol. 2).

In fact, none of the service definitions address the transformation behaviour of a ''service'', such as whether it has to represent a function, i.e., should exhibit deterministic behaviour, or not. The WSDL 1.1 specification \cite{WSDL11_DE} defined four ''transmission primitives'' called ''operations''. WSDL 2.0 \cite{WSDL20_DE} in section 2.2.1 speaks of an ''interface component'' as a ''sequences of messages that a service sends and/or receives'' and of an ''operation'' as an ''interaction with the service consisting of a set of (ordinary and fault) messages exchanged between the service and the other parties involved in the interaction''. The BPMN version 2.0.2 \cite{BPMN202_DE}, which formally refers to the WSDL 2.0 specification, also states in section 8.5.3: ''An operation defines messages that are consumed and, optionally, produced when the operation is called.''

Hence, according to our composition concept, SOA interfaces are not well-defined. In my opinion, this semantic vagueness is the main reason why the concepts in the SOA environment have become quite complex overall, have been subject to frequent changes, and have not brought the hoped-for breakthrough in communication between enterprise applications despite the market power of the companies involved (see also already \cite{Reich2015_kvsi}) . 

The naming as ''service'' is in fact problematic because in the field of economics a ''service'' does not represent itself as a simple function, as the SOA with its WSDL interface syntax suggests. To pick up the example from \cite{Reich2017_InteractionSemantics}: To paint a wall in a newly built house, one has to solicit bids, accept a bid, arrange and, if necessary, rearrange appointments, check the result and, if accepted, pay the bill, and finally keep the documents for the tax return --- a relationship with the craftsmen at eye level, full of state, asynchrony, and nondeterminism. 

It is precisely for this reason that economics describes these interactions as games. In economics today, no one would be taken seriously who would come up with the idea of claiming that there are only two-person games, namely those consisting of service provider and service consumer, and that the service provider would, on top of that, always behave deterministically in the sense of a callable operation --- and all that because only the service consumer can pick up the phone. If one considers the economic concept of a game to be a valid model of economic interaction, then the informatic interface structures of these interactions must be very similar to these games, indeed they would have to be derivable from them --- which they are with the protocols 

\subsubsection{Representational State Transfer (REST)}
Representational State Transfer (REST) \cite{Fielding2000} can be seen as an attempt to apply the principles of stateless communication along with semantic agnosticism - both principles of the highly successful Hypertext Transfer Protocol (HTTP) - to network application interactions. Currently, it is often positioned as a simpler variant of SOA.

A REST call is said to satisfy the principles of addressability, that each resource must have a unique URI, and statelessness, that each REST message should contain all the information necessary for the processing it initiates.
Sometimes idempotence (e.g., \cite{Pautasso2014}) is also mentioned, that a REST call should always have the same effect regardless of its timing.

In some ways, I think REST relies on a misunderstanding of the role of state in distributed applications. For example, if I wanted to give a bank all the information it needs for further processing when I make a transfer, it wouldn't know my account balance. An interesting thought. Strictly speaking, I as a customer would have to quickly explain to them their entire business and their interaction with the other banks, etc., that is, the entire banking business --- nobody really wants that. If you transfer the ''principles'', which are very successful in information transfer, to information processing, you get information transfer.

With this understanding, REST is a methodological chimera. On the one hand, REST merely specifies a generic mailbox semantics that entrusts to the sender some of the functionality of the life cycle of the information sent as typed data. In other words, one can view all REST calls to a web address as a transport operation parameterised by the resource URI and data type. On the other hand, the principle of addressability requires that all resources be published, thus offering a lot of private information publicly.   

The actual transformation behaviour is explicitly not part of the semantics of a REST call, nor is any relationship between different REST calls. Accordingly, REST calls do not represent interfaces in the sense of all the information needed to compose the components with respect to their information processing.

Since the semantics of a REST call does not refer to the transformation behaviour of the receiving systems at all, REST --- like any transport function --- can be used very flexibly, which explains part of its success in the wake of Web services in the sense of SOA. 

However, in my view, this imprecision quickly becomes the downfall of engineers of IT systems when they try to build them out of comparatively self-contained, parallel, stateful so-called ''microservices'' \cite{LewisFowler2014_Microservices} without clearly describing their interfaces. Because then, without being able to know the exact protocols, the many possible, hard-to-find errors of concurrent information processing of semi-autonomous processes will inevitably set in --- or they actually develop (remote) objects with a functional interface, which they then have to orchestrate centrally, but which would then not be ''microservices'' in the sense described --- or one engineer does this and the other does that.

\subsubsection{Client server model}
The client server model is usually not understood as a component model. However, it is an interaction model and as such client and server act as two components. Thus the question arises, considering its practical relevance, what their relationship constitutes in the sense of a component model.

At the beginning of the 1990s the client server model was understood as request/reply scheme \cite{Andrews1991_Paradigms,Tanenbaum1991_MOS_1stEd}: The client sends the request and the server sends the response. With the emergence of SOA, it has been more strongly conceived as the division of a system function into services that are provided by different servers and can be used by a client \cite{Sommerville2009_SE,AlonsoKunoMachiraju2004}. 

The client-server model is also relevant in the context of database-based applications. There, engineers often talk of a 2-tier, 3-tier, or multi-tier architecture in terms of layering, where the database forms tier 1 and user interaction is located either in the second tier or, in conjunction with an intermediate application server layer, in the third tier.

In terms of our model of systems and their interactions, client and server first represent systems whose interaction can fall into any of the different interaction classes. The only significant difference between client and server is that between caller and called party. This also corresponds to the specification of the TCP/IP protocol (RFC 793, 7323), where the server waits at a TCP/IP port for calls from a client.

This criterion can indeed be used to declare an order. The only question is how meaningful it is. True, the division into callers vs. called does fit the semantic direction of a remote radio call. But it would be a complete misunderstanding, in my view, if we were to attribute the success and prevalence of the client-server model to this fit. 

Let's take a look at the so-called 3-tier architecture of modern enterprise applications, consisting of database, application server and user interface (UI) component. Its  scalability has been one of the technological reasons for SAP's business success: after the introduction of the R/3 system, which featured such an architecture, SAP's per capita revenue rose from about 150,000€ in 1992 to 250,000€ in 2000, while the number of employees increased from about 3,000 to 20,000 in the same period. 

This development can only be understood, in my view, by recognising that with a 3-tier architecture, there are actually 3 comparatively independent, stateful applications, each representing its domain of expertise, coordinating multiple non-deterministic, horizontal interactions. The database coordinates physics with the application server's requirements; the application server coordinates interactions to the database and to the UI component by means of business process logic; and the latter coordinates interactions to the application server's business logic and to the unpredictable user. That the interaction between the application server and the UI is not a semantic one-way street is experienced by every engineer who is faced with the problem of displaying relevant changes in the database, which originate from other users, on the current UI of a third party.

Properly understood, then, a 3-tier architecture is not a layered structure in our sense, but it is similar to a division of labour comparable to a medical team of a radiologist, an internist, and a surgeon. The latter has evolved  spontaneously over time to better manage the complexity of the respective domains of expertise. Therefore, a supposed ''simplification'' in the sense of a reduction of a so-called ''3-tier'' to a ''2-tier'' or even ''1-tier'' architecture will result in a complexification --- no one would even think of returning to the state of the fusion of internist and surgeon or radiologist in the sense of a ''simplification''. And just as doctors can cooperate better with each other if they better understand the subfield of their colleagues, the three components of database, application server and UI can also cooperate better, for example in the sense of more efficiently, if they replicate certain partial functionalities of the other components, for example by caching. I.e., loose coupling via protocols allows --- as in real life --- a much more efficient completion of the respective own core task.

\subsection{Interaction oriented architecture}
According to \cite{OASIS2006_SOA_RM_1}, a reference model is 

\begin{quote} an abstract framework for understanding significant relationships among the entities of some environment. It enables the development of specific reference or concrete architectures using consistent standards or specifications supporting that environment. A reference model consists of a minimal set of unifying concepts, axioms and relationships within a particular problem domain, and is independent of specific standards, technologies, implementations, or other concrete details.
\end{quote}

With the functional and the protocol composition, two different compositions exist, which we can express two-dimensionally in a graphical application model, as Fig. \ref{fig_layered_architecture} shows. The layers in the vertical direction represent the order by the ''is-part-of'' relation as created by the functional composition. The horizontal direction relates components of the same layer through protocols. Importantly, in this representation, information is {\it not} exchanged in the vertical direction. The system boundaries are determined by the interactions.

\begin{figure}[htbp] 
\begin{center}
  \includegraphics[width=12cm]{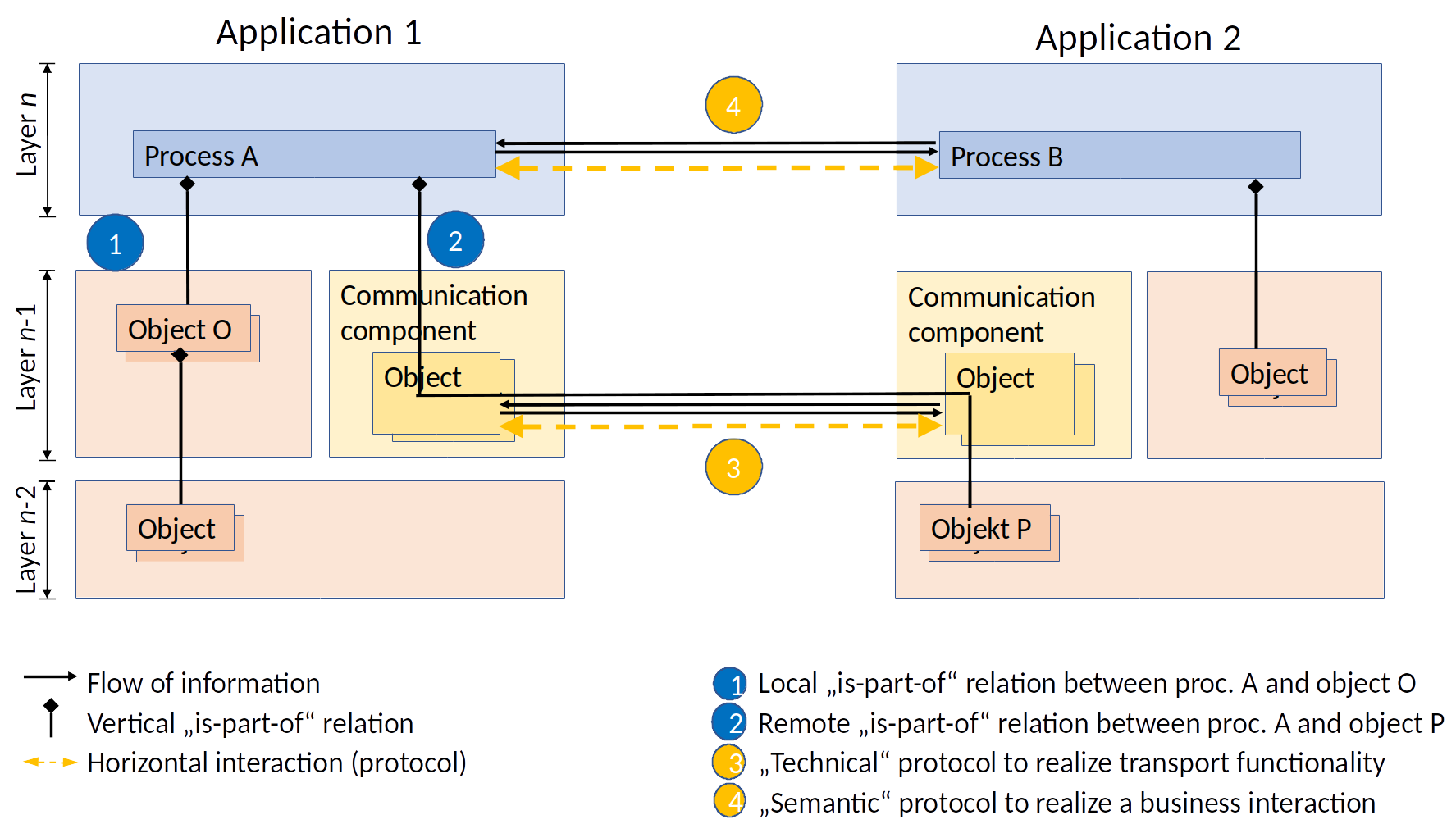} 
\end{center}
\caption[]{\label{fig_layered_architecture} A layered IT system architecture. The layers in the vertical direction are created through the ''is-part-of'' relation by functional composition, and the horizontal direction relates components interacting through protocols.}
\end{figure}

However, as our investigations has revealed, when building interactive IT systems, we are exposed to the tension of describing systems via functions on the one hand, but on the other hand only being able to integrate these systems into the desired interaction networks via non-deterministic, stateful interactions. If one ignores this tension and actually describes interactive systems with an explicitly formulated system function, then one unfortunately loses any guarantee that small changes in individual interactions will also result in only small changes in the structure of the application.

This leads to the obvious requirement to find a balance by an ''{\it interaction-oriented IT system architecture}'', leveraging the distinction between the inner and outer coupling of roles. We thereby reach a reference architecture of interactive systems consisting of the three elements of 
\begin{enumerate}
\item roles, 
\item coordination rules, and 
\item decisions. 
\end{enumerate}

One important aspect of this reference model is that it allows for an adequate description of human-machine interaction, as we can describe humans in exactly this way without risking to ignore their dignity.

\subsubsection{The semantic top layer in an interaction oriented IT system architecture}
First, it is clear that in a finite-size IT system with layered structure there must be a top layer, which I call ''{\it semantic top layer}''. In order not to be part of a corresponding supersystem, the system it contains must be an interactive system coordinating multiple nondeterministic, stateful, asynchronous interactions.

The top layer processes can delegate all their reusable functionality to dependent objects. This applies in particular to generally recursive functions. But please not, that there is no ''interaction'' between a top layer process and its dependent objects, but these objects are essential parts of these processes.

As consequence, the processes of the top layer are the place where the software engineer has to accumulate the non-reusable aspects of the system logic. Conversely, an unclear layer structure can significantly impair reuse, since then, software engineers can only poorly separate the reusable from the non-reusable. 

From our considerations, it immediately follows that protocols are an essential tool of ''programming in the large'' \cite{SinghChopraDesaiMallya2004_protocols}, since their roles demarcate complex, stateful IT systems from one another. 

As roles take part in two different compositions, their design goals do conflict. For their external composition into protocols, the main thing is that all participants agree on a common understanding. It is important to realise that the things about which the IT systems talk in their protocols often do not have the same meaning for all participants, but may have a complementary meaning. For example, the price of a traded good is for the buyer the amount of money he has to pay and for the seller the amount of money she receives. 

For their internal composition, it is important that the role provides all the information at the right time such that the desired coordination with all the other roles of the system can take place. As we have seen, it's the nondeterminism of the roles which plays a major role for enabling a flexible inner coordination.

\subsection{Reference architectures in the literature}
There exist many so called ''reference architectures'', most of them provide some model of layering. It is quite astonishing that, at least based on the ideas presented in this work, many of them actually lack a consistent criterion for their claimed layers. As a result, they do not serve their purpose properly because of the lengthy discussions it may take to work out their fallacies.

\subsubsection{The Open Systems Interconnection (OSI) model}
One of the most influential reference models is certainly the Open Systems Interconnection (OSI) model \cite{ISO_OSI_1994} which every informatics student learns in her first semester. It established the idea of a multi-layer software architecture. However, the assumption of the OSI model, 
\begin{quote}''OSI is concerned with the exchange of information between open systems (and not the internal functioning of each individual real open system).'' 
\end{quote}
is inconsistent. One cannot refrain from saying something about the structure of information processing while making assertions about its internal structure, such as layering. Also, the OSI model was not precise enough about the nature of the hierarchy. For example, the OSI assumption that information processing between components by means of protocols always takes place in the same layer proved to be false in the case of remote function calls. 

Thus, the system and interaction model of this article provides a formal justification for the intuition of the OSI model to consider software applications as layered. However, it also explains at the same time why the OSI model has found its way into reality only up to its 4th layer, since the management of a ''session'' state cannot be assigned to a dedicated layer in the general case.  Only in the case of vertical relation can the interaction-related state be hidden in a state of an intermediate ''session'' layer. In the case of horizontal interaction, the interaction-related state actually belongs to the components of the same semantic layer that interact with each other. 

Also, the kind of relation between the OSI-layers are not that of an interaction, but that of an ''is-part-of''-relation. 

The OSI model was certainly very influential, but nevertheless the alignment with the reality of the Internet protocol stack, with which it could match just up to the 4th layer, took a lot of time and its layers 5-7 could never establish themselves in their generality as intended as well. Now we can say why.

\subsubsection{The Level of Conceptual Interoperability Model (LCIM)}
Another example of a more frequently referenced reference model (e.g. lately in the IIC Connectivity Framework \cite{IIC2018_Connectivity_Framwork}) is the ''{\bf Level of Conceptual Interoperability Model (LCIM)}'' \cite{TolkTurnitsaDialloWinters2006}, which consists of the 7 alleged layers: no [interoperability], technical, syntactic, semantic [not defined in the present sense], pragmatic, dynamic, and conceptual interoperability. 

Obviously, it is not interaction that constitutes this hierarchy, but what else? Even for the technical transport of information, e.g. by the Internet protocol, a certain structure (=syntax) of the transported information is necessary. It is unclear how to separate semantic from pragmatic aspects. For example, how can the meaning of a bank transfer be described without referring to an action that a bank should perform? In my view, ''syntactic interoperability'' is most likely to mean mutual understanding at the level of data types, that is, that an incoming document can be mapped to internally used typed data structures. In my understanding, this is already ''semantic'' and comes into play in every protocol-based, horizontal interaction and is accordingly not limited to one level of interaction.   

\subsubsection{The Reference Architecture Model Industry 4.0 (RAMI4.0)}
A third example is the Reference Architecture Model Industry 4.0 (RAMI4.0) \cite{DIN_SPEC_91345:2016-04}. At the centre of this model is the ''asset'' as an entity that is used in the context of industrial production and has some value to an organisation. RAMI4.0 aims to structure the standardisation efforts of Plattform Industrie 4.0, the network led by the German government to further develop and implement its high-tech strategy in the field of industrial manufacturing, along the following three ''axes'':
\begin{enumerate}
\item Architecture: representation of the information relevant to the role of the asset in terms of the (ascending) ''layers'' Asset, Integration, Communication, Information, Functional and Business.
\item Life cycle: Representation of the life cycle of an asset separated by development of the ''type'' and production of the ''instances'' following IEC 62890.
\item Hierarchy: Assignment of functional models of product, field device, control unit, station, technical plant, company, networked world etc. following the standards DIN EN 62264-1 and DIN EN 61512-1.
\end{enumerate}

As noted in \cite{ReichZentarraLanger2020}, on the one hand, no clear ordering criterion can be identified for the architecture axis and, on the other hand, the hierarchy axis is based on an ''is-part-of'' relationship that is not conceptually independent of the architecture axis. Also, the categories of the hierarchy axis are not mutually exclusive --- a company could also be a product, for example. Finally, the different compositional mechanisms for interacting systems are not considered: A technical system that is part of a company is so in a different way than, say, a taillight is part of a car.

%
\subsection{The relevance of clear system boundaries for organising organisations} \label{s_relevance_of_clear_system_borders}
%

Well-defined boundaries of the systems being constructed are directly relevant to organise IT systems development. Melvin E. Conway \cite{Conway1968_Committees} had stated that the interaction structures of systems are co-determined by the interaction structures of the organisations constructing them, often referred to as ''Conway's Law''. There is indeed empirical evidence for this \cite{MacCormackBaldwinRusnak2012}. However, the model presented suggests the normativity of the inverse of this relationship, i.e., that the interaction structures of the systems being constructed should shape the interaction structures of the organisations constructing them \cite{HerbslebGrinter1999}. 

Thus, in an organisation that develops its IT-systems without well-defined layering, and thus without a semantic top layer, subdepartments tasked with developing reuse-enabled library functions will have unnecessary difficulties achieving their goals. Likewise companies, which rely on frameworks for the integration of complex components, not providing means for explicit protocol definitions, will falsely over- or underestimate possibly the integration expenditures for their components strongly by assigning the integration efforts arbitrarily -- possibly guided by internal power considerations -- instead of systematically to the respective components. Usually, this will also increase the additional communication effort between the responsible departments accordingly. 

In this respect, the definition of vertical interfaces should be a means of entrusting departments with differently abstract topics, and the definition of protocol-based interfaces should be a good instrument for structuring larger software-developing organisations, with which the explicit delimitation of expert domains can be achieved. A simple criterion here is the question of exceptions: do they exist and who takes care of them?


\bibliographystyle{alpha}
\bibliography{soziologie,philosophy,informatics}

\end{document}